\documentclass[journal]{IEEEtran}

\IEEEoverridecommandlockouts

\usepackage[noadjust]{cite}						
\usepackage[nolist,nohyperlinks,printonlyused]{acronym}

\usepackage{booktabs}							
\usepackage{tabularx}							

\usepackage{supertabular}
\usepackage{multirow}

\usepackage{nicefrac}							
\usepackage{siunitx}							
\usepackage{amsmath,amssymb,amsfonts}
\usepackage{gensymb}
\usepackage{mathtools}							

\usepackage{graphicx}							
\usepackage{color}								
\usepackage{xcolor}
\usepackage{caption}
\usepackage{subcaption}
\usepackage{tikz}
\usepackage{pgfplots}
\usepackage{grffile}
\usepackage[european, americanvoltages, nooldvoltagedirection]{circuitikz}		

\usepackage[german,USenglish,UKenglish]{babel}
\usepackage{textcomp}							
\usepackage[T1]{fontenc}
\usepackage{bm}									
\usepackage{dsfont}								
\usepackage{lmodern}							

\usepackage{xspace}
\usepackage{enumerate}							
\usepackage{cancel}								
\usepackage[normalem]{ulem}						
\usepackage{amsthm}								
\usepackage{thmtools}							
\usepackage{hyphenat}							

\usepackage{xifthen}							

\usetikzlibrary{external}
\usetikzlibrary{decorations}
\usetikzlibrary{shapes,arrows}
\usetikzlibrary{calc}
\usetikzlibrary{angles}
\usetikzlibrary{positioning}
\usetikzlibrary{fit}
\usetikzlibrary{backgrounds}
\usetikzlibrary{mindmap,trees}

\pgfplotsset{compat=newest}
\usetikzlibrary{plotmarks}
\usetikzlibrary{arrows.meta}
\usetikzlibrary{bending}
\usetikzlibrary{shadows}
\usepgfplotslibrary{patchplots}

\usepackage{arydshln}					
\usepackage[colorlinks=true,linkcolor=black,anchorcolor=black,citecolor=black,filecolor=black,menucolor=black,runcolor=black,urlcolor=black]{hyperref}

	\definecolor{purpleDark}{RGB}{118, 4, 205}
	\definecolor{purpleLight}{RGB}{186, 102, 250}
	
	\definecolor{blueDark}{RGB}{52, 78, 243}
	\definecolor{blueLight}{RGB}{118, 135, 244}
	\definecolor{redDark}{RGB}{197, 34, 0}
	\definecolor{redLight}{RGB}{255, 91, 57}
	\definecolor{yellowDark}{RGB}{255, 183, 0}
	\definecolor{yellowLight}{RGB}{255, 204, 77}
	\definecolor{greenDark}{RGB}{0, 143, 53}
	\definecolor{greenLight}{RGB}{42, 189, 97}
	
	\colorlet{greenFaint}{green!10!white}
	\colorlet{redFaint}{red!10!white}

	\definecolor{redText}{RGB}{222, 2, 10}
	\definecolor{orangeText}{RGB}{245, 86, 0}
	\definecolor{greenText}{RGB}{20,125,50}
	\definecolor{blueText}{RGB}{0, 114, 190}
	\definecolor{purpleText}{RGB}{115, 38, 146}
	\definecolor{pinkText}{RGB}{255, 107, 250}

	\definecolor{ballblue}{rgb}{0.13, 0.67, 0.8}
	\definecolor{buff}{rgb}{0.94, 0.86, 0.51}
	\definecolor{bronze}{rgb}{0.93,0.53,0.18}
	
	\definecolor{matlabCol1}{rgb}{0.0000,0.4470,0.7410}
	\definecolor{matlabCol2}{rgb}{0.8500,0.3250,0.0980}
	\definecolor{matlabCol3}{rgb}{0.9290,0.6940,0.1250}
	\definecolor{matlabCol4}{rgb}{0.4940,0.1840,0.5560}
	\definecolor{matlabCol5}{rgb}{0.4660,0.6740,0.1880}
	\definecolor{matlabCol6}{rgb}{0.3010,0.7450,0.9330}
	\definecolor{matlabCol7}{rgb}{0.6350,0.0780,0.1840}
	\definecolor{matlabCol8}{rgb}{0.0000,0.0000,0.0000}
	
	\newcommand{\autorefapp}[1]{\hyperref[#1]{Appendix~\ref*{#1}}}
	
	\addto\extrasUKenglish{%
	}

	\addto\extrasUSenglish{%
	}

	\makeatletter
	\newcommand\autorefMulti[1]{\@first@ref#1,@}
	\def\@throw@dot#1.#2@{#1}
	\def\@set@refname#1{
		\edef\@tmp{\getrefbykeydefault{#1}{anchor}{}}%
		\xdef\@tmp{\expandafter\@throw@dot\@tmp.@}%
		\ltx@IfUndefined{\@tmp autorefnameplural}%
			{\def\@refname{\@nameuse{\@tmp autorefname}s}}%
			{\def\@refname{\@nameuse{\@tmp autorefnameplural}}}%
	}
	\def\@first@ref#1,#2{%
		\ifx#2@\autoref{#1}\let\@nextref\@gobble
		\else%
			\@set@refname{#1}
			\@refname~\ref{#1}
			\let\@nextref\@next@ref
		\fi%
		\@nextref#2%
	}
	\def\@next@ref#1,#2{%
		\ifx#2@ and~\ref{#1}\let\@nextref\@gobble
		\else, \ref{#1}
		\fi%
		\@nextref#2%
	}
	\makeatother

	\newcommand{\Matlab}{{\rm \sc Matlab}\xspace}			
	\newcommand{\Simulink}{{\rm \sc Simulink}\xspace}		
	\newcommand{\Simscape}{{\rm \sc Simscape}\xspace}		
	
	\DeclareSIUnit{\pu}{pu}							
	\DeclareSIUnit{\VAR}{\volt\ampere{}R}			
	
	\newcommand{\addWithPreComma}[1]{%
		\if\relax #1\relax
		\else%
		,#1
		\fi%
	}
	\newcommand{\addWithPostComma}[1]{%
		\if\relax #1\relax
		\else%
		#1,
		\fi%
	}
	\newcommand{\addInParentheses}[1]{%
	\if\relax #1\relax
	\else%
		(#1)
	\fi%
	}

\newacro{PI}[PI]{proportional-integral}
\newacro{DDA}[DDA]{dynamic distributed averaging}
\newacro{DGU}[DGU]{distributed generation unit}

\newacro{EIP}[EIP]{equilibrium-independent passivity}
\newacro{IFP}[IFP]{input-feedforward passive}
\newacro{OFP}[OFP]{output-feedback passive}
\newacro{IFOFP}[IF-OFP]{input-feedforward output-feedback passive}

\newacro{ZSO}[ZSO]{zero-state observable}
\newacro{ZSD}[ZSD]{zero-state detectable}

\newtheorem{theorem}{Theorem}
\newtheorem{definition}[theorem]{Definition}
\newtheorem{proposition}[theorem]{Proposition}

\newtheorem{remark}{Remark}
\newtheorem{objective}{Objective}
\newtheorem{assumption}{Assumption}

	\renewcommand{\vec}[1]{\bm{#1}}						
	\renewcommand{\matrix}[1]{\bm{#1}}					
	
	\newcommand{\oSP}[1]{{#1}^*}						
	\newcommand{\oRef}[2]{{#1}_{\mathrm{Ref}\addWithPreComma{#2}}}	
	\newcommand{\oErr}[1]{\tilde{#1}}					
	\newcommand{\oEq}[1]{\hat{#1}}						
	
	\newcommand{\Reals}{\ensuremath{\mathbb{R}}}		
	\newcommand{\RealsPos}{\ensuremath{\mathbb{R}_{+}}}	
	
	\newcommand{\classC}[1]{\ensuremath{C^{#1}}}		
	
	\newcommand{\vOneCol}[1][]{\vec{\mathds{1}}_{#1}}	
	\newcommand{\Ident}[1][]{\matrix{I}_{#1}}			
	
	\newcommand{\Ltwo}{\ensuremath{L_{2}}}				

	\newcommand{\Transpose}{T}							
	
	\newcommand{\hi}[1]{\overline{#1}}					
	\newcommand{\lo}[1]{\underline{#1}}					
	
	\newcommand{\posDef}{\ensuremath{\succ}}			
	\newcommand{\negDef}{\ensuremath{\prec}}			
	\newcommand{\negSemiDef}{\ensuremath{\preccurlyeq}}	
	
	\DeclareMathOperator{\Diag}{Diag}					
	
	\newcommand{\dFull}[2]{\dfrac{\text{d} #1}{\text{d} #2}}					
	
	\newlength\mytemplena
	\newlength\mytemplenb
	\DeclareDocumentCommand\myalignalign{sm}%
	{%
		\settowidth{\mytemplena}{$\displaystyle #2$}%
		\setlength\mytemplenb{\widthof{$\displaystyle=$}/2}%
		\hskip-\mytemplena%
		\hskip\IfBooleanTF#1{-\mytemplenb}{+\mytemplenb}%
	}

	\newcommand{\denoteAct}{\alpha}
	\newcommand{\denoteUnact}{\beta}

	\newcommand{\denoteNLweight}{\mathrm{w}}
	\newcommand{\denoteCtrl}{\mathrm{c}}
	\newcommand{\denoteDGU}{\mathrm{d}}
	\newcommand{\denoteDDA}{\mathrm{a}}
	
	\newcommand{\denoteLoad}{\mathrm{L}}
	\newcommand{\denoteMG}{\mathrm{M}}

	\newcommand{\denoteTx}{\mathrm{t}}
	\newcommand{\denoteVSC}{\mathrm{s}}
	\newcommand{\denoteCrit}{\mathrm{crit}}
	
	\newcommand{\denoteComm}{\mathrm{C}}
	\newcommand{\denotePhys}{\mathrm{P}}
	
\newcommand{\sAct}[1][]{\alpha_{#1}}

\newcommand{\saNLweight}[1][]{a_{\denoteNLweight\addWithPreComma{#1}}}

\newcommand{\mA}[1][]{\matrix{A}_{#1}}
\newcommand{\mAT}[1][]{\matrix{A}_{#1}^\Transpose}

\newcommand{\mADGU}[1][]{\mA[\denoteDGU\addWithPreComma{#1}]}
\newcommand{\mADGUT}[1][]{\mAT[\denoteDGU\addWithPreComma{#1}]}

\newcommand{\sbNLweight}[1][]{b_{\denoteNLweight\addWithPreComma{#1}}}

\newcommand{\vb}[1][]{\vec{b}_{#1}}
\newcommand{\vbT}[1][]{\vec{b}_{#1}^\Transpose}

\newcommand{\vbDGU}[1][]{\vb[\denoteDGU\addWithPreComma{#1}]}
\newcommand{\vbDGUT}[1][]{\vbT[\denoteDGU\addWithPreComma{#1}]}

\newcommand{\scNLweight}[1][]{c_{\denoteNLweight\addWithPreComma{#1}}}

\newcommand{\sclo}[1][]{\lo{c}_{#1}}
\newcommand{\schi}[1][]{\hi{c}_{#1}}

\newcommand{\scLoadlo}[1][]{\sclo[\denoteLoad\addWithPreComma{#1}]}

\newcommand{\vc}[1][]{\vec{c}_{#1}}
\newcommand{\vcT}[1][]{\vec{c}_{#1}^\Transpose}

\newcommand{\vcDGU}[1][]{\vc[\denoteDGU\addWithPreComma{#1}]}
\newcommand{\vcDGUT}[1][]{\vcT[\denoteDGU\addWithPreComma{#1}]}

\newcommand{\sC}[1][]{C_{#1}}
\newcommand{\sCeq}[1][]{C_{\text{eq}\addWithPreComma{#1}}}

\newcommand{\sd}[1][]{d_{#1}}
\newcommand{\vd}[1][]{\vec{d}_{#1}}
\newcommand{\vdT}[1][]{\vec{d}_{#1}^\Transpose}

\newcommand{\mD}[1][]{\matrix{D}_{#1}}

\newcommand{\se}[1][]{e_{#1}}
\newcommand{\see}[1][]{\oErr{e}_{#1}}
\newcommand{\seeq}[1][]{\oEq{e}_{#1}}
\newcommand{\sedot}[1][]{\dot{e}_{#1}}
\newcommand{\seedot}[1][]{\dot{\oErr{e}}_{#1}}

\newcommand{\seDGU}[1][]{\se[\denoteDGU\addWithPreComma{#1}]}
\newcommand{\seDGUe}[1][]{\see[\denoteDGU\addWithPreComma{#1}]}
\newcommand{\seDGUeq}[1][]{\seeq[\denoteDGU\addWithPreComma{#1}]}
\newcommand{\seDGUdot}[1][]{\sedot[\denoteDGU\addWithPreComma{#1}]}
\newcommand{\seDGUedot}[1][]{\seedot[\denoteDGU\addWithPreComma{#1}]}

\newcommand{\vE}[1][]{\vec{e}_{#1}}
\newcommand{\vET}[1][]{\vec{e}_{#1}^\Transpose}

\newcommand{\vEP}[1][]{\vE[\denotePhys\addWithPreComma{#1}]}
\newcommand{\vEPT}[1][]{\vET[\denotePhys\addWithPreComma{#1}]}

\newcommand{\mE}[1][]{\matrix{E}_{#1}}
\newcommand{\mET}[1][]{\matrix{E}_{#1}^\Transpose}
\newcommand{\mEAct}[1][]{\mE[\denoteAct\addWithPreComma{#1}]}
\newcommand{\mEUnact}[1][]{\mE[\denoteUnact\addWithPreComma{#1}]}

\newcommand{\mEP}[1][]{\mE[\denotePhys\addWithPreComma{#1}]}

\newcommand{\graphE}[1][]{\mathcal{E}_{#1}}
\newcommand{\graphEComm}[1][]{\graphE[\denoteComm\addWithPreComma{#1}]}
\newcommand{\graphEPhys}[1][]{\graphE[\denotePhys\addWithPreComma{#1}]}

\newcommand{\vf}[1][]{\vec{f}_{#1}}

\newcommand{\siLtwo}[1][]{\gamma_{\Ltwo\addWithPreComma{#1}}}

\newcommand{\sgDDA}[1][]{\gamma_{\denoteDDA\addWithPreComma{#1}}}

\newcommand{\sg}[1][]{g_{#1}}

\newcommand{\sgNLweight}[1][]{\sg[\denoteNLweight\addWithPreComma{#1}]}

\newcommand{\graphG}[1][]{\mathcal{G}_{#1}}
\newcommand{\graphGComm}[1][]{\graphG[\denoteComm\addWithPreComma{#1}]}
\newcommand{\graphGPhys}[1][]{\graphG[\denotePhys\addWithPreComma{#1}]}

\newcommand{\sh}[1][]{h_{#1}}
\newcommand{\vh}[1][]{\vec{h}_{#1}}

\newcommand{\shNLweight}[1][]{\sh[\denoteNLweight\addWithPreComma{#1}]}

\newcommand{\she}[1][]{\oErr{h}_{#1}}

\newcommand{\mH}[1][]{\matrix{H}_{#1}}

\newcommand{\sii}[1][]{i_{#1}}
\newcommand{\sidot}[1][]{\dot{i}_{#1}}
\newcommand{\vi}[1][]{\vec{i}_{#1}}

\newcommand{\sie}[1][]{\oErr{i}_{#1}}
\newcommand{\siedot}[1][]{\dot{\oErr{i}}_{#1}}
\newcommand{\vie}[1][]{\oErr{\vec{i}}_{#1}}

\newcommand{\vieT}[1][]{\oErr{\vec{i}}_{#1}^\Transpose}
\newcommand{\sieq}[1][]{\oEq{i}_{#1}}
\newcommand{\vieq}[1][]{\oEq{\vec{i}}_{#1}}

\newcommand{\siTx}[1][]{\sii[\denoteTx\addWithPreComma{#1}]}
\newcommand{\siTxdot}[1][]{\sidot[\denoteTx\addWithPreComma{#1}]}
\newcommand{\viTx}[1][]{\vi[\denoteTx\addWithPreComma{#1}]}

\newcommand{\sieTx}[1][]{\sie[\denoteTx\addWithPreComma{#1}]}
\newcommand{\sieTxdot}[1][]{\siedot[\denoteTx\addWithPreComma{#1}]}
\newcommand{\vieTx}[1][]{\vie[\denoteTx\addWithPreComma{#1}]}

\newcommand{\vieTxT}[1][]{\vieT[\denoteTx\addWithPreComma{#1}]}

\newcommand{\sieqTx}[1][]{\sieq[\denoteTx\addWithPreComma{#1}]}
\newcommand{\vieqTx}[1][]{\vieq[\denoteTx\addWithPreComma{#1}]}

\newcommand{\sI}[1][]{I_{#1}}

\newcommand{\setIeq}[1][]{\oEq{\mathcal{I}}_{#1}}

\newcommand{\sIL}[1][]{I_{\denoteLoad\addWithPreComma{#1}}}
\newcommand{\sILv}[1][]{\sIL[#1](\sv[#1])}

\newcommand{\sILe}[1][]{\oErr{I}_{\denoteLoad\addWithPreComma{#1}}}

\newcommand{\sILeve}[1][]{\sILe[#1](\sve[#1])}
\newcommand{\sILveq}[1][]{\sIL[#1](\sveq[#1])}

\newcommand{\setJ}[1][]{\mathcal{J}_{#1}}

\newcommand{\sk}[1][]{k_{#1}}
\newcommand{\skp}[1][]{\sk[#1]^P}
\newcommand{\ski}[1][]{\sk[#1]^I}

\newcommand{\skpCtrl}[1][]{\skp[\denoteCtrl\addWithPreComma{#1}]}
\newcommand{\skiCtrl}[1][]{\ski[\denoteCtrl\addWithPreComma{#1}]}

\newcommand{\skpDDA}[1][]{\skp[\denoteDDA\addWithPreComma{#1}]}
\newcommand{\skiDDA}[1][]{\ski[\denoteDDA\addWithPreComma{#1}]}

\newcommand{\skpDGU}[1][]{\skp[\denoteDGU\addWithPreComma{#1}]}
\newcommand{\skiDGU}[1][]{\ski[\denoteDGU\addWithPreComma{#1}]}

\newcommand{\skLoop}[1][]{\kappa_{#1}}

\newcommand{\skLoopiCtrl}[1][]{\skLoop[\denoteCtrl\addWithPreComma{#1}]^I}

\newcommand{\sL}[1][]{L_{#1}}

\newcommand{\mLap}[1][]{\matrix{\mathcal{L}}_{#1}}
\newcommand{\mLapComm}[1][]{\mLap[\denoteComm\addWithPreComma{#1}]}
\newcommand{\mLapCommT}[1][]{\mLap[\denoteComm\addWithPreComma{#1}]^\Transpose}

\newcommand{\setM}[1][]{\mathcal{M}_{#1}}

\newcommand{\siIn}[1][]{\nu_{#1}}

\newcommand{\siInDGU}[1][]{\siIn[\denoteDGU\addWithPreComma{#1}]}
\newcommand{\siInLoad}[1][]{\siIn[\denoteLoad\addWithPreComma{#1}]}

\newcommand{\siInCtrl}[1][]{\siIn[\denoteCtrl\addWithPreComma{#1}]}
\newcommand{\siInNLweight}[1][]{\siIn[\denoteNLweight\addWithPreComma{#1}]}

\newcommand{\setN}[1][]{\mathcal{N}_{#1}}
\newcommand{\setNAct}[1][]{\setN[\denoteAct\addWithPreComma{#1}]}
\newcommand{\setNUnact}[1][]{\setN[\denoteUnact\addWithPreComma{#1}]}

\newcommand{\spp}[1][]{p_{#1}}
\newcommand{\spSP}[1][]{\oSP{p}_{#1}}

\newcommand{\speSP}[1][]{\oSP{\oErr{p}}_{#1}}
\newcommand{\speqSP}[1][]{\oSP{\oEq{p}}_{#1}}

\newcommand{\vpSP}[1][]{\oSP{\vec{p}}_{#1}}

\newcommand{\vpeSP}[1][]{\oSP{\oErr{\vec{p}}}_{#1}}
\newcommand{\vpeSPT}[1][]{{\oSP{\oErr{\vec{p}}}_{#1}{}}^\Transpose}
\newcommand{\vpeqSP}[1][]{\oSP{\oEq{\vec{p}}}_{#1}}

\newcommand{\speSPAct}[1][]{\speSP[\denoteAct\addWithPreComma{#1}]}
\newcommand{\vpeSPAct}[1][]{\vpeSP[\denoteAct\addWithPreComma{#1}]}
\newcommand{\vpeSPActT}[1][]{\vpeSPT[\denoteAct\addWithPreComma{#1}]}
\newcommand{\speSPUnact}[1][]{\speSP[\denoteUnact\addWithPreComma{#1}]}
\newcommand{\vpeSPUnact}[1][]{\vpeSP[\denoteUnact\addWithPreComma{#1}]}
\newcommand{\vpeSPUnactT}[1][]{\vpeSPT[\denoteUnact\addWithPreComma{#1}]}

\newcommand{\sP}[1][]{P_{#1}}

\newcommand{\mP}[1][]{\matrix{P}_{#1}}

\newcommand{\mPDGU}[1][]{\mP[\denoteDGU\addWithPreComma{#1}]}

\newcommand{\mQ}[1][]{\matrix{Q}_{#1}}
\newcommand{\sQ}[1][]{q_{#1}}

\newcommand{\mQDGU}[1][]{\mQ[\denoteDGU\addWithPreComma{#1}]}

\newcommand{\siOut}[1][]{\rho_{#1}}

\newcommand{\siOutDGU}[1][]{\siOut[\denoteDGU\addWithPreComma{#1}]}
\newcommand{\siOutLoad}[1][]{\siOut[\denoteLoad\addWithPreComma{#1}]}
\newcommand{\siOutTx}[1][]{\siOut[\denoteTx\addWithPreComma{#1}]}

\newcommand{\siOutDDA}[1][]{\siOut[\denoteDDA\addWithPreComma{#1}]}

\newcommand{\siOutCtrl}[1][]{\siOut[\denoteCtrl\addWithPreComma{#1}]}
\newcommand{\siOutNLweight}[1][]{\siOut[\denoteNLweight\addWithPreComma{#1}]}

\newcommand{\sR}[1][]{R_{#1}}

\newcommand{\sRDGUdamp}[1][]{\tilde{R}_{#1}}

\newcommand{\sS}[1][]{S_{#1}}
\newcommand{\sSdot}[1][]{\dot{S}_{#1}}

\newcommand{\sSDGU}[1][]{\sS[\denoteDGU\addWithPreComma{#1}]}
\newcommand{\sSDGUdot}[1][]{\sSdot[\denoteDGU\addWithPreComma{#1}]}
\newcommand{\sSLoad}[1][]{\sS[\denoteLoad\addWithPreComma{#1}]}
\newcommand{\sSLoaddot}[1][]{\sSdot[\denoteLoad\addWithPreComma{#1}]}
\newcommand{\sSTx}[1][]{\sS[\denoteTx\addWithPreComma{#1}]}
\newcommand{\sSTxdot}[1][]{\sSdot[\denoteTx\addWithPreComma{#1}]}
\newcommand{\sSMG}[1][]{\sS[\denoteMG\addWithPreComma{#1}]}
\newcommand{\sSMGdot}[1][]{\sSdot[\denoteMG\addWithPreComma{#1}]}

\newcommand{\sSCtrl}[1][]{\sS[\denoteCtrl\addWithPreComma{#1}]}
\newcommand{\sSCtrldot}[1][]{\sSdot[\denoteCtrl\addWithPreComma{#1}]}

\newcommand{\sSDDA}[1][]{\sS[\denoteDDA\addWithPreComma{#1}]}

\newcommand{\sSDDAdot}[1][]{\sSdot[\denoteDDA\addWithPreComma{#1}]}

\newcommand{\setS}[1][]{\mathcal{S}_{#1}}

\newcommand{\siCross}[1][]{\sigma_{#1}}

\newcommand{\siCrossCtrl}[1][]{\siCross[\denoteCtrl\addWithPreComma{#1}]}
\newcommand{\siCrossDGU}[1][]{\siCross[\denoteDGU\addWithPreComma{#1}]}
\newcommand{\siCrossNLweight}[1][]{\siCross[\denoteNLweight\addWithPreComma{#1}]}

\newcommand{\vdist}[1][]{\vec{\tau}_{#1}}

\newcommand{\vdistDGU}[1][]{\vdist[\denoteDGU\addWithPreComma{#1}]}

\newcommand{\su}[1][]{u_{#1}}
\newcommand{\vu}[1][]{\vec{u}_{#1}}
\newcommand{\vuT}[1][]{\vec{u}_{#1}^\Transpose}

\newcommand{\sue}[1][]{\oErr{u}_{#1}}
\newcommand{\sueq}[1][]{\oEq{u}_{#1}}
\newcommand{\vue}[1][]{\oErr{\vec{u}}_{#1}}

\newcommand{\vueq}[1][]{\oEq{\vec{u}}_{#1}}

\newcommand{\suNLweight}[1][]{\su[\denoteNLweight\addWithPreComma{#1}]}
\newcommand{\vuNLweight}[1][]{\vu[\denoteNLweight\addWithPreComma{#1}]}

\newcommand{\suCtrl}[1][]{\su[\denoteCtrl\addWithPreComma{#1}]}
\newcommand{\vuCtrl}[1][]{\vu[\denoteCtrl\addWithPreComma{#1}]}
\newcommand{\vuCtrle}[1][]{\vue[\denoteCtrl\addWithPreComma{#1}]}
\newcommand{\vuCtrlT}[1][]{\vuT[\denoteCtrl\addWithPreComma{#1}]}

\newcommand{\suDDA}[1][]{\su[\denoteDDA\addWithPreComma{#1}]}
\newcommand{\vuDDA}[1][]{\vu[\denoteDDA\addWithPreComma{#1}]}
\newcommand{\vuDDAe}[1][]{\vue[\denoteDDA\addWithPreComma{#1}]}
\newcommand{\vuDDAT}[1][]{\vuT[\denoteDDA\addWithPreComma{#1}]}

\newcommand{\setU}[1][]{\mathcal{U}_{#1}}

\newcommand{\sv}[1][]{v_{#1}}
\newcommand{\svdot}[1][]{\dot{v}_{#1}}
\newcommand{\vv}[1][]{\vec{v}_{#1}}

\newcommand{\sve}[1][]{\oErr{v}_{#1}}
\newcommand{\svedot}[1][]{\dot{\oErr{v}}_{#1}}
\newcommand{\vve}[1][]{\oErr{\vec{v}}_{#1}}

\newcommand{\vveT}[1][]{\oErr{\vec{v}}_{#1}^\Transpose}

\newcommand{\sveAct}[1][]{\sve[\denoteAct\addWithPreComma{#1}]}
\newcommand{\vveAct}[1][]{\vve[\denoteAct\addWithPreComma{#1}]}
\newcommand{\vveActT}[1][]{\vveT[\denoteAct\addWithPreComma{#1}]}
\newcommand{\sveUnact}[1][]{\sve[\denoteUnact\addWithPreComma{#1}]}
\newcommand{\vveUnact}[1][]{\vve[\denoteUnact\addWithPreComma{#1}]}
\newcommand{\vveUnactT}[1][]{\vveT[\denoteUnact\addWithPreComma{#1}]}

\newcommand{\sveq}[1][]{\oEq{v}_{#1}}
\newcommand{\vveq}[1][]{\oEq{\vec{v}}_{#1}}

\newcommand{\svRef}[1][]{\oRef{v}{#1}}
\newcommand{\svCrit}{\sv[\denoteCrit]}

\newcommand{\svVSC}[1][]{\sv[\denoteVSC\addWithPreComma{#1}]}

\newcommand{\setV}[1][]{\mathcal{V}_{#1}}

\newcommand{\sw}[1][]{w_{#1}}
\newcommand{\swlo}[1][]{\lo{w}_{#1}}
\newcommand{\swhi}[1][]{\hi{w}_{#1}}

\newcommand{\swCtrl}[1][]{\sw[\denoteCtrl\addWithPreComma{#1}]}
\newcommand{\swDDA}[1][]{\sw[\denoteDDA\addWithPreComma{#1}]}

\newcommand{\swDGU}[1][]{\sw[\denoteDGU\addWithPreComma{#1}]}
\newcommand{\swLoad}[1][]{\sw[\denoteLoad\addWithPreComma{#1}]}
\newcommand{\swTx}[1][]{\sw[\denoteTx\addWithPreComma{#1}]}
\newcommand{\swMG}[1][]{\sw[\denoteMG\addWithPreComma{#1}]}
\newcommand{\swMGdep}[1][]{\sw[\denoteMG,\denoteAct\denoteUnact\addWithPreComma{#1}]}
\newcommand{\swMGAct}[1][]{\sw[\denoteMG,\denoteAct\addWithPreComma{#1}]}
\newcommand{\swMGActlo}[1][]{\swlo[\denoteMG,\denoteAct\addWithPreComma{#1}]}
\newcommand{\swMGUnact}[1][]{\sw[\denoteMG,\denoteUnact\addWithPreComma{#1}]}
\newcommand{\swMGUnacthi}[1][]{\swhi[\denoteMG,\denoteUnact\addWithPreComma{#1}]}

\newcommand{\mW}[1][]{\matrix{W}_{#1}}

\newcommand{\sx}[1][]{x_{#1}}
\newcommand{\sxSP}[1][]{\oSP{x}_{#1}}
\newcommand{\sxdot}[1][]{\dot{x}_{#1}}
\newcommand{\vx}[1][]{\vec{x}_{#1}}
\newcommand{\vxT}[1][]{\vec{x}_{#1}^\Transpose}
\newcommand{\vxdot}[1][]{\dot{\vec{x}}_{#1}}

\newcommand{\sxe}[1][]{\oErr{x}_{#1}}
\newcommand{\sxeq}[1][]{\oEq{x}_{#1}}

\newcommand{\vxe}[1][]{\oErr{\vec{x}}_{#1}}

\newcommand{\vxeq}[1][]{\oEq{\vec{x}}_{#1}}

\newcommand{\vxeDGU}[1][]{\vxe[\denoteDGU\addWithPreComma{#1}]}

\newcommand{\sxCtrl}[1][]{\sx[\denoteCtrl\addWithPreComma{#1}]}
\newcommand{\sxCtrldot}[1][]{\sxdot[\denoteCtrl\addWithPreComma{#1}]}
\newcommand{\vxCtrl}[1][]{\vx[\denoteCtrl\addWithPreComma{#1}]}
\newcommand{\vxCtrlT}[1][]{\vxT[\denoteCtrl\addWithPreComma{#1}]}
\newcommand{\vxCtrldot}[1][]{\vxdot[\denoteCtrl\addWithPreComma{#1}]}

\newcommand{\vxDDA}[1][]{\vx[\denoteDDA\addWithPreComma{#1}]}

\newcommand{\vxDDAdot}[1][]{\vxdot[\denoteDDA\addWithPreComma{#1}]}

\newcommand{\setXeq}[1][]{\oEq{\mathcal{X}}_{#1}}

\newcommand{\sy}[1][]{y_{#1}}
\newcommand{\vy}[1][]{\vec{y}_{#1}}
\newcommand{\vyT}[1][]{\vec{y}_{#1}^\Transpose}

\newcommand{\sye}[1][]{\oErr{y}_{#1}}
\newcommand{\syeq}[1][]{\oEq{y}_{#1}}
\newcommand{\vye}[1][]{\oErr{\vec{y}}_{#1}}

\newcommand{\vyeq}[1][]{\oEq{\vec{y}}_{#1}}

\newcommand{\syNLweight}[1][]{\sy[\denoteNLweight\addWithPreComma{#1}]}

\newcommand{\syCtrl}[1][]{\sy[\denoteCtrl\addWithPreComma{#1}]}
\newcommand{\vyCtrl}[1][]{\vy[\denoteCtrl\addWithPreComma{#1}]}
\newcommand{\vyCtrle}[1][]{\vye[\denoteCtrl\addWithPreComma{#1}]}
\newcommand{\vyCtrlT}[1][]{\vyT[\denoteCtrl\addWithPreComma{#1}]}
\newcommand{\vyCtrlLoop}[1][]{\vy[\denoteCtrl\addWithPreComma{#1}]^\kappa}

\newcommand{\syDDA}[1][]{\sy[\denoteDDA\addWithPreComma{#1}]}
\newcommand{\vyDDA}[1][]{\vy[\denoteDDA\addWithPreComma{#1}]}
\newcommand{\vyDDAe}[1][]{\vye[\denoteDDA\addWithPreComma{#1}]}

\newcommand{\setY}[1][]{\mathcal{Y}_{#1}}

\newcommand{\vz}[1][]{\vec{z}_{#1}}

\newcommand{\vzdot}[1][]{\dot{\vec{z}}_{#1}}

\newcommand{\vzDDA}[1][]{\vz[\denoteDDA\addWithPreComma{#1}]}

\newcommand{\vzDDAdot}[1][]{\vzdot[\denoteDDA\addWithPreComma{#1}]}

\newcommand{\sZinv}[1][]{Z_{#1}^{-1}}

\newcommand{\sZCritinv}[1][]{\sZinv[\denoteCrit\addWithPreComma{#1}]}

\newcommand{\sdampCtrl}[1][]{\zeta_{\denoteCtrl\addWithPreComma{#1}}}

\graphicspath{{./03_Img/}}

\def\BibTeX{{\rm B\kern-.05em{\sc i\kern-.025em b}\kern-.08em
    T\kern-.1667em\lower.7ex\hbox{E}\kern-.125emX}}

\begin{document}
	
	\title{Passivity-based power sharing and voltage regulation in DC microgrids with unactuated buses}
	\author{Albertus Johannes Malan, Pol Jané-Soniera, Felix Strehle, and S{\"o}ren Hohmann
		\thanks{
		This work was supported in part by Germany’s Federal Ministry for Economic Affairs and Climate Action (BMWK) through the RegEnZell project (reference number 0350062C).
		(\emph{Corresponding author: A.\ J.\ Malan.})}
		\thanks{A.\ J.\ Malan, P.\ Jané-Soniera, F.\ Strehle, and S.\ Hohmann are with the Institute of Control Systems (IRS), Karlsruhe Institute of Technology (KIT), 76131, Karlsruhe, Germany. Emails: albertus.malan@kit.edu, pol.soneira@kit.edu, felix.strehle@kit.edu, soeren.hohmann@kit.edu.}
	}


	\maketitle
	
	\begin{abstract}
%
In this paper, we propose a novel four-stage distributed controller for a DC microgrid that achieves power sharing and average voltage regulation for the voltages at actuated and unactuated buses.
The controller is presented for a DC microgrid comprising multiple distributed generating units (DGUs) with time-varying actuation states; dynamic RLC lines; nonlinear constant impedance, current and power (ZIP) loads and a time-varying network topology. 
The controller comprising a nonlinear gain, PI controllers, and two dynamic distributed averaging stages is designed for asymptotic stability.
This constitutes first deriving passivity properties for the DC microgrid, along with each of the controller subsystems. Thereafter, design parameters are found through a passivity-based optimisation using the worst-case subsystem properties.
The resulting closed-loop is robust against DGU actuation changes, network topology changes, and microgrid parameter changes.
The stability and robustness of the proposed control is verified via simulations.

%
%
%

\end{abstract}

	
	\section{Introduction} \label{sec:Introduction}
%
\IEEEPARstart{T}{he advent} of localised power generation and storage increasingly challenges the prevailing centralised power-generation structures. Originally proposed in \cite{Lasseter2001}, the \emph{microgrids} paradigm envisions networks that can operate autonomously through advanced control while meeting consumer requirements. Although current electrical grids predominantly use AC, high and low voltage DC networks have been made technically feasible due to the continual improvements of power electronics. Indeed, DC microgrids exhibit significant advantages over their AC counterparts, demonstrating a higher efficiency and power quality while simultaneously being simpler to regulate \cite{Justo2013, Meng2017}.

In microgrids, power generation and storage units are typically grouped into \acp{DGU} which connect to the microgrid through a single DC-DC converter for higher efficiency \cite{Justo2013}. This changes the traditionally centralised regulation problem in power grids into a problem of coordinating the \acs{DGU} connected throughout the microgrid. This coordination is generally realised as \emph{average} or \emph{global voltage regulation} in combination with \emph{load sharing} between the \acp{DGU} (see e.g.\ \cite{Nasirian2015, Tucci2018, Zhao2015}).

\paragraph*{Literature Review}
%
A vast number of approaches have been proposed for the voltage regulation and load sharing of DC microgrids, as detailed in the overview papers \cite{Dragicevic2016, Meng2017, Kumar2019} along with the sources therein. These approaches are broadly categorised as either centralised, decentralised or distributed in nature \cite{Dragicevic2016, Meng2017, Kumar2019}. While centralised controllers can optimally coordinate the \acp{DGU}, they offer reduced scalability and flexibility and have a single point of failure \cite{Kumar2019}. On the other hand, decentralised controllers either only attempt to achieve voltage stability \cite{Tucci2016, Strehle2020DC, Cucuzzella2023} or achieve load sharing at the cost of voltage regulation quality (e.g.\ the droop-based approaches in \cite{Meng2017}).

In response to these limitations, numerous controllers for voltage regulation and load sharing which operate in a distributed manner have been proposed \cite{Nasirian2015, Tucci2018, Trip2019, Cucuzzella2019cons, Sadabadi2022, Han2018, Nahata2020, Nahata2022, Zhao2015, DePersis2017, Fan2019, Cucuzzella2019pbc}. In \cite{Nasirian2015}, distributed averaging is employed to find a global voltage estimate with which voltage regulation is achieved, but the microgrid dynamics are neglected in the stability analysis. Distributed averaging with dynamic microgrid models is used in \cite{Tucci2018, Trip2019}, although \cite{Tucci2018} requires LMIs to be solved before buses are allowed to connect whereas \cite{Trip2019} only considers constant current loads. Similarly, a sliding-mode controller is proposed in \cite{Cucuzzella2019cons} for a dynamic microgrid with constant current loads. On the other hand, \cite{Sadabadi2022} proposes a cyberattack-resilient controller for a microgrid with constant conductance loads and resistive lines. A consensus-based distributed controller with event-triggered communication is presented in \cite{Han2018}. Consensus-based controllers are also utilised in \cite{Nahata2020, Nahata2022, Zhao2015}, where \cite{Zhao2015} uses a consensus-based integral layer on top of a droop-based controller. Finally, while many contributions strive to achieve proportional current sharing \cite{Nasirian2015, Tucci2018, Trip2019, Cucuzzella2019cons, Sadabadi2022, Han2018, Nahata2020, Nahata2022, Zhao2015, Cucuzzella2019pbc}, nonlinear controllers that achieve proportional power sharing have also been proposed in \cite{DePersis2017,Fan2019}.

While the literature listed above differ greatly in their approaches, we note a commonality in their omission of buses without actuation. This omission is typically motivated either by considering a microgrid comprising only \emph{actuated} \ac{DGU} buses \cite{Nasirian2015, Tucci2018, Nahata2020, Nahata2022}, or by eliminating the \emph{unactuated} buses with the Kron-reduction \cite{Trip2019, Cucuzzella2019cons, Sadabadi2022, Han2018, Zhao2015, DePersis2017, Fan2019, Cucuzzella2019pbc}. However, considering a network comprising only actuated buses severely limits the flexibility of a microgrid, since each bus must be able to supply or consume enough power at all times. On the other hand, the Kron-reduction requires loads to be described as positive conductances (see e.g.\ \cite{Dorfler2013}). While research into Kron-reduced networks with negative loads is ongoing (see e.g.\ \cite{Chen2021}), the general inclusion of negative loads, e.g.\ non-controllable power sources, in Kron-reducible networks remains out of reach at present. Furthermore, consider the case where a \ac{DGU} can no longer supply or consume the required amount of power, e.g.\ a fully charged or discharged battery storage. Such a \ac{DGU} then loses the ability to regulate itself and fully support the grid. In the approaches considered above \cite{Nasirian2015, Tucci2018, Trip2019, Cucuzzella2019cons, Sadabadi2022, Han2018, Nahata2020, Nahata2022, Zhao2015, DePersis2017, Fan2019, Cucuzzella2019pbc}, such a \ac{DGU} is forced to disconnect from the microgrid and its local measurements are discarded. For \acp{DGU} with intermittent power sources, this could result in significant swings in the number of controlled and observed buses in the microgrid.

\paragraph*{Main Contribution}
In this paper, we consider a DC microgrid as a physically interconnected multi-agent system. Extending our work in \cite{Malan2022a}\footnote{The controller proposed in \cite{Malan2022a} is extended by weighing the error with a nonlinear function. Moreover, in addition to applying the controller to a DC microgrid context, we here propose a new dissipativity-based analysis that investigates the closed loop stability analytically as opposed to the numerical results in \cite{Malan2022a}.}, we propose a four-stage controller that achieves voltage regulation and power sharing in a DC microgrid with actuated and unactuated buses in a distributed manner. The four-stage controller comprises a nonlinear weighting function, two \ac{DDA} stages and a \ac{PI} controller. The asymptotic stability of the closed loop comprising the DC microgrid and the four-stage controller interconnected in feedback is proven by means of passivity theory. In detail, the contributions comprise:
\begin{enumerate}
	\item A four-stage distributed controller for DC microgrids which achieves \emph{consensus} on the weighted average voltage error of actuated and unactuated buses and assures \emph{coordination} through power sharing at the actuated buses.
	\item A nonlinear weighting function that penalises voltage errors outside a given tolerance band more strongly than those within.
	\item Passivity classifications for each of the constitutive microgrid subsystems (\acp{DGU}, loads, and lines) and for each of the controller stages (weighting function, \ac{DDA}, and PI).
	\item A method for calculating the \ac{IFOFP} indices of the nonlinear power-controlled \acp{DGU} through optimisation.
	\item An \ac{IFOFP} formulation for the DC microgrid with a supply rate that is independent of the network topology, the number of buses and their states of actuation.
	\item A passivity-based stability analysis for the equilibrium of the DC microgrid connected in feedback with the four-stage controller. 
\end{enumerate}
In addition to the contributions listed above, we also contribute a theoretical result comprising a formalisation of the obstacle presented by cascaded \ac{IFP} and \ac{OFP} systems in the analysis of dissipative systems. This theoretical contribution informs and motivates parameter choices for the four-stage controller in Contribution 1.

We highlight that the proposed controller can achieve exact voltage regulation and power sharing with the stability verified with the eigenvalues of the linearised system. Moreover, by employing leaky PI controllers, we demonstrate a passivity-based stability analysis that is independent of and robust against changes in the communication topology, changes in the electrical topology, load changes, changes in the actuation status of \acp{DGU}, uncertainties in component parameters, and buses connecting or disconnecting.

\paragraph*{Paper Organisation}
The introduction concludes with some notation and preliminaries on graph theory. In \autoref{sec:Prelim}, we recall and introduce results relating to dissipativity theory. Next, in \autoref{sec:Problem}, the problem is modelled and objectives for the steady state are formalised. In \autoref{sec:Control}, a four-stage control structure is introduced that fulfils objectives from \autoref{sec:Problem}. Thereafter, the passivity properties of the constituent subsystems are investigated in \autoref{sec:Passivity} and the controller is designed for asymptotic stability of the closed loop in \autoref{sec:Stability}. Finally, in \autoref{sec:Simulation}, a simulation is used to verify the asymptotic stability and robustness of the closed loop. Concluding remarks are provided in \autoref{sec:Conclusion}.

\paragraph*{Notation and Preliminaries}
Define as a vector $\vec{a} = (a_k)$ and a matrix $\matrix{A} = (a_{kl})$. $\vOneCol[k]$ is a $k$-dimensional vector of ones and $\Ident[k]$ is the identity matrix of dimension $k$.
$\Diag[\cdot]$ creates a (block-)diagonal matrix from the supplied vectors (or matrices).
The upper and lower limits of a value $a$ are given by $\hi{a}$ and $\lo{a}$.
For a variable $\sx$, we denote its unknown steady state as $\sxeq$, its error state as $\sxe \coloneqq \sx - \sxeq$, and a desired setpoint as $\sxSP$.
Whenever clear from context, we omit the time dependence of variables.

We denote by $\graphG = (\setN, \graphE)$ a finite, weighted, undirected graph with vertices $\setN$ and edges $\graphE \subseteq \setN \times \setN$. Let $|\setN|$ be the cardinality of the set $\setN$. Let $\mLap$ be the \emph{Laplacian matrix} of $\graphG$.
By arbitrarily assigning directions to each edge in $\graphE$, the \emph{incidence matrix} $\mE \in \Reals^{|\setN|\times|\graphE|}$ of $\graphG$ is defined by
\begin{equation} \label{eq:Intro:Incidence}
	\se[kl] = \left\{\begin{array}{@{}rl}
		+1 & \text{if vertex $k$ is the sink of edge $l$}, \\
		-1 & \text{if vertex $k$ is the source of edge $l$}, \\
		0 & \text{otherwise}.
	\end{array}\right.
\end{equation}
	\section{Dissipativity Preliminaries} \label{sec:Prelim}
We here recall and introduce preliminaries of dissipativity theory for nonlinear systems. In \autoref{sec:Prelim:Diss_system} we provide definitions relating to dissipativity and passivity theory. Thereafter in \autoref{sec:Prelim:Diss_static}, we investigate the passivity properties of static functions. Finally, in \autoref{sec:Prelim:Diss_interconnect}, we recall a result on the interconnection of dissipative systems with quadratic supply rates and formalise a new result on the limitations of such an interconnection.
\subsection{Dissipative Systems} \label{sec:Prelim:Diss_system}
Consider a nonlinear system
\begin{equation} \label{eq:Prelim:NL_System}
	\left\{\begin{aligned}
		\vxdot &= \vf(\vx,\vu), \\
		\vy &= \vh(\vx),
	\end{aligned}\right.
\end{equation}
where $\vx \in \Reals^{n}$, $\vu \in \Reals^{m}$, $\vy \in \Reals^{m}$ and where $\vf\colon\Reals^{n}\times\Reals^{m}\rightarrow\Reals^{n}$ and $\vh\colon\Reals^{n}\times\Reals^{m}\rightarrow\Reals^{m}$ are class \classC{1} functions.
\begin{definition}[Dissipative system, c.f.\ \cite{vdSchaft2017,Arcak2006,Khalil2002}] \label{def:Prelim:disspativity}
	A system \eqref{eq:Prelim:NL_System} with a class \classC{1} storage function $\cramped{\sS \colon \Reals^n \times \Reals^m \rightarrow \RealsPos}$ is dissipative w.r.t.\ a supply rate $\sw(\vu,\vy)$ if $\sSdot \le \sw(\vu,\vy)$.
\end{definition}
\begin{definition}[Quadratic supply rates, c.f.\ \cite{vdSchaft2017,Arcak2006,Khalil2002}] \label{def:Prelim:passive_rates}
	A system \eqref{eq:Prelim:NL_System} that is dissipative w.r.t.\ $\sw(\vu,\vy)$ is
	\begin{itemize}
		\item passive if $\sw = \vuT \vy$,
		\item \acf{IFP} if $\sw = \vuT \vy - \siIn \vuT \vu$,
		\item \acf{OFP} if $\sw = \vuT \vy - \siOut \vyT \vy$,
		\item \acf{IFOFP} if $\sw = \cramped{(1 + \siIn \siOut)} \vuT \vy  - \siIn \vuT \vu - \siOut \vyT \vy$,
		\item has an \Ltwo{}-gain of $\siLtwo$ if $\sw = \siLtwo^2 \vuT \vu - \vyT \vy$, 
	\end{itemize}
	where $\siLtwo > 0$ and $\siIn, \siOut \in \Reals$.
\end{definition}
%
%
\begin{definition}[{\Ac{ZSO} \cite[p.~46]{vdSchaft2017}}] \label{def:Passive:ZSO}
	A system \eqref{eq:Prelim:NL_System} is \ac{ZSO} if $\vu \equiv \vec{0}$ and $\vy \equiv \vec{0}$ implies $\vx \equiv \vec{0}$.
\end{definition}
%
%
For cases where the desired equilibrium of a system is not at the origin but at some constant value, the shifted passivity \cite[p.~96]{vdSchaft2017} or \ac{EIP} \cite{Hines2011} of a system must be investigated. Naturally, this requires that an equilibrium exists, i.e.\ there is a unique input $\vueq \in \Reals^m$ for every equilibrium $\vxeq \in \setXeq \subset \Reals^n$ such that \eqref{eq:Prelim:NL_System} produces $\vf(\vxeq,\vueq) = 0$ and $\vyeq = \vh(\vxeq,\vueq)$ \cite[p.~24]{Arcak2016}.
\begin{definition}[{\Ac{EIP} \cite[p.~24]{Arcak2016}}] \label{def:Passive:EIP}
	A system \eqref{eq:Prelim:NL_System} is \ac{EIP} if there exists a class \classC{1} storage function $\sS(\vx, \vxeq, \vu)$, $\cramped{\sS \colon \Reals^n \times \setXeq \times \Reals^m \rightarrow \RealsPos}$, with $\sS(\vxeq, \vxeq, \vueq) = 0$, that is dissipative w.r.t.\ $\sw(\vu - \vueq, \vy - \vyeq)$ for any equilibrium $(\vueq, \vyeq)$.
\end{definition}
\subsection{Passive Static Functions} \label{sec:Prelim:Diss_static}
Recall that a sector-bounded static nonlinear function is dissipative to a supply rate defined by the sector bound \cite[Def.~6.2]{Khalil2002}. We now consider the arbitrarily shifted single-input single-output function
%
\begin{equation} \label{eq:Prelim:Static_function}
	\left\{
	\begin{aligned}
		\sy &= \sh(\su), \quad \su, \sueq \in \setU, \quad \sy,\syeq \in \setY, \quad \sh: \setU \rightarrow \setY, \\
		\sye &= \she(\sue) \coloneqq \sh(\su) - \sh(\sueq) = \sy - \syeq, \quad \sue \coloneqq \su - \sueq
	\end{aligned}
	\right.
\end{equation}
and show how its dissipativity properties may be derived.
\begin{proposition}[\ac{EIP} static functions] \label{prop:Prelim:Static_load}
	A static function \eqref{eq:Prelim:Static_function} of class \classC{0} is \ac{IFOFP}$(\sclo,1/\schi)$ w.r.t.\ the arbitrarily shifted input-output pair $(\sue, \sye)$ if
	\begin{equation} \label{eq:Prelim:Load_gradient}
		\sclo \le \dFull{\sh(\su)}{\su} \le \schi, \quad \forall \su \in \setU.
	\end{equation}
	and $0 < \schi < \infty$.
\end{proposition}
\begin{proof}
	Consider for \eqref{eq:Prelim:Static_function} the slope between an arbitrary shift $(\sueq, \syeq) \in \setU\times\setY$ and a point $(\su, \sy)$, for which the upper and lower bounds are given by
	\begin{equation} \label{eq:Prelim:Static_gradient}
		\sclo \le \frac{\sy - \syeq}{\su - \sueq} \le \schi, \quad \forall (\su, \sy), (\sueq, \syeq) \in \setU\times\setY.
	\end{equation}
	Changing to the shifted variables $\sue$ and $\sye$ as in \eqref{prop:Prelim:Static_load} and multiplying through by $\sue^2$ yields 
	\begin{equation} \label{eq:Prelim:Load_sector_alt}
		\begin{aligned}
			\sclo\sue^2 \le \sue\sye \le \schi\sue^2 &\iff (\sye - \sclo\sue)(\sye - \schi\sue) \le 0 \\
			&\iff (\sye - \sclo\sue)(\frac{1}{\schi}\sye - \sue) \le 0, 
		\end{aligned}
	\end{equation}
	for $\schi > 0$, which describes an \ac{IFOFP} function (see \cite[p.~231]{Khalil2002}).
	Finally, through the mean value theorem, the bounds in \eqref{eq:Prelim:Static_gradient} may be found from \eqref{eq:Prelim:Load_gradient}.
\end{proof}
We note that the restrictions on $\schi$ in \autoref{prop:Prelim:Static_load} are needed from a computational point of view $(\schi < \infty)$ and to ensure that the passivity indices correspond to the correct sector\footnote{Consider e.g.\ the sector \autoref{prop:Prelim:Static_load} would yield if $\sclo \le \schi < 0$.} $(\schi > 0)$. However, this limits the passivity properties attainable through \autoref{prop:Prelim:Static_load} to $\siOut = 1/\schi > 0$.
\begin{remark}[Symmetrical sectors] \label{rem:Prelim:Symmetrical_sectors}
	Placing the additional restriction $\sclo = -\schi$ in \eqref{eq:Prelim:Load_gradient} results in the Lipschitz continuity of $\sh(\su)$. Moreover, this implies that the arbitrarily shifted function $\she(\sue)$ has a finite \Ltwo-gain of $\schi$ \cite{Malan2022b}.
\end{remark}
%
%
\subsection{Interconnected Quadratic Dissipative Systems} \label{sec:Prelim:Diss_interconnect}
Building upon the results on the interconnection of dissipative systems in \cite{Arcak2016,Moylan1978}, we now provide a method for finding dissipativity properties for a subset of the interconnected subsystems such that interconnected stability is guaranteed. Specifically, we look for the dissipative supply rates that restrict the subset of subsystems as little as possible. For a set $\setS$ of subsystems, define $\vu = [\vuT[1], \dots, \vuT[|\setS|]]^\Transpose$ and $\vy = [\vyT[1], \dots, \vyT[|\setS|]]^\Transpose$.
\begin{theorem}[Minimally restrictive stabilising indices] \label{thm:Prelim:Min_restrictive_diss}
	Consider $|\setS|$ subsystems of the form \eqref{eq:Prelim:NL_System} which are dissipative w.r.t.\ the supply rates $\sw[i] = 2\siCross[i] \vuT[i] \vy[i]  - \siIn[i] \vuT[i] \vu[i] - \siOut[i] \vyT[i] \vy[i]$ and are linearly interconnected according to $\vu = \mH \vy$. The stability of the interconnected system is guaranteed if there exists a $\mD$ and $\siIn[j], \siOut[j] \in \Reals$ with $j\in\setJ$ such that
	\begin{equation} \label{eq:Prelim:Opt_interconnected_stability}
		\begin{array}{cl}
			\!\displaystyle \underset{\begin{array}{c}
					\\[-3.6ex]
					\scriptstyle \mD,\, \siIn[j],\, \siOut[j], \\[-3.5pt]
					\scriptstyle j\in\setJ
				\end{array}}{\min} & \! \displaystyle \sum_{j \in \setJ}\left( \siIn[j] + \siOut[j]\right) \\
			\text{s.t.} & \siCross[j] = \nicefrac{1}{2}(1 + \siIn[j]\siOut[j]), \quad j\in\setJ, \\
			 & \mQ \negSemiDef 0, \quad \mD^2 \posDef 0
		\end{array}
	\end{equation}
	where the subsystems with configurable supply rates are represented by the set $\setJ \subset \setS$, and
	\begin{align}
		\label{eq:Prelim:Optimisation_mat_Q}
		\mQ &\coloneqq \begin{bmatrix} \mH \\ \Ident \end{bmatrix}^\Transpose \mD \mW \mD \begin{bmatrix} \mH \\ \Ident \end{bmatrix} \\
		\label{eq:Prelim:Optimisation_mat_D}
		\mD &\coloneqq \Diag[\vdT,\vdT], \qquad && \vd = (\sqrt{\sd[i]}), \\
		\label{eq:Prelim:Optimisation_mat_W}
		\mW &\coloneqq \begin{bmatrix}
			-\Diag[\siIn[i]] & \Diag[\siCross[i]] \\
			\Diag[\siCross[i]] & -\Diag[\siOut[i]]
		\end{bmatrix}, \quad && i \in \setS.
	\end{align}
\end{theorem}
The proof for \autoref{thm:Prelim:Min_restrictive_diss} follows analogously to the proof of \cite[Theorem~13]{Malan2022b} with application of \cite[Remark~5]{Malan2022b} and is thus omitted for brevity. Note that if $\setJ = \emptyset$ in \eqref{eq:Prelim:Opt_interconnected_stability}, \autoref{thm:Prelim:Min_restrictive_diss} can be used to verify the stability of interconnected dissipative systems.

Despite the design flexibility provided by \autoref{thm:Prelim:Min_restrictive_diss}, certain cascade configurations present obstacles to the application of dissipativity theory. The following proposition formalises the problem presented by one such configuration which arises in the sequel and is used to inform the control design.
\begin{proposition}[Non-dissipativity of cascaded \ac{IFP}-\ac{OFP} systems] \label{prop:Prelim:non_diss_IFP_OFP}
	Consider $|\setS| \ge 2$ subsystems \eqref{eq:Prelim:NL_System} which are dissipative w.r.t.\ $\sw[i] = 2\siCross[i] \vuT[i] \vy[i]  - \siIn[i] \vuT[i] \vu[i] - \siOut[i] \vyT[i] \vy[i]$ and linearly interconnected according to $\vu = \mH \vy$. Let $i = 1$ and $i = 2$ arbitrarily denote subsystems that are \ac{IFP} and \ac{OFP}, respectively. If these systems are connected in exclusive casade and do not form a feedback connection, i.e.\
	\begin{equation}
		\mH = \begin{bmatrix}
			0 & 0 & \ast \\
			1 & 0 & \vec{0} \\
			\vec{0} & \ast & \ast
		\end{bmatrix}, 
	\end{equation}
	then investigating stability via separable storage functions as in \autoref{thm:Prelim:Min_restrictive_diss} fails.
\end{proposition}
\begin{proof}
	Evaluating the stability criteria in \eqref{eq:Prelim:Opt_interconnected_stability} under the imposed \ac{IFP} and \ac{OFP} conditions yields the $\mQ$ \eqref{eq:Prelim:Optimisation_mat_Q} entries
	\begin{equation}
		\sQ[11] = \sd[1]\siOut[1] + \sd[2]\siIn[2] = 0, \qquad \sQ[12] = \sQ[21] = \frac{\sd[2]\siCross[2]}{2} = \frac{\sd[2]}{2}.
	\end{equation}
	Since $\sd[i] > 0$, $\mQ$ constitutes an indefinite saddle-point matrix \cite[Section~3.4]{Benzi2005}, violating the requirement in \eqref{eq:Prelim:Opt_interconnected_stability}.
\end{proof}
\begin{remark}[Non-separable storage functions] \label{rem:Prelim:non_seperable_functions}
	The obstacle in \autoref{prop:Prelim:non_diss_IFP_OFP} arises due to the storage functions being compartmentalised by the subsystem boundaries. While the separability of storage functions is a central motivation for the use of dissipativity theory, forgoing this allows for a stability analysis through less conservative methods (e.g.\ the KYP lemma).
\end{remark}
	\section{Problem Description} \label{sec:Problem}
In this section, the components comprising the DC mircrogrid are introduced in \autoref{sec:Problem:Elec_Network}. This is followed by \autoref{sec:Problem:DGU}, where controllers are added which regulate the output power of actuated buses in order to facilitate power sharing in the sequel. Finally, we formulate the coordination and cooperation goals as a control problem in \autoref{sec:Problem:Control_Problem}.
\subsection{DC Network} \label{sec:Problem:Elec_Network}
\begin{figure*}
	\centering
	\scalebox{.8}{\begin{tikzpicture}
	\def\cHeight{1.85cm}	
	\def\colorLine{blue!80!black}	
	\coordinate(dgu_filter_base);
	\draw
	(dgu_filter_base) to [open, o-, v<=${\svVSC[k]}$] ++(0,\cHeight) coordinate(dgu_filter_base_high)
	to [short, o-] ++(0.4,0)
	to [ccgsw, -, name=switchAlpha] ++(1.0,0)
	to [short, i=${\sii[k]}$] ++(0.8,0)
	to [R, l=${\sR[k]}$] ++(1.2,0) 
	to [L, l=${\sL[k]}$] ++(2.2,0) coordinate(dgu_filter_output_high) 
	to [C, l_=${\sC[k]}$] ++(0,-\cHeight) coordinate(dgu_filter_output_low) 
	to [short] (dgu_filter_base)
	
	(dgu_filter_output_high) to [short] ++(1.2,0) coordinate(dgu_load_high)
	to [american current source, l=${\sILv[k]}$] ++(0,-\cHeight) coordinate(dgu_load_low)
	to [short] (dgu_filter_output_low)
	
	(dgu_load_high) to [short, -o] ++(2.1, 0) coordinate(dgu_pcc_high)
	to [open] ++(0, -\cHeight) coordinate(dgu_pcc_low)
	to [short, o-] (dgu_load_low)
	
	(dgu_pcc_high) to [open, v^>=${\sv[k]}$] (dgu_pcc_low)
	
	(dgu_filter_base) to [short] ++(-2.2,0) coordinate(dgu_source_low)
	to [battery, invert] ++(0,\cHeight) coordinate(dgu_source_high)
	to [short] (dgu_filter_base_high);
	
	\path (dgu_filter_base) +(-1.0,\cHeight/2) coordinate (inverter);
	\node[align=center,draw,fill=white,rectangle,minimum height =\cHeight + 0.6cm,minimum width =1.0cm ](nVSC) at(inverter) {\footnotesize Buck$_k$};
	
	\node[below=0.08] at (dgu_pcc_low) {Bus$_{k}$};
	
%

	\coordinate(line_left_low) at (dgu_pcc_low);
	\coordinate(line_left_high) at (dgu_pcc_high);
	\draw	
	(line_left_high) to [short, o-] ++(1.8,0) coordinate(line_Ci_high)
	to [C, color=\colorLine, \colorLine, l_=$\dfrac{\sC[kl]}{2}$] ++(0,-\cHeight) coordinate(line_Ci_low)
	to [short, -o] (line_left_low)
	
	
	(line_Ci_high) to [short] ++(1.0,0)
	to [R, color=\colorLine, \colorLine, l=${\sR[kl]}$] ++(1.2,0)
	to [short, i^>=${\siTx[kl]}$] ++(1.2, 0)
	to [L, color=\colorLine, \colorLine, l=${\sL[kl]}$] ++(1.2,0)
	to [short] ++(0.8,0) coordinate(line_Cj_high)
	to [C, color=\colorLine, \colorLine, l=$\dfrac{\sC[kl]}{2}$] ++(0,-\cHeight) coordinate(line_Cj_low)
	to [short] (line_Ci_low)
	
	(line_Cj_high) to [short, -o] ++(1.0,0) coordinate(line_right_high)
	to [open] ++(0,-\cHeight) coordinate(line_right_low)
	to [short, o-] (line_Cj_low);
	
	\node[below=0.08] at (line_right_low) {Bus$_{l}$};

	\path (dgu_source_low) +(-0.5,-0.5) coordinate (dgu_box_bottom_left);
	\path (line_Ci_high) +(0.5,0.75) coordinate (dgu_box_top_right);
	
	\path (line_Ci_low) +(0.9, -0.5) coordinate (line_box_bottom_left);
	\path (line_Cj_high) +(-0.7, 0.75) coordinate (line_box_top_right);
	
	\begin{scope}[on background layer]
		\node[draw,dashed,rounded corners=0.25cm,fit=(dgu_box_bottom_left) (dgu_box_top_right)] (dgu_box) {};
		
		\node[draw,dashed,rounded corners=0.25cm,fit=(line_box_bottom_left) (line_box_top_right)] (line_box) {};
		\node[above] at (line_box.south) {Line$_{kl}$};
	\end{scope}
	
	\begin{scope}[on background layer]
		\path (switchAlpha.mid) +(0, 0.9) coordinate (alpha_input);
		\draw[-, red, thick] (alpha_input) -- node[pos=0, anchor = south]{$\sAct[k] \in \{0,1\}$} (switchAlpha.mid);
		
		\path (nVSC.north) +(-1.0, 1.0) coordinate (setpoint_input);
		\draw[-latex, red, thick] (setpoint_input) -| node[pos=0, anchor = east]{$\spSP[k]$} (nVSC.north);
		
		\path (dgu_pcc_high) ++(0, 0.2) coordinate (measure_ouput_start) 
			++(-1.0, 1.0) coordinate (measure_ouput);
		\draw[{Circle[length=2.5pt]}-latex, red, thick] (measure_ouput_start) |- node[pos=1, anchor = east]{$\sv[k]$} (measure_ouput);
	\end{scope}

	
\end{tikzpicture}}
	\caption{
		Circuit diagram of a bus comprising a DC-DC buck converter, a filter, and a current source representing a load, connected to a $\pi$-model line (blue); the line capacitances considered to be part of the respective buses.}
	\label{fig:Problem:DC_Elements}
\end{figure*}
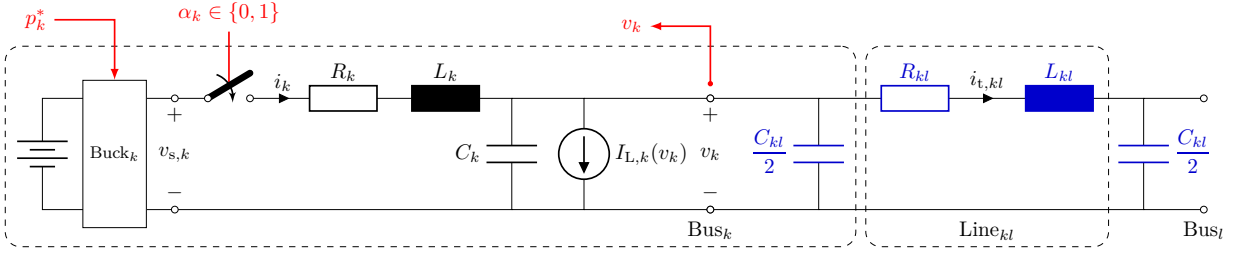
We consider a DC microgrid comprising $N = |\setN|$ buses connected by via $\pi$-model electrical lines, as depicted in \autoref{fig:Problem:DC_Elements}. Let the graph $\graphGPhys = (\setN, \graphEPhys)$ describe the interconnection with $\setN$ as the set of buses and $\graphEPhys$ as the set of lines. Without loss of generalisation, we allow each node to inject power through a DC-DC buck converter connected via a lossy LC-filter. Note that a time-averaged model (see e.g.\ \cite{Trip2019}) is used for the buck converter and the energy source is assumed to be ideal but finite.

Let the buses be split into an actuated set $\setNAct$ and an unactuated set $\setNUnact$, according to whether the buck converter can freely regulate the amount of power injected at a given time. Buses may freely switch between the sets $\setNAct$ and $\setNUnact$, but $\setNAct \cap \setNUnact = \emptyset$ and $\setNAct \cup \setNUnact = \setN$ always hold. To characterise this actuation state of a bus, define the piecewise-constant, time-varying actuation parameter $\sAct[k](t)$ as 
\begin{equation} \label{eq:Problem:actuation}
	\sAct[k](t) \coloneqq \left\lbrace\begin{array}{lll}
		1, & \qquad & k \in \setNAct, \\
		0, & \qquad & k \in \setNUnact .
	\end{array}\right.
\end{equation}
Note that we omit the time dependence of $\sAct[k]$ in the sequel.

The dynamics for actuated buses with \acp{DGU}, where $\sAct[k] = 1$ with $k \in \setNAct$ are described by
\begin{equation} \label{eq:Problem:node_act_dynamics}
	\begin{bmatrix}
		\sL[k]\sidot[k] \\ \sCeq[k]\svdot[k]
	\end{bmatrix} = \begin{bmatrix}
		-\sR[k] & -1 \\
		1 & 0
	\end{bmatrix} \!\! \begin{bmatrix}
		\sii[k] \\ \sv[k]
	\end{bmatrix} \!\: + \!\: \begin{bmatrix}
		\svVSC[k] \\ -\vEPT[k]\viTx - \sILv[k]
	\end{bmatrix}
\end{equation}
where $\sCeq[k] = \sC[k] + \nicefrac{1}{2} \vEPT[k] \Diag[\sC[kl]] \vEP[k]$; $\sC[k], \sC[kl], \sL[k] > 0$; $\sii[k] \in \Reals$; and $\sv[k] \in \RealsPos$. The line currents $\viTx$ connect to the capacitor voltages according to incidence matrix $\mEP = (\vEPT[k])$ of $\graphGPhys$.
The dynamics of the unactuated \emph{load} buses with $\sAct[k] = 0$ correspond to the simplified system
\begin{equation} \label{eq:Problem:node_unact_dynamics}
	\sCeq[k]\svdot[k] = -\vEPT[k]\viTx - \sILv[k], \quad k \in \setNUnact
\end{equation}

In both the actuated \eqref{eq:Problem:node_act_dynamics} and unactuated \eqref{eq:Problem:node_unact_dynamics} cases, the loads are considered static, nonlinear voltage-dependent current sources which are described by class \classC{0} functions. In this work, we utilise the standard ZIP-model comprising constant impedance, constant current and constant power parts. Note that other continuous functions may also be used without restriction\footnote{This includes exponential loads (see e.g.\ \cite{Strehle2020Load}).}. As described in \cite[pp.~110--112]{Machowski2008}, we define a critical voltage $\svCrit$, typically set to $0.7\svRef$, below which the loads are purely resistive. Thus,
\begin{align} \label{eq:Problem:ZIP_model}
	\sILv[k] =& \left\{\begin{array}{lll}
		\sZinv[k]\cdot\sv[k] + \sI[k] + \dfrac{\sP[k]}{\sv[k]}, & \qquad & \sv[k] \ge \svCrit, \\[1pt]
		\sZCritinv[k]\cdot\sv[k], & \qquad & \sv[k] < \svCrit ,
	\end{array} \right. \\[2pt]
	\label{eq:Problem:Z_crit}
	\sZCritinv[k] \coloneqq& \; \frac{\sIL[k](\svCrit)}{\svCrit} = \sZinv[k] + \frac{\sI[k]}{\svCrit} + \frac{\sP[k]}{\svCrit^2} \, ,
\end{align}
describes a static, nonlinear load which conforms to \eqref{eq:Prelim:Static_function}.

Lastly, the $\pi$-model transmission lines physically connecting the nodes are governed by the dynamics
\begin{equation} \label{eq:Problem:line_dynamics}
	\sL[kl]\siTxdot[kl] = -\sR[kl]\siTx[kl] + \vEPT[kl]\vv , \quad kl \in \graphEPhys,
\end{equation}
where $\siTx[kl] \in \Reals$, $\sL[kl], \sR[kl] > 0$ and $(\vEPT[kl])^\Transpose = \mEP$. Note that the line capacitances are included in the equivalent capacitances $\sCeq[k]$ at the buses.
\subsection{\ac{DGU} Power Regulator} \label{sec:Problem:DGU}
To allow for power sharing between the actuated buses \eqref{eq:Problem:node_act_dynamics} in the sequel, we equip each \ac{DGU} with a controller that can regulate the injected power to a desired setpoint $\spSP[k]$. This regulator has the form 
\begin{equation} \label{eq:Problem:DGU_Power_Control}
	\begin{aligned}
		\seDGUdot[k] &= \sAct[k](\spSP[k] - \spp[k]) \\
		\svVSC[k] &= \skpDGU (\spSP[k] - \spp[k]) + \skiDGU \seDGU[k] + \sRDGUdamp[k] \sii[k] + \svRef
	\end{aligned}
\end{equation}
where $\seDGU \in \Reals$, $\spp[k] = \sv[k]\sii[k]$ is the actual power injected, $\sRDGUdamp \in \Reals$ is the damping added to the system, and $\skpDGU, \skiDGU > 0$ are the control parameters. Combining \eqref{eq:Problem:DGU_Power_Control} with \eqref{eq:Problem:node_act_dynamics} yields the nonlinear system describing the actuated agents $k \in \setNAct$

\begin{equation} \label{eq:Problem:Controlled_DGU}
	\begin{aligned}
		\begin{bmatrix}
			\seDGUdot[k] \\ \sL[k] \sidot[k] \\ \sC[k] \svdot[k]
		\end{bmatrix} \! =&  \begin{bmatrix}
			0 & -\sv[k] & 0 \\
			\skiDGU & \sRDGUdamp[k] - \sR[k] - \skpDGU \sv[k] & -1 \\
			0 & 1 & 0
		\end{bmatrix} \!\!\! \begin{bmatrix}
			\seDGU[k] \\ \sii[k] \\ \sv[k]
		\end{bmatrix} \\
		&+ \begin{bmatrix}
			\spSP[k] \\ \skpDGU \spSP[k] + \svRef \\ -\vET[k]\viTx -\sILv[k]
		\end{bmatrix}
	\end{aligned}
\end{equation}
\begin{remark}[Regulating current or voltage] \label{rem:Problem:I_or_V_regulation}
	Without invalidating the stability analysis in the sequel, the regulator in \eqref{eq:Problem:DGU_Power_Control} can be exchanged for simpler, purely linear current or voltage regulators (see e.g.\ \cite{Tucci2016,Strehle2020DC,Cucuzzella2023}).
\end{remark}
\begin{remark}[Constrained \ac{DGU} operation] \label{rem:Problem:Power_limited_DGU}
	If an actuated \ac{DGU} cannot provide the desired power $\spSP[k]$, e.g.\ due to current, storage or temperature limitations, the \ac{DGU} may simply set its actuation state $\sAct[k] = 0$ to disable its control. If some power can still be supplied, it may simply be regarded as a negative load. This allows \acp{DGU} to contribute to the power supply of the network, even in the face of control limitations.
\end{remark}
\subsection{Control Problem} \label{sec:Problem:Control_Problem}
A central requirement for DC microgrids is voltage stability, which requires the bus voltages to remain within a given tolerance band around the reference $\svRef$. Specifically, this requirement should be met throughout the network, and not only at the actuated buses. Due to the presence of lossy lines, power flows are associated with voltage differences between buses, meaning that $\cramped{\sv[k] \to \svRef}, \, \forall k \in \setN$ is not practical. Ideally, the voltages at all buses should be arrayed in the tolerance band around $\svRef$ and be as close to $\svRef$ as possible\footnote{The magnitude of the errors $\svRef - \sv[k]$ strongly depend on the loads and line resistance. Small errors therefore presuppose adequate network design.}. The manipulated variables used to achieve this are the power setpoints $\spSP[k]$ supplied to the actuated \acp{DGU} \eqref{eq:Problem:DGU_Power_Control}. This leads to the first objective for the control of the DC microgrid, which involves finding the setpoints $\spSP[k]$ that ensure the weighted average voltage equals $\svRef$ at steady state.
\begin{objective}[Weighted voltage consensus] \label{obj:Problem:weighted_errors}
	%
	\begin{equation} \label{eq:Problem:Obj_voltage_consensus}
		\text{Find\ } \, \spSP[k] \text{\ s.t.\ } \lim_{t \to \infty} \frac{1}{N} \sum_{k \in \setN} \sh(\sv[k](t)) = \svRef
	\end{equation}
	for a strictly increasing weighting function $\sh: \Reals \to \Reals$.
\end{objective}
By choosing a nonlinear $\sh$, large voltage errors may be weighed more strongly. This allows for better utilisation of the tolerance band since bus voltages can be further from $\svRef$ before registering as a significant error.

In addition to \autoref{obj:Problem:weighted_errors}, it is desired that all actuated \acp{DGU} contribute towards supplying and stabilising this network. Ensuring that all \acp{DGU} receive the same setpoint spreads the load across actuated buses, leading to a reduction in localised stress on the \acp{DGU}. We thus formulate the second objective as requiring uniform setpoints for the \acp{DGU} in steady state.
\begin{objective}[Cooperative power sharing] \label{obj:Problem:power_sharing}
	\begin{equation}
		\lim_{t \to \infty} (\spSP[k](t) - \spSP[l](t)) = 0, \quad \forall \, k, l \in \setN
	\end{equation}
\end{objective}

Achieving \autorefMulti{obj:Problem:weighted_errors, obj:Problem:power_sharing} thus yields a controlled microgrid where the average weighted voltage error of all buses tends to zero through the coordinated action of the actuated buses in a distributed fashion. These objectives also allow \acp{DGU} to transition seamlessly between actuated and unactated states and ensure no measurement information is discarded simply because a bus cannot regulate itself. Notice that disregarding the unactuated buses in \autorefMulti{obj:Problem:weighted_errors, obj:Problem:power_sharing} yields the objectives typically used in the literature \cite{Nasirian2015, Trip2019, Cucuzzella2019cons, Sadabadi2022, Nahata2020, Nahata2022, Zhao2015, Cucuzzella2019pbc}.

To achieve these objectives, we make the following assumptions related to appropriate network design.
\begin{assumption}[Feasible network] \label{ass:Problem:Feasible_network}
	The available power sources can feasibly supply the loads with power over the given electrical network, i.e.\ a suitable equilibrium for the microgrid exists.
\end{assumption}
\begin{assumption}[Number of actuated \acp{DGU}] \label{ass:Problem:Actuated_agents}
	At least one \ac{DGU} is actuated at any given time, i.e.\ $\setNAct \neq \emptyset$.
\end{assumption}
\begin{assumption}[Connected topologies] \label{ass:Problem:Connected_graphs}
	\autorefMulti{obj:Problem:weighted_errors, obj:Problem:power_sharing} only apply to a subset of buses electrically connected as per $\graphGPhys$. Moreover, for a distributed control, a connected communication graph exclusively interconnects the same subset of buses.
\end{assumption}
Note that \autoref{ass:Problem:Feasible_network} is a typically made implicitly or explicitly in the literature (see e.g.\ the discussion in \cite{Nahata2020}). \autorefMulti{ass:Problem:Actuated_agents, ass:Problem:Connected_graphs} further specify requirements that allow a distributed control to achieve the feasible state in \autoref{ass:Problem:Feasible_network}, i.e.\ by ensuring that at least one source of stabilisation is present in the network (\autoref{ass:Problem:Connected_graphs}), and by ensuring that the coordination corresponds to the network to be controlled \autoref{rem:Problem:proportional_power_sharing}.
\begin{remark}[Proportional power sharing] \label{rem:Problem:proportional_power_sharing}
	By normalising the power setpoint $\spSP[k]$ and weighing the input in \eqref{eq:Problem:DGU_Power_Control} according to the rated power of a given \ac{DGU}, \autoref{obj:Problem:power_sharing} automatically describes a proportional power sharing. With reference to \autoref{rem:Problem:Power_limited_DGU}, this also allows the constrained \acp{DGU} to lower their maximum injectable power instead of setting the \acp{DGU} to the unactuated state $\sAct[k] = 0$. We omit the extension to proportional power sharing in this work for simplicity.
\end{remark}
	\section{Control Structure} \label{sec:Control}
\begin{figure}[t]
	\centering
	\resizebox{\columnwidth}{!}{\hspace*{10pt}%
		\tikzsetnextfilename{03_Img/control_structure}%
		\begin{tikzpicture}[thick,>=latex']	
	\def\blockHeight{0.8cm}
	\def\blockWidth{0.85cm}
	\def\muxWidth{0.1cm}
	\def\nodedist{1.7}
	
	\def\feedbackHeight{2.4}
	\def\localControlHeight{3.3cm}
	
	\tikzstyle{systemElement}	= [ fill=white ]
	
	\tikzstyle{sum}				= [node distance=\nodedist, draw, circle, systemElement]
	\tikzstyle{box}				= [node distance=\nodedist, draw, rectangle, minimum height=\blockHeight, minimum width=\blockWidth, systemElement]
	\tikzstyle{mux}				= [node distance=\nodedist, draw, fill=black, rectangle, minimum height={\localControlHeight}, minimum width=\muxWidth, inner sep=-2]
	
	\tikzstyle{arrow}			= [->]
	\tikzstyle{noArrow}			= [-]
	\tikzstyle{comms}			= [-,bend left, densely dotted, line width=1.0pt, red]
	\tikzstyle{hiddenSystems}	= [-, dotted, line width=1.4pt]

	\coordinate(input_Ref);
	
	\path
	(input_Ref) to ++(\nodedist*6/7,0) coordinate (sum_feedback)
	to ++(\nodedist*4/7,0) coordinate (demux_left)
	to ++(\nodedist*5.5/7,0) coordinate (weight_func)
	to ++(\nodedist*8.5/7,0) coordinate (dda_1_output)
	to ++(\nodedist*9/7,0) coordinate (control_pi)
	to ++(\nodedist*8.5/7,0) coordinate (dda_2_control)
	to ++(\nodedist*8/7,0) coordinate (mux_right)
	to ++(\nodedist*3/7,-\feedbackHeight) coordinate (sum_output);

	\path (control_pi)
	to +(0,{\localControlHeight/3}) coordinate (control_pi_1)
	to +(0,{-\localControlHeight/3}) coordinate (control_pi_n);
	
	\path (weight_func)
	to +(0,{\localControlHeight/3}) coordinate (weight_func_1)
	to +(0,{-\localControlHeight/3}) coordinate (weight_func_n);
	
	\path (dda_1_output)
	to +(0,{\localControlHeight/3}) coordinate (dda_1_output_1)
	to +(0,{-\localControlHeight/3}) coordinate (dda_1_output_n);
	
	\path (dda_2_control)
	to +(0,{\localControlHeight/3}) coordinate (dda_2_control_1)
	to +(0,{-\localControlHeight/3}) coordinate (dda_2_control_n);

	\path ( $(demux_left)!0.5!(mux_right)$ ) to ++(0,-\feedbackHeight) coordinate (system_block);
	
	\path (sum_output) to ++(\nodedist*6/7,0) coordinate (input_out_Ref);

	\node[sum] at (sum_feedback)	(node_sum_feedback)		{};
	\node[mux] at (demux_left)		(node_demux_left)		{};
	\node[box] at (weight_func_1)	(node_weight_func_1)	{{$\shNLweight$}};
	\node[box] at (weight_func_n)	(node_weight_func_n)	{{$\shNLweight$}};
	\node[box] at (dda_1_output_1)	(node_dda_1_output_1)	{DDA$_{2,1}$};
	\node[box] at (dda_1_output_n)	(node_dda_1_output_n)	{DDA$_{2,N}$};
	\node[box] at (control_pi_1)	(node_control_pi_1)		{PI$_1$};
	\node[box] at (control_pi_n)	(node_control_pi_n)		{PI$_N$};
	\node[box] at (dda_2_control_1)	(node_dda_2_control_1)	{DDA$_{4,1}$};
	\node[box] at (dda_2_control_n)	(node_dda_2_control_n)	{DDA$_{4,N}$};
	\node[mux] at (mux_right)		(node_mux_right)		{};
	\node[box] at (system_block)	(node_system_block)		{DC MG};
	
	\path (node_weight_func_1) -- (node_weight_func_n) coordinate[pos=0.3] (node_weight_func_dots_start) coordinate[pos=0.5] (node_weight_func_dots) coordinate[pos=0.7] (node_weight_func_dots_end);
	\draw[hiddenSystems] (node_weight_func_dots_start) -- (node_weight_func_dots_end);
	
	\path (node_control_pi_1) -- (node_control_pi_n) coordinate[pos=0.3] (node_control_pi_dots_start) coordinate[pos=0.5] (node_control_pi_dots) coordinate[pos=0.7] (node_control_pi_dots_end);
	\draw[hiddenSystems] (node_control_pi_dots_start) -- (node_control_pi_dots_end);
	
	\path (node_dda_1_output_1) -- (node_dda_1_output_n) coordinate[pos=0.3] (node_dda_1_output_dots_start) coordinate[pos=0.5] (node_dda_1_output_dots) coordinate[pos=0.7] (node_dda_1_output_dots_end);
	\draw[hiddenSystems] (node_dda_1_output_dots_start) -- (node_dda_1_output_dots_end);
	
	\path (node_dda_2_control_1) -- (node_dda_2_control_n) coordinate[pos=0.3] (node_dda_2_control_dots_start) coordinate[pos=0.5] (node_dda_2_control_dots) coordinate[pos=0.7] (node_dda_2_control_dots_end);
	\draw[hiddenSystems] (node_dda_2_control_dots_start) -- (node_dda_2_control_dots_end);
	
	\node[above=24pt, align=center, anchor=base](nodeTextWeightFunc) at(weight_func_1) {Stage 1};
	\node[above=24pt, align=center, anchor=base](nodeTextStage1) at(dda_1_output_1) {Stage 2};
	\node[above=24pt, align=center, anchor=base](nodeTextStage2) at(node_control_pi_1) {Stage 3};
	\node[above=24pt, align=center, anchor=base](nodeTextStage3) at(node_dda_2_control_1) {Stage 4};

	\draw[arrow] (node_sum_feedback) -- node[pos=0.4, anchor = south]{\large $\vuNLweight$} (node_demux_left);
	\draw[arrow] (node_demux_left.east) ++(0,{\localControlHeight/3}) -- node[pos=0.5, anchor = south, name = node_weight_func_1_text]{$\suNLweight[1]$} (node_weight_func_1);
	\draw[arrow] (node_demux_left.east) ++(0,{-\localControlHeight/3}) -- node[pos=0.5, anchor = south, name = node_weight_func_n_text]{$\suNLweight[N]$} (node_weight_func_n);
	\draw[arrow] (node_weight_func_1) -- node[pos=0.5, anchor = south]{$\syNLweight[1]$} (node_dda_1_output_1);
	\draw[arrow] (node_weight_func_n) -- node[pos=0.5, anchor = south]{$\syNLweight[N]$} (node_dda_1_output_n);
	\draw[arrow] (node_dda_1_output_1) -- node[pos=0.5, anchor = south]{$\syDDA[2,1]$} (node_control_pi_1);
	\draw[arrow] (node_dda_1_output_n) -- node[pos=0.5, anchor = south]{$\syDDA[2,N]$} (node_control_pi_n);
	\draw[arrow] (node_control_pi_1) -- node[pos=0.5, anchor = south]{$\syCtrl[1]$} (node_dda_2_control_1);
	\draw[arrow] (node_control_pi_n) -- node[pos=0.5, anchor = south]{$\syCtrl[N]$} (node_dda_2_control_n);
	\draw[arrow] (node_dda_2_control_1) -- node[pos=0.5, anchor = south, name = node_dda_2_control_1_text]{$\syDDA[4,1]$} ($(node_mux_right.west) + (0,{\localControlHeight/3})$);
	\draw[arrow] (node_dda_2_control_n) -- node[pos=0.5, anchor = south, name = node_dda_2_control_n_text]{$\syDDA[4,N]$} ($(node_mux_right.west) + (0,{-\localControlHeight/3})$);
	
	\draw[arrow] (node_mux_right) -- ++({\nodedist/4},0) |- (node_system_block) node[pos=0.95, anchor = south west]{\large $\vpSP$};
	
	\draw[arrow] (node_system_block) -| node[pos=0.05, anchor = south east]{\large $\vv$} (node_sum_feedback) node[pos=0.95, anchor = north east]{$-$};
	
	\draw[arrow] (input_Ref) node[anchor = south west] {\large $\svRef\vOneCol[N]$} -- (node_sum_feedback) node[pos=0.95, anchor = north east]{$+$};
	
	\draw[comms] (node_dda_1_output_n) to (node_dda_1_output_1);
	\draw[comms] ($(node_dda_1_output_dots) + (4pt,-3pt)$) to (node_dda_1_output_n);
	\draw[comms] (node_dda_1_output_1) to ($(node_dda_1_output_dots) + (4pt,3pt)$);
	
	\draw[comms] (node_dda_2_control_1) to (node_dda_2_control_n);
	\draw[comms] (node_dda_2_control_n) to ($(node_dda_2_control_dots) + (-4pt,-3pt)$);
	\draw[comms] ($(node_dda_2_control_dots) + (-4pt,3pt)$) to (node_dda_2_control_1);

	\path (node_weight_func_1.west) +(0.05,0) coordinate (cBound_1_left);
	\path (node_dda_2_control_1.east) +(-0.03,0) coordinate (cBound_1_right);
	\path (node_weight_func_1.north) +(0,0.02) coordinate (cBound_1_top);
	\path (node_dda_2_control_1.south) +(0,0.0) coordinate (cBound_1_bottom);
	
	\path (node_weight_func_n.west) +(+0.05,0) coordinate (cBound_n_left);
	\path (node_dda_2_control_n.east) +(-0.05,0) coordinate (cBound_n_right);
	\path (node_weight_func_n.north) +(0,-0.0) coordinate (cBound_n_top);
	\path (node_dda_2_control_n.south) +(0,-0.02) coordinate (cBound_n_bottom);
	
	\path (cBound_1_left) -- (cBound_n_left) coordinate[pos=0.5] (cBound_dots_left);
	\path (cBound_1_right) -- (cBound_n_right) coordinate[pos=0.5] (cBound_dots_right);
	\path (node_weight_func_dots_start.north) +(0,-0.03) coordinate (cBound_dots_top);
	\path (node_dda_2_control_dots_end.south) +(0,0.03) coordinate (cBound_dots_bottom);

	\begin{scope}[on background layer]
		\node[fill=black!10,rounded corners=0.05cm,fit=(cBound_1_left) (cBound_1_right) (cBound_1_top) (cBound_1_bottom)] (box_control) {};
		\node[fill=black!10,rounded corners=0.05cm,fit=(cBound_dots_left) (cBound_dots_right) (cBound_dots_top) (cBound_dots_bottom)] (box_control) {};
		\node[fill=black!10,rounded corners=0.05cm,fit=(cBound_n_left) (cBound_n_right) (cBound_n_top) (cBound_n_bottom)] (box_control) {};
	\end{scope}
\end{tikzpicture}%
	\hspace*{-35pt}}
	\caption{Distributed four-stage control connected in feedback to the microgrid and with indicated communication links %
		\begin{tikzpicture}[baseline=-0.5ex]
			\protect\draw[-,bend left, densely dotted, line width=1.2pt, red] (0,0) to (0.27cm,0);
		\end{tikzpicture} %
		between the local control structures.}
	\label{fig:Control:control_structure}
\end{figure}
To meet \autorefMulti{obj:Problem:weighted_errors, obj:Problem:power_sharing}, we propose the four-stage control structure depicted in \autoref{fig:Control:control_structure}. This control structure comprises two \ac{DDA} implementations separated by agent PI controllers local to the buses as in \cite{Malan2022a}. This is prepended by a nonlinear weighting function $\shNLweight$. In the \autorefMulti{sec:Control:DDA, sec:Control:Leaky_PI, sec:Control:NL_Func}, we successively introduce these respective subsystems. Finally in \autoref{sec:Control:Steady_State}, we show that the control structure meets the objectives.
\subsection{DDA Controller} \label{sec:Control:DDA}
Consider the communiation graph $\graphGComm = (\setN, \graphEComm)$ linking the buses of the DC microgrid. The communication graph comprises the same vertices as the physical interconnection graph $\graphGPhys$ but possibly with a different topology. Let $\mLapComm$ denote the Laplacian of $\graphGComm$. For Stages 2 and 4 of the control structure, each agent implements an instance of the \ac{DDA}\footnote{We implement the PI-\ac{DDA} variant proposed in \cite{Freeman2006} and use the same communication graph for the proportional and integral terms.} described in \cite{Freeman2006}. The instances of the respective stages may be combined into vector form as
\begin{equation} \label{eq:Control:DDA}
	\text{DDA}_s \! \left\{ \!\!\!\;
		\begin{aligned}
			\begin{bmatrix} \vxDDAdot[s] \\ \vzDDAdot[s] \end{bmatrix} \!\! &= \!\! \begin{bmatrix}
				\!-\sgDDA \Ident[N] \!-\! \mLapComm[P] \!\! & \! \mLapCommT[I] \\
				-\mLapComm[I] & \matrix{0}
			\end{bmatrix} \!\!\! \begin{bmatrix} \vxDDA[s] \\ \vzDDA[s] \end{bmatrix} \!\!+\!\!
			\begin{bmatrix}
				\!\sgDDA \Ident[N] \! \\ \matrix{0}
			\end{bmatrix} \!\! \vuDDA[s], \\
			\vyDDA[s] &= \vxDDA[s],
		\end{aligned}
	\right.
\end{equation}
where $s \in \{2,4\}$ denotes the stage in \autoref{fig:Control:control_structure}, and $\vxDDA[s], \vzDDA[s] \in \Reals^N$ are the consensus and integral states respectively.
Furthermore, $\sgDDA > 0$ is a global estimator parameter (see \cite{Freeman2006}), and $\mLapComm[I] = \skiDDA \mLapComm$ and $\mLapComm[P] = \skpDDA \mLapComm$ are Laplacian matrices weighted for the integral and proportional responses, respectively. Recall from \cite{Freeman2006} that a constant input $\vuDDA[s]$ yields
\begin{equation} \label{eq:Control:DDA_equilibrium}
	\lim_{t \to \infty} \syDDA[s, k] = \frac{\vuDDAT[s]\vOneCol[N]}{N} , \quad \forall\,k.
\end{equation}
\subsection{Agent PI Controller} \label{sec:Control:Leaky_PI}
In Stage 3, we equip each bus $k \in \setN$ with a leaky agent PI controller similar to the approach in \cite{Weitenberg2018} 
\begin{equation} \label{eq:Control:local_PI}
	\text{PI}_k \left\{ \,\, \begin{aligned}
		\sxCtrldot[k] &= -\sdampCtrl\sxCtrl[k] + \suCtrl[k], \\
		\syCtrl[k] &= \skiCtrl \sxCtrl[k] + \skpCtrl \suCtrl[k],
	\end{aligned} \right.
\end{equation}
where $\sxCtrl[k] \in \Reals$, $\sdampCtrl \ge 0$ and $\skpCtrl, \skiCtrl > 0$. Note that $\sdampCtrl = 0$ reduces \eqref{eq:Control:local_PI} to an ideal PI controller. The combined form of the $N$ agent controllers is
\begin{equation} \label{eq:Control:local_PI_vector}
	\begin{aligned}
		\vxCtrldot &= -\sdampCtrl\vxCtrl + \vuCtrl, \\
		\vyCtrl &= \skiCtrl \vxCtrl + \skpCtrl \vuCtrl
	\end{aligned}
\end{equation}
\begin{remark}[Non-ideal integrators] \label{rem:Control:Agent_PI:Non_ideal}
	As shown in the sequel, ideal PI controllers only exhibit an \ac{IFP} property, whereas the \ac{DDA} controller is \ac{OFP}. The interconnection in \autoref{fig:Control:control_structure} thus yields a cascaded \ac{IFP}-\ac{OFP} structure which obstructs the dissipativity analysis (see \autoref{prop:Prelim:non_diss_IFP_OFP}). The use of leaky integrators $(\sdampCtrl > 0)$ overcomes this obstacle at the cost of negatively affecting the steady-state properties, since \eqref{eq:Control:local_PI} forces the equilibrium
	\begin{equation} \label{eq:Control:leaky_PI_steady_state}
		\vuCtrl = \sdampCtrl \vxCtrl
	\end{equation} 
	instead of $\vuCtrl = \vec{0}$. In the context of \autoref{fig:Control:control_structure}, this corresponds to a unwanted steady-state offset for the average weighted voltage error.
\end{remark}
\begin{remark}[Agent PI controller anti-windup] \label{rem:Control:Agent_PI:Anti_windup}
	To prevent controller windup, the input to the PI control in \eqref{eq:Control:local_PI} should be zeroed for any unactuated agents that are disconnected from the communication network.
\end{remark}
\begin{remark}[Non-participating agents] \label{rem:Control:Agent_PI:Non_participants}
	Implementing \eqref{eq:Control:local_PI} at each bus $k \in \setN$ allows for a faster reaction to disturbances at the cost of controller redundancy. By setting $\suDDA[4,m] \coloneqq \syDDA[4, m]$ at Stage 4 \ac{DDA} of the control structure for some agents $m \in \setM \subset \setN$, the PI control \eqref{eq:Control:local_PI} can be omitted at the agents in $\setM$ without affecting the steady state. Nevertheless, the measurements of the buses in $k \in \setM$ are still included in the Stage~2 \ac{DDA}. Note that at least one participating agent PI controller is required (see \cite[Remark~8]{Malan2022a}).
\end{remark}
\subsection{Weighting Function} \label{sec:Control:NL_Func}
\begin{figure}[!t]
	\centering
	\resizebox{0.7\columnwidth}{!}{%
		\tikzsetnextfilename{03_Img/nonlinear_weighting_function}%
		\begin{tikzpicture}[>=latex']
    
    \def\aw{0.5}
    \def\bw{1.5}
    \def\cw{2}
    
    \def\Xlim{6}
    \def\Ylim{4}
    \def\domainX{-\Xlim:\Xlim}
    \def\rangeY{-\Ylim+1:\Ylim}
    \def\scaleXY{0.4}
    
    \def\domainXa{-\Xlim:-\cw}
    \def\domainXb{-\cw:\cw}
    \def\domainXc{\cw:\Xlim}
    
    \begin{scope}[range=-\rangeY, scale=\scaleXY]
		\draw[very thin, color=gray!40, domain=\domainX] (-\Xlim,-\Ylim+1) grid (\Xlim,\Ylim);
    \end{scope}

	\draw[->,thick, scale=\scaleXY](-\Xlim-0.2,0) -- (\Xlim+0.2,0) node[below left=0.25, circle, inner sep=-0.4, align=center] {$\su$} coordinate (cMidRight);
	\draw[->,thick, scale=\scaleXY](0,-\Ylim+1-0.2) -- (0,\Ylim+0.2) node[below left] {$\sy$};

    
    \begin{scope}[range=\rangeY, scale=\scaleXY, very thick]
    	
    	\draw[domain=\domainXa, color=matlabCol1] plot[id=hw_l, samples=100] function{\aw*x + \bw*(x + \cw) - \bw*tanh(x + \cw)};
    	\draw[domain=\domainXb, color=matlabCol1] plot[id=hw_m] function{\aw*x};
    	\draw[domain=\domainXc, color=matlabCol1] plot[id=hw_r, samples=100] function{\aw*x + \bw*(x - \cw) - \bw*tanh(x - \cw)} node[right=0.15, rectangle, fill=white!50, inner sep=0.2] {$\shNLweight(\su)$};
    	
    	\draw[domain=\domainXa, color=matlabCol2] plot[id=dhw_l] function{\aw + \bw*tanh(x + \cw)*tanh(x + \cw)};
    	\draw[domain=\domainXb, color=matlabCol2] plot[id=dhw_m] function{\aw};
    	\draw[domain=\domainXc, color=matlabCol2] plot[id=dhw_r] function{\aw + \bw*tanh(x - \cw)*tanh(x - \cw)} node[right] {$\dFull{\shNLweight(\su)}{\suNLweight}$};
    \end{scope}
\end{tikzpicture}%
	}
	\caption{Example of the weighting function $\shNLweight$ \eqref{eq:Control:NL_Func} and its derivative \eqref{eq:Passivity:Ctrl:NL_der} on a unit grid, with $\saNLweight = 0.5$, $\sbNLweight = 1.5$ and $\scNLweight = 2$.}
	\label{fig:Control:NL_Func}
\end{figure}
To allow for a better utilisation of the tolerance band around $\svRef$, we desire a weighting function that assigns a low gain for errors within the tolerance band and a high gain for larger errors. We therefore define the class \classC{1} function $\syNLweight[k] = \shNLweight(\suNLweight[k])$ conforming to \eqref{eq:Prelim:Static_function}, where
\begin{align} \label{eq:Control:NL_Func}
	&\shNLweight(\su) \coloneqq \saNLweight \su + \sbNLweight \sgNLweight(\su) - \sbNLweight\tanh(\sgNLweight(\su)), \\
	\label{eq:Control:NL_input_func}
	&\sgNLweight(\su) \coloneqq \left\{ \begin{Array}{lr@{\,\,}l@{\,\,}l}
		\su + \scNLweight, \quad && \su &< -\scNLweight \\
		0, \quad & -\scNLweight \le& \su &\le \scNLweight \\
		\su - \scNLweight, \quad &  \scNLweight <& \su &
	\end{Array}\right.
\end{align}
and where \eqref{eq:Control:NL_input_func} describes a dead-zone parametrised by $\scNLweight$. An example of \eqref{eq:Control:NL_Func} is depicted in \autoref{fig:Control:NL_Func} along with its derivative. For a strictly increasing function as per \autoref{obj:Problem:weighted_errors}, set $\cramped{\saNLweight > 0}$ and $\sbNLweight > -\saNLweight$.
\subsection{Equilibrium Analysis} \label{sec:Control:Steady_State}
In a first step towards analysing the closed loop, we analyse the assumed equilibrium of the interconnected microgrid and four-stage controller (see \autoref{ass:Problem:Feasible_network}). Specifically, we verify that the proposed control yields an equilibrium which satisfies \autorefMulti{obj:Problem:weighted_errors, obj:Problem:power_sharing}.
\begin{proposition}[Controller equilibrium analysis] \label{prop:Control:desired_equilibrium}
	Consider the DC microgrid comprising \eqref{eq:Problem:node_unact_dynamics}, \eqref{eq:Problem:line_dynamics}, and \eqref{eq:Problem:Controlled_DGU} which is connected in feedback with the four-stage controller comprising \eqref{eq:Control:DDA}, \eqref{eq:Control:local_PI_vector}, and \eqref{eq:Control:NL_Func} as in \autoref{fig:Control:control_structure}. Let \autorefMulti{ass:Problem:Feasible_network, ass:Problem:Actuated_agents, ass:Problem:Connected_graphs} hold. Then, \autoref{obj:Problem:power_sharing} is met for the equilibrium imposed by the control structure. Moreover, \autoref{obj:Problem:weighted_errors} is achieved exactly for ideal integrators $\sdampCtrl = 0$ in \eqref{eq:Control:local_PI_vector}. For lossy integrators with $\sdampCtrl > 0$, the remaining error for \autoref{obj:Problem:weighted_errors} is be described by the steady-state value of $\vyDDA[2]$, where
	\begin{equation} \label{eq:Control:Steady_state_error_output}
		\vyDDA[2] = \frac{\sdampCtrl}{\skiCtrl (1 + \sdampCtrl \skpCtrl)} \vyDDA[4] .
	\end{equation}
\end{proposition}
The proof of \autoref{prop:Control:desired_equilibrium} can be found in \autorefapp{sec:App_Proofs}. Through \autoref{prop:Control:desired_equilibrium} we thus confirm that the proposed controller yields an equilibrium which meets the requirements, even though the requirements are not perfectly met when leaky agent PI controllers are used. We also note that \autoref{prop:Control:desired_equilibrium} only considers the controlled microgrid already in equilibrium and does not consider the convergence to the equilibrium.
\begin{remark}[Compensating leaky-integral errors] \label{rem:Passive:Compensate_leaky_offset}
	As indicated by \eqref{eq:Control:Steady_state_error_output} in \autoref{prop:Control:desired_equilibrium}, the leaky agent PI controllers result in a constant steady-state error for the average voltage regulation (\autoref{obj:Problem:weighted_errors}). Since a positive $\vyDDA[2]$ corresponds to voltages below the desired $\svRef$, it follows that setting $\svRef$ above the actual desired voltage reference will result in higher bus voltages. Changing $\svRef$ thus allows the steady-state effects of the leaky integrators to be compensated.
	Moreover, notice that $\vyDDA[4]$ is the controller output, i.e.\ the power setpoint $\vpSP$ used for the \acp{DGU} (see \autoref{fig:Control:control_structure}). Thus, the error measure in \eqref{eq:Control:Steady_state_error_output}, which is only dependent on the controller output, can be used to determine the offset to $\svRef$ for exact voltage regulation.
	Note, however, that modifying $\svRef$ based on $\vpSP$ results in a new loop which requires an additional stability analysis.
\end{remark}
	\section{Subsystem Passivity Analysis} \label{sec:Passivity}
Having verified whether the desirable steady state is achieved by the controller, we now set about analysing the convergence to this steady state. With the aim of applying \autoref{thm:Prelim:Min_restrictive_diss} for the closed-loop stability, we first analyse the passivity properties of the individual subsystems. Since the steady-state bus voltages $\sveq[k]$ are unknown and non-zero, we investigate the passivity properties shifted to any plausible point of operation using \ac{EIP}. To this end, we construct an \ac{EIP} formulation for the DC microgrid from its constitutive elements in \autoref{sec:Passivity:MG}. This is followed by the respective analyses of the various controller stages in \autoref{sec:Passivity:Ctrl}. Note that we omit the bus indices $k$ and $l$ in this section where clear from context.
\subsection{DC Microgrid Passivity} \label{sec:Passivity:MG}
For the stability of the microgrid at the equilibrium $\vveq$, we desire an \ac{EIP} property relating the shifted input power setpoints $\vpeSP = \vpSP - \vpeqSP$ to the output voltage errors $\vve = \vv - \vveq$ of all nodes, since this port $(\vpeSP, \vve)$ is used by the controller in \autoref{fig:Control:control_structure}.
To this end, we derive \ac{EIP} properties for the load, \ac{DGU} and line subsystems of the microgrid, making sure to shift the subsystem dynamics to the assumed equilibrium in each case (see \autoref{ass:Problem:Feasible_network}). Thereafter, we combine the results of these subsystems, to construct an \ac{EIP} property for the microgrid as a whole. Where applicable, an analysis of the zero-state dynamics is performed to ensure the eventual stability of the controlled microgrid.
\subsubsection{Load Passivity} \label{sec:Passivity:MG:Loads}
Let the unactuated bus dynamics in \eqref{eq:Problem:node_unact_dynamics} for the buses in $\setNUnact$ be shifted to the equilibrium $(\vieqTx, \sveq)$, yielding
\begin{equation} \label{eq:Problem:Shited_load_dynamics}
	\sCeq\svedot = -\vEPT[k]\vieTx - \sILeve + (\vEPT[k]\vieqTx + \sILveq) ,
\end{equation}
for the static load function shifted according to \eqref{eq:Prelim:Static_function}. In \eqref{eq:Problem:Shited_load_dynamics}, $\vEPT[k]\vieqTx = -\sILveq$ since the load is fully supplied by the cumulative line currents in steady state. 
\begin{proposition}[Load \ac{EIP}] \label{prop:Passivity:Load_passivity}
	The shifted load dynamics in \eqref{eq:Problem:Shited_load_dynamics} are \ac{OFP}$(\siOutLoad)$ w.r.t.\ the input-output pair $(-\vEPT[k]\vieTx,\sve)$ with $\siOutLoad = \scLoadlo$ the smallest gradient of the static load function $\sILv$.
\end{proposition}
\begin{proof}
	Consider the storage function $\sSLoad$ along with its time derivative
	\begin{align} \label{eq:Passive:load_storage}
		\sSLoad &= \frac{\sCeq}{2} \sve^2, \\
		\label{eq:Passive:load_Hdot}
		\sSLoaddot &= -\sve\vEPT[k]\vieTx - \sve\sILeve .
	\end{align}
	Since the static load function $\sILv$ is \ac{IFOFP} according to \autoref{prop:Prelim:Static_load}, it is bounded from below by $\scLoadlo \sve^2 \le \sve\sILeve$ (see \eqref{eq:Prelim:Load_sector_alt}). Incorporate this lower bound into \eqref{eq:Passive:load_Hdot} to obtain
	\begin{equation} \label{eq:Passive:load_Hdot_ineq}
		\sSLoaddot \le \swLoad \coloneqq -\sve\vEPT[k]\vieTx - \scLoadlo \sve^2
	\end{equation}
	which yields the \ac{OFP} property from \autoref{def:Prelim:passive_rates}.
\end{proof}
\begin{remark}[ZIP load passivity] \label{rem:Passivity:load_measure}
	\autoref{prop:Passivity:Load_passivity} and \eqref{eq:Prelim:Load_gradient} demonstrate that the passivity properties of the unactuated buses are directly linked to the smallest gradient of the load function. For the ZIP load in \eqref{eq:Problem:ZIP_model}, this yields
	\begin{equation} \label{eq:Passive:ZIP_passivity}
		\scLoadlo = \min\left(\sZinv, \, \sZinv - \frac{\sP}{\svCrit^2}, \, \sZCritinv\right).
	\end{equation}
	Considering the strictly passive case $(\scLoadlo = 0)$ along with $\sI, \sP \ge 0$ yields the passivity condition $\sZinv \svCrit^2 \ge \sP$ frequently used in the literature \cite{DePersis2017, Fan2019, Strehle2020DC, Nahata2020, Cucuzzella2019pbc}. 
\end{remark}
\subsubsection{DGU Passivity} \label{sec:Passivity:MG:DGU}
Shift the states $(\se, \sii, \sv)$ and inputs $(\spSP, \viTx)$ of the DGU dynamics in \eqref{eq:Problem:Controlled_DGU} for the buses in $\setNAct$ to the respective error variables $(\see, \sie, \sve)$ and $(\speSP, \vieTx)$ to obtain \eqref{eq:Passivity:Shifted_DGU} on the next page,
\begin{figure*}[ht!]
	\addtocounter{figure}{-1}
	\begin{equation} \label{eq:Passivity:Shifted_DGU}
			\!\begin{bmatrix}
				\seDGUedot \\ \sL \siedot \\ \sCeq \svedot
			\end{bmatrix} \!\! =  \underbrace{ \!\!\begin{bmatrix}
					0 & -\sv & -\sieq \\
					\skiDGU & \sRDGUdamp \!-\! \sR \!-\! \skpDGU \sv & -1 \!-\! \skpCtrl \!\: \sieq \\
					0 & 1 & - \frac{\sILeve}{\sve}
				\end{bmatrix}\!\!}_{\textstyle \mADGU(\sv, \sieq, \frac{\sILeve}{\sve})} \, \underbrace{\!\!\begin{bmatrix}
					\seDGUe \\ \sie \\ \sve
				\end{bmatrix}\!\!}_{\textstyle \vxeDGU} + \, \underbrace{\!\!\begin{bmatrix}
					1 \\
					\skpDGU \\
					0 
				\end{bmatrix}\!\!}_{\textstyle \vbDGU[1]} \, \speSP 
			{\,-} \underbrace{\!\!\begin{bmatrix}
					0 \\ 0 \\ 1
				\end{bmatrix}\!\!}_{\textstyle \vbDGU[2]} \! \vEPT[k] \vieTx
			+ \underbrace{\!\!\begin{bmatrix}
					\speqSP - \sveq \sieq \\ \skiCtrl \seDGUeq + (\sRDGUdamp \!-\! \sR) \, \sieq - \sveq + \svRef - \skpCtrl (\speqSP \:\!-\!\: \sveq \sieq) \\ \sieq -\vEPT[k]\vieqTx -\sILveq
				\end{bmatrix}\!\!}_{\textstyle \vdistDGU}
	\end{equation}
	\hrulefill
\end{figure*}
where the static load function is incorporated into the matrix $\mADGU$.
Furthermore, the measured power $\spp = \sv \sii = \sv (\sie + \sieq)$ in \eqref{eq:Problem:DGU_Power_Control} is left partially in unshifted variables such that $\mADGU$ is also dependent on the unshifted voltage $\sv$ and the steady-state current $\sieq$.

Note that the constant $\vdistDGU$ in \eqref{eq:Passivity:Shifted_DGU} is found by setting the error variables $(\speSP, \vieTx, \seDGUe, \sie, \sve)$ and their time derivatives to zero. As such, the constant $\vdistDGU \equiv 0$ can be disregarded in the passivity analysis. We now analyse the shifted nonlinear system in \eqref{eq:Passivity:Shifted_DGU} for \ac{EIP}.
\begin{theorem}[\ac{EIP} \acp{DGU}] \label{thm:Passivity:Shifted_DGU}
	The shifted DGU dynamics in \eqref{eq:Passivity:Shifted_DGU} are simultaneously \ac{IFOFP}$(\siInDGU[1], \siOutDGU)$ w.r.t.\ the input-output pair $(\speSP, \sve)$ and \ac{IFP}$(\siInDGU[2])$ w.r.t.\ the input-output pair $(-\vET[k]\vieTx, \sve)$, if a feasible solution can be found for
	%
	\begin{equation} \label{eq:Passivity:Opt_DGU_Passivity}
		\begin{array}{@{}c@{\,}l@{}}
			\!\displaystyle \max_{\mPDGU,\, \siInDGU[1],\, \siInDGU[2],\, \siOutDGU} & \siInDGU[1] + \siInDGU[2] + \siOutDGU \\
			\text{s.t.} & \eqref{eq:Passivity:Opt_DGU_Passivity_condition} \,\, \text{holds}\,\, \forall \, \sv \in \setV \subseteq \RealsPos, \forall \, \sieq \in \setIeq \subseteq \Reals
		\end{array}
	\end{equation}
	\begin{figure*}[ht!]
		\addtocounter{figure}{-1}
		\begin{equation} \label{eq:Passivity:Opt_DGU_Passivity_condition}
			\begin{aligned}
				\begin{bmatrix}
					\mQDGU(\sv, \sieq, \scLoadlo) + \siOutDGU\vcDGU\vcDGUT & \mPDGU \vbDGU[1] - \frac{1+\siInDGU[1]\siOutDGU}{2} \vcDGU & \mPDGU \vbDGU[2] - \frac{1}{2} \vcDGU \\
					\vbDGUT[1] \mPDGU - \frac{1+\siInDGU[1]\siOutDGU}{2} \vcDGUT & \siInDGU[1] & 0 \\
					\vbDGUT[2] \mPDGU - \frac{1}{2} \vcDGUT & 0 & \siInDGU[2]
				\end{bmatrix} \negDef 0, \qquad \mPDGU \posDef 0
			\end{aligned}
		\end{equation}
		\hrulefill
	\end{figure*}
	where $\mQDGU(\sv, \sieq, \scLoadlo) \coloneqq \mPDGU \mADGU(\sv, \sieq, \scLoadlo) + \mADGUT(\sv, \sieq, \scLoadlo) \mPDGU$,
	\begin{equation} \label{eq:Passivity:DGU_mod_A}
		\mADGU(\sv, \sieq, \scLoadlo) = \begin{bmatrix}
			0 & -\sv & -\sieq \\
			\skiDGU & \sRDGUdamp - \sR - \skpDGU \sv & -1 - \skpCtrl \sieq\\
			0 & 1 & -\scLoadlo
		\end{bmatrix} ,
	\end{equation}
	and with $\siInDGU[1], \siInDGU[2], \siOutDGU \in \Reals$, $\scLoadlo$ as in \eqref{eq:Prelim:Load_gradient} and $\vcDGU = [0,0,1]^\Transpose$.
\end{theorem}
\begin{proof}
	Consider for \eqref{eq:Passivity:Shifted_DGU} the storage function 
	\begin{equation} \label{eq:Passivity:DGU_Storage}
		\sSDGU = \begin{bmatrix}
			\seDGUe \\ \sie \\ \sve
		\end{bmatrix}^\Transpose \!\!\!
		\mPDGU \!
		\begin{bmatrix}
			\seDGUe \\ \sL \sie \\ \sCeq \sve
		\end{bmatrix} ,
	\end{equation}
	with $\mPDGU \posDef 0$. The time derivative of \eqref{eq:Passivity:DGU_Storage} is
	\begin{equation} \label{eq:Passivity:DGU_Storage_dot}
		\sSDGUdot \!=\!\! \begin{bmatrix}
			\vxeDGU \\ \speSP \\ \vEPT[k] \vieTx
		\end{bmatrix}^{\!\!\Transpose} \!\! \begin{bmatrix}
		\mQDGU(\sv, \sieq, \frac{\sILeve}{\sve}) & \mPDGU \vbDGU[1] & \mPDGU \vbDGU[2] \\
		\vbDGUT[1] \mPDGU & 0 & 0 \\
		\vbDGUT[2] \mPDGU & 0 & 0
		\end{bmatrix}\!\!\! \begin{bmatrix}
			\vxeDGU \\ \speSP \\ \vEPT[k] \vieTx
		\end{bmatrix}\! ,
	\end{equation}
	with $\vxeDGU$ as in \eqref{eq:Passivity:Shifted_DGU}.
	Since it follows from \eqref{eq:Prelim:Load_sector_alt} that $-\sve \sILeve \le -\scLoadlo\sve^2$, this bound can be incorporated into the inequality
	\begin{equation} \label{eq:Passivity:DGU_Storage_dot_load_bound}
		\sSDGUdot \!\le\!\! \begin{bmatrix}
			\vxeDGU \\ \speSP \\ \vEPT[k] \vieTx
		\end{bmatrix}^{\!\!\Transpose} \!\! \begin{bmatrix}
			\mQDGU(\sv, \sieq, \scLoadlo) & \mPDGU \vbDGU[1] & \mPDGU \vbDGU[2] \\
			\vbDGUT[1] \mPDGU & 0 & 0 \\
			\vbDGUT[2] \mPDGU & 0 & 0
		\end{bmatrix}\!\!\! \begin{bmatrix}
			\vxeDGU \\ \speSP \\ \vEPT[k] \vieTx
		\end{bmatrix}\! .
	\end{equation}
	The desired \ac{IFOFP} and \ac{IFP} properties for the \ac{DGU} are described by the supply rate
	\begin{equation} \label{eq:Passivity:DGU_Supply_rate}
		\begin{aligned}
			\swDGU = \;& (1 + \siInDGU[1]\siOutDGU)\speSP\sve - \siInDGU[1](\speSP)^2 - \siOutDGU\sve^2 \\
			&-  \sve \vEPT[k]\vieTx - \siInDGU[2]\left(\vEPT[k]\vieTx\right)^2
		\end{aligned}
	\end{equation}
	These properties are guaranteed, if $\sSDGUdot - \swDGU < 0$ for all valid inputs and outputs and for $\sv \in \setV$ and $\sieq \in \setIeq$. Combining \eqref{eq:Passivity:DGU_Storage_dot_load_bound} and \eqref{eq:Passivity:DGU_Supply_rate} in this manner directly leads to constraint \eqref{eq:Passivity:Opt_DGU_Passivity_condition} in \eqref{eq:Passivity:Opt_DGU_Passivity}. Finally, the objective function in \eqref{eq:Passivity:Opt_DGU_Passivity} seeks to find the largest indices for which the constraints are satisfied in a similar manner to \autoref{thm:Prelim:Min_restrictive_diss}.
\end{proof}
Although \autoref{thm:Passivity:Shifted_DGU} demonstrates the \ac{EIP} of the actuated buses, notice that the $\seDGUe$ and $\sieq$ of \eqref{eq:Passivity:Shifted_DGU} are not included in the supply rate $\swDGU$ in \eqref{eq:Passivity:DGU_Supply_rate}. As such, an investigation of the zero state dynamics of the \ac{DGU} is required.
\begin{proposition}[\ac{ZSO} \acp{DGU}] \label{prop:Passivity:DGU_zero_state}
	The shifted \ac{DGU} dynamics in \eqref{eq:Passivity:Shifted_DGU} are \ac{ZSO}.
\end{proposition}
\begin{proof}
	In \eqref{eq:Passivity:Shifted_DGU}, set the inputs $\speSP \equiv 0$, $\vieTx \equiv 0$ and the output $\sve \equiv 0$. Since $\vdistDGU = 0$ and $\sILe(0) = 0$, verify from the equation for $\svedot$ that $\sie \equiv 0$. From the equation for $\siedot$, it then follows that $\seDGUe \equiv 0$ which concludes this proof.
\end{proof}
\begin{remark}[Compensating non-passive loads] \label{rem:Passivity:Non_passive_loads}
	As demonstrated in \cite{Cucuzzella2023}, adding a term dependent on $\svdot[k]$ to the regulator output $\svVSC[k]$ in \eqref{eq:Problem:DGU_Power_Control} allows for damping to be added to the unactuated state $\sv[k]$. This in turn allows for regulation in the presence of non-passive loads and can yield more favourable passivity indices when applying \autoref{thm:Passivity:Shifted_DGU}.
\end{remark}
\subsubsection{Line Passivity} \label{sec:Passivity:MG:Lines}
The dynamics of the line subsystem \eqref{eq:Problem:line_dynamics} shifted to the equilibrium $(\sieqTx, \vveq)$ yield
\begin{equation} \label{eq:Passivity:Shifted_line}
	\sL[kl]\sieTxdot = -\sR[kl]\sieTx + \vEPT[kl]\vve ,
\end{equation}
which can now be analysed for passivity.
\begin{proposition}[\ac{OFP} lines] \label{prop:Passivity:Line_Passivity}
	The shifted line dynamics in \eqref{eq:Passivity:Shifted_line} are \ac{OFP}$(\siOutTx)$ with $\siOutTx = \sR[kl]$ w.r.t.\ the input-output pair $(\vEPT[kl]\vve, \sieTx)$ with the storage function
	\begin{equation} \label{eq:Passivity:Line_Storage}
		\sSTx = \frac{\sL[kl]}{2}\sieTx^2 .
	\end{equation}
\end{proposition}
\begin{proof}
	The proof follows trivially by verifying that 
	\begin{equation} \label{eq:Passivity:Line_Hdot}
		\sSTxdot = \sieTx\vEPT[kl]\vve -\sR[kl]\sieTx^2 \eqqcolon \swTx ,
	\end{equation}
	where $\swTx$ in an \ac{OFP} supply rate as per \autoref{def:Prelim:passive_rates}.
\end{proof}
\subsubsection{Interconnected Microgrid Dissipativity} \label{sec:Passivity:MG:Interconnected}
Having separately analysed the subsystems comprising the microgrid, we now combine the results to formulate the dissipativity of the full microgrid w.r.t.\ the input-output pair $(\vpeSP, \vve)$. For simplicity, we group the buses according to their actuation states \eqref{eq:Problem:actuation}. Thus, $\vpeSP = [\vpeSPActT, \vpeSPUnactT]^\Transpose$ and $\vve = [\vveActT, \vveUnactT]^\Transpose$ have the same dimensions. Note that we include the inputs $\vpeSPUnact$ for the unactuated buses in $\setNUnact$ as provided by the four-stage controller (see \autoref{fig:Control:control_structure}), even though these inputs are not used.
\begin{proposition}[Microgrid dissipativity] \label{prop:Passivity:MG_Supply}
	A DC microgrid comprising DGUs \eqref{eq:Problem:Controlled_DGU}, lines \eqref{eq:Problem:line_dynamics} and loads \eqref{eq:Problem:node_unact_dynamics} with an interconnection topology described by a connected graph $\graphGPhys$ is dissipative w.r.t.\ the supply rate
	\begin{equation} \label{eq:Passive:MG_supply_rate}
		\begin{aligned}
			\swMGdep =& \,  (1+\siInDGU[1]\siOutDGU)\vpeSPActT \vveAct - \siInDGU[1]\vpeSPActT \vpeSPAct  \\
			& - \siOutDGU\vveActT\vveAct - \siOutLoad \vveUnactT\vveUnact ,
		\end{aligned}
	\end{equation}
	if $\cramped{\siInDGU[2] +\siOutTx \ge 0}$ for the worst-case indices of the buses and lines calculated in \autoref{prop:Passivity:Load_passivity} $(\siOutLoad)$, \autoref{prop:Passivity:Line_Passivity} $(\siOutTx)$, and \autoref{thm:Passivity:Shifted_DGU} $(\siInDGU[1], \siInDGU[2], \siOutDGU)$, i.e.\
	\begin{equation} \label{eq:Passive:worst_case_indices}
		\begin{aligned}
			\siInDGU[1] \!&=\! \min_{k \in \setNAct} \siInDGU[1, k], \!\! &
			\siInDGU[2] \!&=\! \min_{k \in \setNAct} \siInDGU[2, k], \!\! &
			\siOutDGU \!&=\! \min_{k \in \setNAct} \siOutDGU[k], \!\! \\
			\siOutLoad \!&=\! \min_{k \in \setNUnact} \siOutLoad[k], \!\! &
			\siOutTx \!&=\! \min_{kl \in \graphEPhys} \siOutTx[k] .
		\end{aligned}
	\end{equation}
\end{proposition}
\begin{proof}
	Define for the interconnected microgrid the storage function
	\begin{equation} \label{eq:Passive:MG_Storage}
		\sSMG = \sum_{k \in \setNAct} \sSDGU[k] + \sum_{k \in \setNUnact} \sSLoad[k] + \sum_{kl \in \graphEPhys} \sSTx[kl].
	\end{equation}
	An upper bound for time derivative of \eqref{eq:Passive:MG_Storage} may then be found by combining the supply rates in \eqref{eq:Passive:load_Hdot_ineq}, \eqref{eq:Passivity:DGU_Supply_rate} and \eqref{eq:Passivity:Line_Hdot}
	\begin{equation} \label{eq:Passive:MG_Hdot}
		\begin{aligned}
			\sSMGdot \le& \, (1+\siInDGU[1]\siOutDGU)\vpeSPActT \vveAct - \siInDGU[1]\vpeSPActT\vpeSPAct - \siOutDGU\vveActT\vveAct \\
			& +  \vieTxT\mET\vve - \vveActT \mEAct\vieTx - \vveUnactT\mEUnact\vieTx \\
			&- \siOutLoad \vveUnactT\vveUnact - (\siInDGU[2] +\siOutTx)\vieTxT\vieTx  .
		\end{aligned}
	\end{equation}
	The skew-symmetric interconnection of the nodes and lines results in $\vieTxT\mET\vve = \vveActT \mEAct\vieTx + \vveUnactT\mEUnact\vieTx$. Furthermore with $\siInDGU[2] +\siOutTx \ge 0$, we can drop the unnecessary strictly negative $\vieTxT\vieTx$ term and verify that $\cramped{\sSMGdot \le \swMGdep}$.
\end{proof}
Through \autoref{prop:Passivity:MG_Supply}, the dissipativity of the entire microgrid is formulated using the desired input and output vectors. However, the supply rate in \eqref{eq:Passive:MG_supply_rate} is dependent on the actuation states of the buses. We now remove this dependence by finding a supply rate for a specific bus that encompasses both its actuated and unactuated state. By considering a quadratic supply rate as a sector condition (see \cite{Khalil2002, Malan2022b}), a combined supply rate is found through the union of the sectors for the actuated and unactuated cases.
\begin{theorem}[Actuation independent passivity] \label{thm:Passive:MG_indep_passive}
	A DC microgrid for which \autoref{prop:Passivity:MG_Supply} holds is \ac{IFOFP}$(\siInDGU[1], \siOutDGU)$ w.r.t.\ the supply rate
	\begin{equation} \label{eq:Passive:MG_IFOFP}
		\swMG = (1 + \siInDGU[1]\siOutDGU)\vpeSPT \vve - \siInDGU[1] \vpeSPT \vpeSP - \siOutDGU \vveT \vve
	\end{equation}
	if, for an arbitrarily small $\siInLoad > 0$,
	\begin{align} 
		\label{eq:Passive:MG_lines_dominate_DGU}
		0 &\le \siInDGU[2] +\siOutTx, \\
		\label{eq:Passive:MG_Unact_passive_loads}
		0 &< \siOutLoad < 1, \\
		\label{eq:Passive:MG_Act_iIn_passivity}
		0 &> \siInDGU[1].
	\end{align}
\end{theorem}

The proof of \autoref{thm:Passive:MG_indep_passive} can be found in \autorefapp{sec:App_Proofs}.
Through \eqref{eq:Passive:MG_IFOFP}, we thus show that a single \ac{IFOFP} supply rate describes the input-output passivity of the entire microgrid, irrespective of the states of actuation of the buses. This supply rate is derived from the properties of the \acp{DGU} in \autoref{thm:Passivity:Shifted_DGU} and accounts for the worst-case loads.
\begin{remark}[Non-passive loads at \acp{DGU}] \label{rem:Passive:non_passive_loads_act}
	While \eqref{eq:Passive:MG_Unact_passive_loads} in \autoref{thm:Passive:MG_indep_passive} requires strictly passive loads at unactuated buses, this is not required for the loads at actuated buses. Indeed, the loads at \acp{DGU} may exhibit a lack of passivity with $\scLoadlo < 0$. However, this would be reflected by the indices obtained in \autoref{thm:Passivity:Shifted_DGU} and the supply rate in \eqref{eq:Passive:MG_IFOFP}.
\end{remark}
\begin{remark}[Non-static loads] \label{rem:Passive:non_static_loads}
	Due to the use of passivity in this section, the analysis presented here effortlessly extends to the case of dynamic loads. Such dynamic loads simply need to exhibit equivalent \ac{IFP} properties (see e.g.\ \autoref{prop:Passivity:Load_passivity}) and must be \ac{ZSO}.
\end{remark}
\begin{remark}[Passivity-based controllers] \label{rem:Passive:passivity_based_control}
	In addition to the four-stage controller proposed in this work, the passivity formulation of the DC microgrid in \autoref{thm:Passive:MG_indep_passive} can be used alongside any other controller which provides suitable passivity indices. This includes methods such as \emph{interconnection and damping assignment passivity-based control} \cite[p.~190]{vdSchaft2017} or \emph{passivity-based model predictive control} (see e.g.\ \cite{Raff2007}).
\end{remark}
%
%
\subsection{Controller Passivity} \label{sec:Passivity:Ctrl}
Having analysed the passivity of the microgrid subsystems and their interconnection, we now investigate the passivity properties of the control structure in \autoref{sec:Control}. This is done successively for each part of the controller: the DDA stages, the PI stage and the weighting function.
\subsubsection{DDA Passivity} \label{sec:Passivity:Ctrl:DDA}
Consider the \ac{DDA} stages in \autoref{fig:Control:control_structure}.
\begin{proposition}[\ac{DDA} Passivity] \label{prop:Passivity:Ctrl:DDA_Passive}
	The DDA controller in \eqref{eq:Control:DDA} with the storage function
	\begin{equation} \label{eq:Passivity:Ctrl:H_DDA}
		\sSDDA[s] = \frac{1}{2\sgDDA} \left(\vxDDA[s]^\Transpose \vxDDA[s] + \vzDDA[s]^\Transpose \vzDDA[s]\right)
	\end{equation}
	is \ac{OFP}$(\siOutDDA)$, $\siOutDDA = 1$, w.r.t.\ $(\vuDDA[s], \vyDDA[s])$ and is \ac{ZSO}.
\end{proposition}
\begin{proof}
	The time derivative of \eqref{eq:Passivity:Ctrl:H_DDA} is
	\begin{equation} \label{eq:Passivity:Ctrl:H_DDA_dot}
		\begin{aligned}
			\sSDDAdot[s] &= - \vxDDA[s]^\Transpose \vxDDA[s] - \frac{1}{\sgDDA} \vxDDA[s]^\Transpose \mLapComm[P]\vxDDA[s] + \vxDDA[s]^\Transpose \vuDDA[s] \\
			&\le \swDDA[s] \coloneqq \vxDDA[s]^\Transpose \vuDDA[s] - \vxDDA[s]^\Transpose \vxDDA[s]
		\end{aligned}
	\end{equation}
	since $\mLapComm[P] > 0$ and $\sgDDA > 0$, thus verifying the \ac{OFP} property for $\vyDDA[s] = \vxDDA[s]$. Furthermore, the \ac{DDA} controller is \ac{ZSO} since the system dynamics in \eqref{eq:Control:DDA} is Hurwitz \cite[Theorem~5]{Freeman2006}.
\end{proof}
The \ac{OFP} result in \autoref{prop:Passivity:Ctrl:DDA_Passive} also means that \eqref{eq:Control:DDA} has an \Ltwo-gain of 1 \cite[p.~3]{Arcak2016}. Note that since the \ac{DDA} in \eqref{eq:Control:DDA} is linear, the properties in \autoref{prop:Passivity:Ctrl:DDA_Passive} also hold for the shifted input-output combination $(\vuDDAe[s], \vyDDAe[s])$ \cite[p.~26]{Arcak2016}.
\subsubsection{PI Passivity} \label{sec:Passivity:Ctrl:DDA_PI}
The ideal PI controller in \eqref{eq:Control:local_PI_vector} with $\sdampCtrl = 0$ can trivially be shown to be \ac{IFP}$(\skpCtrl)$ for the storage function $\cramped{\sSCtrl = \skiCtrl \vxCtrlT \vxCtrl / 2}$. The leaky PI control with $\sdampCtrl > 0$ exhibits the following properties.
\begin{proposition}[Leaky PI Passivity] \label{prop:Passivity:Leaky_PI}
	The leaky PI control in \eqref{eq:Control:local_PI_vector} with the storage function $\cramped{\sSCtrl = \skiCtrl \vxCtrlT \vxCtrl / 2}$ is dissipative w.r.t.\ the supply rate
	\begin{equation} \label{eq:Passivity:Leaky_PI_supply}
		\begin{aligned}
			\swCtrl \, {=} \,
				\underbrace{\!\!\left(\!1 {+} \frac{2 \sdampCtrl \skpCtrl}{\skiCtrl}\! \right)\!\!}_{\textstyle 2\siCrossCtrl} \, \vuCtrlT \vyCtrl 
				- \, \underbrace{\!\!\left(\!\skpCtrl {+} \frac{{\sdampCtrl} {\skpCtrl}^2}{\skiCtrl}\! \right)\!\!}_{\textstyle \siInCtrl} \, \vuCtrlT \vuCtrl 
				- \!\! \underbrace{\!\!\!\frac{\sdampCtrl}{\skiCtrl}\!\!\!}_{\textstyle \siOutCtrl} \! \vyCtrlT \vyCtrl
		\end{aligned}
	\end{equation}
\end{proposition}
\begin{proof}
	Calculate the time derivative of $\sSCtrl$ as $\sSCtrldot = \skiCtrl \vxCtrlT \vuCtrl - \sdampCtrl \skiCtrl \vxCtrlT \vxCtrl$. Substitute in $\skiCtrl \vxCtrl = \vyCtrl - \skpCtrl \vuCtrl$ from the output in \eqref{eq:Control:local_PI_vector} and simplify to verify that $\sSCtrldot = \swCtrl$.
\end{proof}
Note that while $\swCtrl$ in \eqref{eq:Passivity:Leaky_PI_supply} has a quadratic form, it does not directly match the \ac{IFOFP} form in \autoref{def:Prelim:passive_rates}. However, by appropriately weighing the storage function $\sSCtrl$, the form in \autoref{def:Prelim:passive_rates} is easily obtained. For simplicity and without invalidating the results in the sequel, we omit this step here.
Furthermore, we note that the linearity of \eqref{eq:Control:local_PI_vector} ensures that the properties in \autoref{prop:Passivity:Leaky_PI} also hold for the shifted input-output combination $(\vuCtrle, \vyCtrle)$ \cite[p.~26]{Arcak2016}.
\subsubsection{Weighting Function Passivity} \label{sec:Passivity:Ctrl:Weighting}
The derivative of the weighting function in \eqref{eq:Control:NL_Func} is described by (see e.g.\ \autoref{fig:Control:NL_Func})
\begin{align} \label{eq:Passivity:Ctrl:NL_der}
	\dFull{\syNLweight}{\suNLweight} = \saNLweight + \sbNLweight\tanh^2(\sgNLweight(\suNLweight)) .
\end{align}
By setting $\sbNLweight > -\saNLweight$ and applying \autoref{prop:Prelim:Static_load}, \eqref{eq:Control:NL_Func} is found to be \ac{IFOFP}$(\siInNLweight, \siOutNLweight)$ with
\begin{equation} \label{eq:Passivity:Ctrl:Weighting_indices}
	\siInNLweight = \saNLweight, \qquad \siOutNLweight = \frac{1}{\saNLweight + \sbNLweight},
\end{equation}
	\section{Interconnected Stability} \label{sec:Stability}
Using the passivity properties of the microgrid and controller subsystems obtained in \autoref{sec:Passivity}, we now investigate the stability of the microgrid and controller interconnected as in \autoref{fig:Control:control_structure}. However, we note that the agent PI controller and the Stage~4 DDA controller exhibit a cascaded \ac{IFP}-\ac{OFP} obstacle (see \autoref{prop:Prelim:non_diss_IFP_OFP}) if the PI controller is ideal ($\sdampCtrl = 0$) which prevents a closed-loop analysis with dissipativity. Thus, in \autoref{sec:Stability:P_Stability}, we derive stability conditions using leaky agent PI controllers with $\sdampCtrl > 0$. 
\subsection{Leaky PI-Controlled Stability} \label{sec:Stability:P_Stability}
Consider the case where the passivity properties of all subsystems in \autoref{fig:Control:control_structure} except for the weighting function \eqref{eq:Control:NL_Func} are fixed. Combining the results in \autoref{sec:Passivity} with \autoref{thm:Prelim:Min_restrictive_diss}, we now determine the weighting function parameters which guarantee closed-loop stability.
\begin{theorem}[Designed closed-loop stability] \label{thm:Stability:Ctrl_MG_Stability}
	The closed-loop in \autoref{fig:Control:control_structure} is guaranteed to be asymptotically stable for the weighting function parameters $\saNLweight = \siInNLweight$, $\cramped{\sbNLweight = 1/\siOutNLweight - \saNLweight}$ if a feasible solution is found for
	\begin{equation} \label{eq:Stability:Opt_P_stability}
		\begin{array}{cl}
			\underset{\begin{array}{c}
					\\[-3.75ex]
					\scriptstyle \siInNLweight,\, \siOutNLweight,\, \sd[i], 
				\end{array}}{\min} & \siInNLweight + \siOutNLweight \\
			\text{s.t.} & \mQ \negDef 0, \quad \sd[i] > 0, \quad i = 1,\dots,5,
		\end{array}
	\end{equation}
	where $\siCrossNLweight = \nicefrac{1}{2}(1 + \siInNLweight\siOutNLweight)$,  $\siCrossDGU = \nicefrac{1}{2}(1 + \siInDGU[1]\siOutDGU)$, and 
	\begin{equation} \label{eq:Stability:Test_matrix_Q}
		\mQ {=} \!\!\left[\begin{Array}{@{\!}c@{}c@{}c@{\,}c@{}c@{\!}}
			-\siOutNLweight \sd[1] & \frac{\sd[2]}{2} & 0 & 0 & - \siCrossNLweight \sd[1]\\
			\frac{\sd[2]}{2} & -\siOutDDA \sd[2] {-} \siInCtrl \sd[3] & \siCrossCtrl \sd[3] & 0 & 0 \\
			0 & \siCrossCtrl \sd[3] & -\siOutCtrl \sd[3] & \frac{\skpCtrl \sd[4]}{2} & 0 \\
			0 & 0 & \frac{\skpCtrl \sd[4]}{2} & - \siOutDDA \sd[4] {-} \siInDGU[1] \sd[5] & \siCrossDGU \sd[5] \\
			- \siCrossNLweight \sd[1] & 0 & 0 & \siCrossDGU \sd[5]  & -\siOutDGU \sd[5] {-} \siInNLweight \sd[1]
		\end{Array}\right]
	\end{equation}
\end{theorem}
\begin{proof}
	Use the supply rates for the DC microgrid in \eqref{eq:Passive:MG_IFOFP}, the two DDA controllers in \eqref{eq:Passivity:Ctrl:H_DDA_dot}, the agent PI controller in \eqref{eq:Passivity:Leaky_PI_supply}, and the \ac{IFOFP} supply rate for the weighting function \eqref{eq:Passivity:Ctrl:Weighting_indices} to construct $\mW$ in \eqref{eq:Prelim:Optimisation_mat_W}.
	Let the output of the PI controller be normalised according to
	\begin{equation} \label{eq:Stability:Normalised_PI_output}
		\vyCtrl = \skiCtrl \vxCtrl + \skpCtrl \vuCtrl = \skpCtrl (\skLoopiCtrl \vxCtrl + \vuCtrl) = \skpCtrl \vyCtrlLoop .
	\end{equation}
	Furthermore, the five subsystems in \autoref{fig:Control:control_structure} are interconnected by $\vu = \mH\vy$, where
	\begin{equation} \label{eq:Stability:Subsystem_interconnection}
		\mH = \begin{bmatrix}
			0 & 0 & 0 & 0 & -1 \\
			1 & 0 & 0 & 0 & 0 \\
			0 & 1 & 0 & 0 & 0 \\
			0 & 0 & \skpCtrl & 0 & 0 \\
			0 & 0 & 0 & 1 & 0 \\
		\end{bmatrix}.
	\end{equation}
	Apply \autoref{thm:Prelim:Min_restrictive_diss}, with $\mD$ as in \eqref{eq:Prelim:Optimisation_mat_D} and simplify $\mQ$ in \eqref{eq:Prelim:Optimisation_mat_Q} to obtain \eqref{eq:Stability:Test_matrix_Q}. This yields the optimisation problem \eqref{eq:Stability:Opt_P_stability}, where the indices of the weighting function $(\siInNLweight, \siOutNLweight)$ are configurable. Asymptotic stability is ensured by changing the matrix inequality in \eqref{eq:Prelim:Opt_interconnected_stability} to a strict inequality and by ensuring that any states not present in $\vy$ are asymptotically stable. The latter condition is ensured through the zero-state analyses in \autoref{prop:Passivity:DGU_zero_state} and \autoref{prop:Passivity:Ctrl:DDA_Passive} and through the condition in \autoref{prop:Passivity:MG_Supply}. Finally, the parameters $\saNLweight$ and $\sbNLweight$ are calculated from \eqref{eq:Passivity:Ctrl:Weighting_indices}.
\end{proof}
Through the application of \autoref{thm:Stability:Ctrl_MG_Stability}, the parameters for the weighting function can thus be designed to ensure stability. We highlight that the results in \autoref{sec:Passivity} and \autoref{thm:Stability:Ctrl_MG_Stability} hold irrespective of the physical or communication topologies and are independent of the actuation states of the nodes, as long as \autorefMulti{ass:Problem:Actuated_agents, ass:Problem:Connected_graphs} hold. Therefore, verifying \autoref{thm:Stability:Ctrl_MG_Stability} ensures robustness against any changes which do not alter the worst-case passivity indices of the respective subsystems (see \eqref{eq:Passive:worst_case_indices}). Note that the presented stability analysis requires strictly passive loads and leaky agent PI controllers (see \autoref{rem:Control:Agent_PI:Non_ideal}). As demonstrated via simulation, these requirements are sufficient for stability, but not necessary.

	\section{Simulation} \label{sec:Simulation}
\begin{figure*}[!ht]
	\centering
	\vspace*{4pt}
		\begin{minipage}[t!]{0.3\textwidth}
			\centering
			\phantomcaption\label{fig:Simulation:network:Stage_A}
			\resizebox{\textwidth}{!}{%
		\tikzsetnextfilename{03_Img/network_simulation_A}%
\begin{tikzpicture}
	\def\nodeXdist{3cm}
	\def\nodeYdist{1.8cm}
	\def\nodeXshift{\nodeXdist/2}
	
	\def\busTerminalXdist{10pt}
	\def\busGenYdist{14pt}
	
	\def\Nodes{1,2,3,4,5,6,7,8,9,10}
	\def\NodesGeneration{
		2/south/0*/-1*,
		3/north/0*/1*,
		4/south/0*/-1*,
		7/north/0*/1*
	}
	
	\def\CommEdgesA{1/4, 2/3, 2/4, 2/5, 4/6, 5/7, 5/8, 7/9}
	\def\CommEdgesPerm{
		1/south/1*/4/north/-2*/-8pt/0pt/8pt/0pt,
		2/south/2*/3/south/-2*/-20pt/-12pt/20pt/-12pt,
		2/south/-2*/4/north/-1*/8pt/4pt/-16pt/12pt,
		2/south/1*/5/north/0*/-12pt/0pt/12pt/0pt,
		4/south/-2*/6/north/0*/8pt/4pt/-12pt/8pt,
		5/south/-1*/7/north/1*/8pt/0pt/-8pt/0pt,
		5/south/1*/8/north/-1*/-8pt/0pt/8pt/0pt
	}
	\def\CommEdgesTemp{
		7/south/0*/9/north/1*/12pt/0pt/-12pt/0pt/0.6/0.4
	}
	
	\def\PhysEdgesA{1/2, 2/3, 2/4, 3/5, 3/8, 4/5, 4/6, 6/7, 6/9, 7/8}
	\def\PhysEdgesPerm{
		1/north/1*/2/north/-1*/8pt,
		2/north/1*/3/north/-1*/8pt,
		2/south/-1*/4/north/0*/0pt,
		3/south/-1*/5/north/1*/0pt,
		3/south/0*/8/north/0*/0pt,
		4/north/1*/5/north/-1*/8pt,
		4/south/-1*/6/north/1*/0pt,
		6/south/1*/7/south/-1*/-8pt,
		7/south/1*/8/south/-1*/-8pt
	}
	\def\PhysEdgesTemp{
		6/south/0*/9/north/-1*/0pt
	}

	\pgfdeclaredecoration{ignore}{final}{
		\state{final}{}
	}

	\pgfdeclaremetadecoration{noMiddle}{initial}{
		\state{initial}[
		width={(\pgfmetadecoratedpathlength - \the\pgfdecorationsegmentlength)/2},
		next state=middle
		]
		{\decoration{curveto}}
		
		\state{middle}[
		width={\the\pgfdecorationsegmentlength},
		next state=final
		]
		{
			\decoration{moveto}
		}
		
		\state{final}
		{
			\decoration{curveto}
		}
	}
	
	\tikzset{noMiddle segment/.style={decoration={noMiddle},decorate, segment length=#1}}
	
	\tikzstyle{bus}				= [draw, rectangle, rounded corners=1pt, fill=black, minimum width = 50pt, minimum height = 2.5pt, inner sep=0pt]
	\tikzstyle{generation}		= [draw, circle, inner sep=1pt, fill=black!10!white]
	
	\tikzstyle{linePhysPerm}	= [-, line width=1.0pt]
	\tikzstyle{lineCommPerm}	= [-, densely dotted, line width=1.0pt, red]
	\tikzstyle{lineCommTemp}	= [-, densely dotted, line width=1.0pt, red]
	\tikzstyle{lineCommTempSw}	= [-, line width=1.0pt, red]
	
	\coordinate(cNode1);
	\path (cNode1) ++(\nodeXdist,0) coordinate (cNode2) ++(\nodeXdist,0) coordinate (cNode3);
	\path (cNode1) ++(\nodeXshift,-\nodeYdist) coordinate (cNode4) ++(\nodeXdist,0) coordinate (cNode5);
	\path (cNode1) ++(0,-2*\nodeYdist) coordinate (cNode6) ++(\nodeXdist,0) coordinate (cNode7) ++(\nodeXdist,0) coordinate (cNode8);
	\path (cNode1) ++(\nodeXshift,-3*\nodeYdist) coordinate (cNode9) ++(\nodeXdist,0) coordinate (cNode10);
	
	\foreach \a in \Nodes {
		\node[bus](node\a) at(cNode\a) {};
		\node[right=2pt](nodeText\a) at(node\a.east) {\large \a};
	}

	\foreach \a/\b/\c/\d in \NodesGeneration {
		\path (node\a.\b) ++(\c\busTerminalXdist,\d\busGenYdist) coordinate (cGen\a);
		\node[generation](gen\a) at(cGen\a) {\small $\denoteDGU$};
		\draw[linePhysPerm] ($(node\a.\b) + (\c\busTerminalXdist,0)$) -- (gen\a);
	}
	
	\foreach \a/\b/\c/\d/\e/\f/\g in \PhysEdgesPerm {
		\path ($(node\a.\b) + (\c\busTerminalXdist,0)$) -- ($(node\d.\e) + (\f\busTerminalXdist,0)$) coordinate[pos=0.5] (mid_path_phys\a\d);
		\draw[linePhysPerm]($(node\a.\b) + (\c\busTerminalXdist,0)$) |- ($(mid_path_phys\a\d) + (0,\g)$) -| ($(node\d.\e) + (\f\busTerminalXdist,0)$);
	}
	\foreach \a/\b/\c/\d/\e/\f/\g in \PhysEdgesTemp {
		\path ($(node\a.\b) + (\c\busTerminalXdist,0)$) -- ($(node\d.\e) + (\f\busTerminalXdist,0)$) coordinate[pos=0.5] (mid_path_phys\a\d);
		\draw[linePhysPerm] let \p1 = ($(node\a.\b) + (\c\busTerminalXdist,0)$), \p2 = (mid_path_phys\a\d), \p3 = ($(node\d.\e) + (\f\busTerminalXdist,0)$) in
			(\p1) |- (\x1,\y2)
			to [cosw, bipoles/length=1.0cm, switches/scale=0.9, bipoles/cuteswitch/thickness=0.5] (\x3,\y2)
			-| (\p3); 
	}

	\foreach \a/\b/\c/\d/\e/\f/\g/\h/\i/\j in \CommEdgesPerm {
		\path ($(node\a.\b) + (\c\busTerminalXdist,0)$) -- ($(node\d.\e) + (\f\busTerminalXdist,0)$) coordinate[pos=0.5] (mid_path_comm\a\d);
		\draw[lineCommPerm]($(node\a.\b) + (\c\busTerminalXdist,0)$) .. controls ($(mid_path_comm\a\d) + (\g,\h)$) and ($(mid_path_comm\a\d) + (\i,\j)$) .. ($(node\d.\e) + (\f\busTerminalXdist,0)$);
	}
	\foreach \a/\b/\c/\d/\e/\f/\g/\h/\i/\j/\k/\l in \CommEdgesTemp {
		\path ($(node\a.\b) + (\c\busTerminalXdist,0)$) -- ($(node\d.\e) + (\f\busTerminalXdist,0)$) coordinate[pos=0.5] (mid_path_comm\a\d);
		
		\draw[lineCommTemp, noMiddle segment = 0.5cm]
			($(node\a.\b) + (\c\busTerminalXdist,0)$) .. 
			controls ($(mid_path_comm\a\d) + (\g,\h)$) and ($(mid_path_comm\a\d) + (\i,\j)$) .. 
			($(node\d.\e) + (\f\busTerminalXdist,0)$)
			coordinate[pos=\k] (mid_switch_start\a\b) 
			coordinate[pos=\l] (mid_switch_end\a\b);
			
		\draw[lineCommTempSw] (mid_switch_start\a\b) to [cosw, color=red, bipoles/length=1.0cm, switches/scale=0.9, bipoles/cuteswitch/thickness=0.5] (mid_switch_end\a\b);
	}

	\node[above=0.5cm] at (node2.north) {\Large \underline{State A}};
\end{tikzpicture}
%
	}
			\vspace*{-4pt}
		\end{minipage}
		\hspace*{22pt}
		\begin{minipage}[t!]{0.3\textwidth}
			\centering
			\phantomcaption\label{fig:Simulation:network:Stage_B}
			\resizebox{\textwidth}{!}{%
		\tikzsetnextfilename{03_Img/network_simulation_B}%
\begin{tikzpicture}
	\def\nodeXdist{3cm}
	\def\nodeYdist{1.8cm}
	\def\nodeXshift{\nodeXdist/2}
	
	\def\busTerminalXdist{10pt}
	\def\busGenYdist{14pt}
	
	\def\Nodes{1,2,3,4,5,6,7,8,9,10}
	\def\NodesGeneration{
		1/north/-1*/1*,
		3/north/0*/1*,
		4/south/0*/-1*,
		10/south/0*/-1*
	}
	
	\def\CommEdgesA{1/4, 2/3, 2/4, 2/5, 4/6, 5/7, 5/8, 7/9}
	\def\CommEdgesPerm{
		1/south/1*/2/south/-1*/-20pt/-12pt/20pt/-12pt,
		1/south/-1*/6/north/-1*/-12pt/20pt/-12pt/-20pt,
		2/south/1*/3/south/-2*/-20pt/-12pt/20pt/-12pt,
		2/south/0*/5/north/0*/-12pt/0pt/12pt/0pt,
		3/south/1*/8/north/1*/12pt/20pt/12pt/-20pt,
		4/south/-1*/6/north/0*/8pt/4pt/-12pt/8pt,
		5/south/-1*/7/north/1*/8pt/0pt/-8pt/0pt,
		5/south/1*/8/north/-1*/-8pt/0pt/8pt/0pt,
		6/north/1*/7/north/-1*/-20pt/12pt/20pt/12pt
	}
	\def\CommEdgesTemp{
		7/south/0*/10/north/-1*/-12pt/0pt/12pt/0pt/0.4/0.6,
		8/south/1*/10/north/1*/32pt/-32pt/-12pt/-8pt/0.5/0.3
	}
	
	\def\PhysEdgesA{1/2, 2/3, 2/4, 3/5, 3/8, 4/5, 4/6, 6/7, 6/9, 7/8}
	\def\PhysEdgesPerm{
		1/north/0*/2/north/-1*/8pt,
		1/south/0*/4/north/-1*/0pt,
		2/north/1*/3/north/-1*/8pt,
		3/south/-1*/5/north/1*/0pt,
		3/south/0*/8/north/0*/0pt,
		4/north/1*/5/north/-1*/8pt,
		4/south/1*/7/north/-0*/0pt,
		6/south/1*/7/south/-1*/-8pt,
		7/south/1*/8/south/-1*/-8pt
	}
	\def\PhysEdgesTemp{
		10/north/0*/8/south/0*/0pt
	}

	\pgfdeclaredecoration{ignore}{final}{
		\state{final}{}
	}

	\pgfdeclaremetadecoration{noMiddle}{initial}{
		\state{initial}[
		width={(\pgfmetadecoratedpathlength - \the\pgfdecorationsegmentlength)/2},
		next state=middle
		]
		{\decoration{curveto}}
		
		\state{middle}[
		width={\the\pgfdecorationsegmentlength},
		next state=final
		]
		{
			\decoration{moveto}
		}
		
		\state{final}
		{
			\decoration{curveto}
		}
	}
	
	\tikzset{noMiddle segment/.style={decoration={noMiddle},decorate, segment length=#1}}
	
	\tikzstyle{bus}				= [draw, rectangle, rounded corners=1pt, fill=black, minimum width = 50pt, minimum height = 2.5pt, inner sep=0pt]
	\tikzstyle{generation}		= [draw, circle, inner sep=1pt, fill=black!10!white]
	
	\tikzstyle{linePhysPerm}	= [-, line width=1.0pt]
	\tikzstyle{lineCommPerm}	= [-, densely dotted, line width=1.0pt, red]
	\tikzstyle{lineCommTemp}	= [-, densely dotted, line width=1.0pt, red]
	\tikzstyle{lineCommTempSw}	= [-, line width=1.0pt, red]
	
	\coordinate(cNode1);
	\path (cNode1) ++(\nodeXdist,0) coordinate (cNode2) ++(\nodeXdist,0) coordinate (cNode3);
	\path (cNode1) ++(\nodeXshift,-\nodeYdist) coordinate (cNode4) ++(\nodeXdist,0) coordinate (cNode5);
	\path (cNode1) ++(0,-2*\nodeYdist) coordinate (cNode6) ++(\nodeXdist,0) coordinate (cNode7) ++(\nodeXdist,0) coordinate (cNode8);
	\path (cNode1) ++(\nodeXshift,-3*\nodeYdist) coordinate (cNode9) ++(\nodeXdist,0) coordinate (cNode10);
	
	\foreach \a in \Nodes {
		\node[bus](node\a) at(cNode\a) {};
		\node[right=2pt](nodeText\a) at(node\a.east) {\large \a};
	}

	\foreach \a/\b/\c/\d in \NodesGeneration {
		\path (node\a.\b) ++(\c\busTerminalXdist,\d\busGenYdist) coordinate (cGen\a);
		\node[generation](gen\a) at(cGen\a) {\small $\denoteDGU$};
		\draw[linePhysPerm] ($(node\a.\b) + (\c\busTerminalXdist,0)$) -- (gen\a);
	}
	
	\foreach \a/\b/\c/\d/\e/\f/\g in \PhysEdgesPerm {
		\path ($(node\a.\b) + (\c\busTerminalXdist,0)$) -- ($(node\d.\e) + (\f\busTerminalXdist,0)$) coordinate[pos=0.5] (mid_path_phys\a\d);
		\draw[linePhysPerm]($(node\a.\b) + (\c\busTerminalXdist,0)$) |- ($(mid_path_phys\a\d) + (0,\g)$) -| ($(node\d.\e) + (\f\busTerminalXdist,0)$);
	}
	\foreach \a/\b/\c/\d/\e/\f/\g in \PhysEdgesTemp {
		\path ($(node\a.\b) + (\c\busTerminalXdist,0)$) -- ($(node\d.\e) + (\f\busTerminalXdist,0)$) coordinate[pos=0.5] (mid_path_phys\a\d);
		\draw[linePhysPerm] let \p1 = ($(node\a.\b) + (\c\busTerminalXdist,0)$), \p2 = (mid_path_phys\a\d), \p3 = ($(node\d.\e) + (\f\busTerminalXdist,0)$) in
			(\p1) |- (\x1,\y2)
			to [cosw, bipoles/length=1.0cm, switches/scale=0.9, bipoles/cuteswitch/thickness=0.5] (\x3,\y2)
			-| (\p3); 
	}

	\foreach \a/\b/\c/\d/\e/\f/\g/\h/\i/\j in \CommEdgesPerm {
		\path ($(node\a.\b) + (\c\busTerminalXdist,0)$) -- ($(node\d.\e) + (\f\busTerminalXdist,0)$) coordinate[pos=0.5] (mid_path_comm\a\d);
		\draw[lineCommPerm]($(node\a.\b) + (\c\busTerminalXdist,0)$) .. controls ($(mid_path_comm\a\d) + (\g,\h)$) and ($(mid_path_comm\a\d) + (\i,\j)$) .. ($(node\d.\e) + (\f\busTerminalXdist,0)$);
	}
	\foreach \a/\b/\c/\d/\e/\f/\g/\h/\i/\j/\k/\l in \CommEdgesTemp {
		\path ($(node\a.\b) + (\c\busTerminalXdist,0)$) -- ($(node\d.\e) + (\f\busTerminalXdist,0)$) coordinate[pos=0.5] (mid_path_comm\a\d);
		
		\draw[lineCommTemp, noMiddle segment = 0.5cm]
			($(node\a.\b) + (\c\busTerminalXdist,0)$) .. 
			controls ($(mid_path_comm\a\d) + (\g,\h)$) and ($(mid_path_comm\a\d) + (\i,\j)$) .. 
			($(node\d.\e) + (\f\busTerminalXdist,0)$)
			coordinate[pos=\k] (mid_switch_start\a\b) 
			coordinate[pos=\l] (mid_switch_end\a\b);
			
		\draw[lineCommTempSw] (mid_switch_start\a\b) to [cosw, color=red, bipoles/length=1.0cm, switches/scale=0.9, bipoles/cuteswitch/thickness=0.5] (mid_switch_end\a\b);
	}

	\node[above=0.5cm] at (node2.north) {\Large \underline{State B}};
\end{tikzpicture}
%
	}
		\end{minipage}
		\hspace*{22pt}
		\begin{minipage}[t!]{0.25\textwidth}
			\centering
			\resizebox{0.8\textwidth}{!}{%
		\tikzsetnextfilename{03_Img/network_simulation_legend}%
\begin{tikzpicture}
	\def\nodeXdist{0.5cm}
	\def\nodeYdist{0.5cm}
	\def\nodeXshift{\nodeXdist/2}
	
	\def\busTerminalXdist{10pt}
	\def\busGenYdist{14pt}
	
	\def\Nodes{1,2,3,4,5,6,7,8,9,10}
	\def\NodesGeneration{
		3/north/0*/1*,
		4/south/0*/-1*,
		10/south/0*/-1*
	}
	
	\def\CommEdgesA{1/4, 2/3, 2/4, 2/5, 4/6, 5/7, 5/8, 7/9}
	\def\CommEdgesPerm{
		1/south/1*/2/south/-1*/-20pt/-12pt/20pt/-12pt,
		1/south/-1*/6/north/-1*/-12pt/20pt/-12pt/-20pt,
		2/south/1*/3/south/-2*/-20pt/-12pt/20pt/-12pt,
		2/south/0*/5/north/0*/-12pt/0pt/12pt/0pt,
		3/south/1*/8/north/1*/12pt/20pt/12pt/-20pt,
		4/south/-1*/6/north/0*/8pt/4pt/-12pt/8pt,
		5/south/-1*/7/north/1*/8pt/0pt/-8pt/0pt,
		5/south/1*/8/north/-1*/-8pt/0pt/8pt/0pt,
		6/north/1*/7/north/-1*/-20pt/12pt/20pt/12pt
	}
	\def\CommEdgesTemp{
		7/south/0*/10/north/-1*/-12pt/0pt/12pt/0pt/0.4/0.6,
		8/south/1*/10/north/1*/32pt/-32pt/-12pt/-8pt/0.5/0.3
	}
	
	\def\PhysEdgesA{1/2, 2/3, 2/4, 3/5, 3/8, 4/5, 4/6, 6/7, 6/9, 7/8}
	\def\PhysEdgesPerm{
		1/north/0*/2/north/-1*/8pt,
		1/south/0*/4/north/-1*/0pt,
		2/north/1*/3/north/-1*/8pt,
		3/south/-1*/5/north/1*/0pt,
		3/south/0*/8/north/0*/0pt,
		4/north/1*/5/north/-1*/8pt,
		4/south/1*/7/north/-0*/0pt,
		6/south/1*/7/south/-1*/-8pt,
		7/south/1*/8/south/-1*/-8pt
	}
	\def\PhysEdgesTemp{
		10/north/0*/8/south/0*/0pt
	}

	\pgfdeclaredecoration{ignore}{final}{
		\state{final}{}
	}

	\pgfdeclaremetadecoration{noMiddle}{initial}{
		\state{initial}[
		width={(\pgfmetadecoratedpathlength - \the\pgfdecorationsegmentlength)/2},
		next state=middle
		]
		{\decoration{curveto}}
		
		\state{middle}[
		width={\the\pgfdecorationsegmentlength},
		next state=final
		]
		{
			\decoration{moveto}
		}
		
		\state{final}
		{
			\decoration{curveto}
		}
	}
	
	\tikzset{noMiddle segment/.style={decoration={noMiddle},decorate, segment length=#1}}
	
	\tikzstyle{bus}				= [draw, rectangle, rounded corners=1.5pt, fill=black, minimum width = 20pt, minimum height = 2.75pt, inner sep=0pt]
	\tikzstyle{generation}		= [draw, circle, inner sep=1pt, fill=black!10!white]
	
	\tikzstyle{linePhysPerm}	= [-, line width=1.0pt]
	\tikzstyle{lineCommPerm}	= [-, densely dotted, line width=1.0pt, red]
	\tikzstyle{lineCommTemp}	= [-, densely dotted, line width=1.0pt, red]
	\tikzstyle{lineCommTempSw}	= [-, line width=1.0pt, red]
	
	\coordinate(cBus);
	\path (cBus) ++(\nodeXdist,0) coordinate (cBusText);
	\path (cBus) ++(0,-\nodeYdist) coordinate (cDGU) ++(\nodeXdist,0) coordinate (cDGUText);
	\path (cBus) ++(0,-2*\nodeYdist) coordinate (cPhys) ++(\nodeXdist,0) coordinate (cPhysText);
	\path (cBus) ++(0,-3*\nodeYdist) coordinate (cComm) ++(\nodeXdist,0) coordinate (cCommText);
	
	\node[bus] at (cBus) {};
	\node[generation] at (cDGU) {\small $\denoteDGU$};
	\draw[linePhysPerm] ($(cPhys) + (-10pt,0)$) -- ($(cPhys) + (10pt,0)$);
	\draw[lineCommPerm] ($(cComm) + (-10pt,0)$) -- ($(cComm) + (10pt,0)$);
	
	\node[anchor=west] at (cBusText) {Bus};
	\node[anchor=west] at (cDGUText) {Active DGU};
	\node[anchor=west] at (cPhysText) {Electrical line};
	\node[anchor=west] at (cCommText) {Communication line};

\end{tikzpicture}
%
	}
			\newline
			
			\resizebox{\textwidth}{!}{\includegraphics[scale=1]{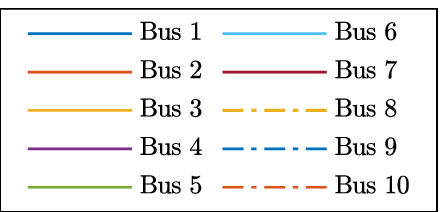}}
		\end{minipage}
	\caption{Two different states for a 10-bus DC microgrid along with electrical and communication connections. The loads at the buses are omitted for clarity.}
	\label{fig:Simulation:network}
\end{figure*}
\begin{figure}[!t]
	\centering
	\resizebox{\columnwidth}{!}{\includegraphics[scale=1]{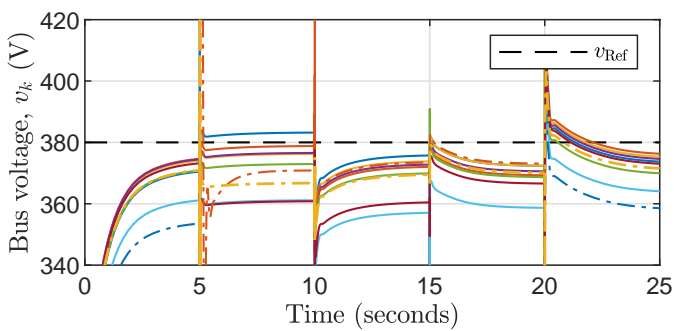}}
	\caption{Simulated bus voltages with line colours as per the legend in \autoref{fig:Simulation:network}.}
	\label{fig:Simulation:voltages}
\end{figure}
\begin{figure}[!t]
	\centering
	\resizebox{\columnwidth}{!}{\includegraphics[scale=1]{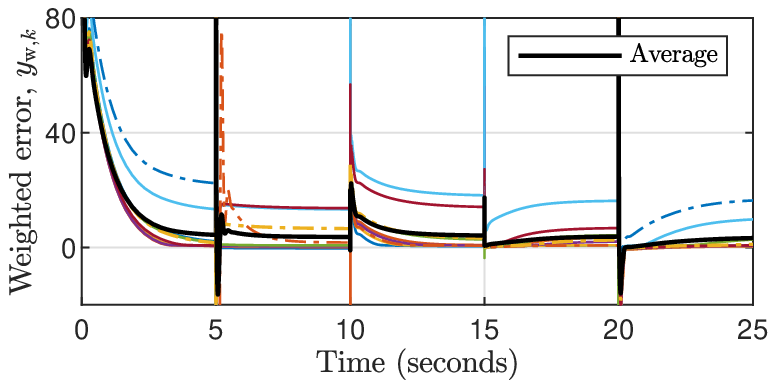}}
	\caption{Simulated weighted voltage errors and the average error of connected agents with agent line colours as per the legend in \autoref{fig:Simulation:network}.}
	\label{fig:Simulation:errors}
\end{figure}
\begin{figure}[!t]
	\centering
	\resizebox{\columnwidth}{!}{\includegraphics[scale=1]{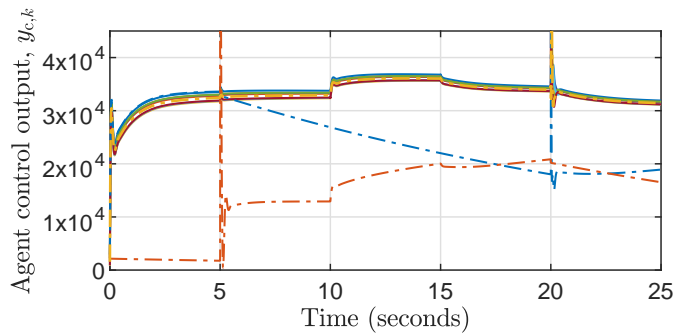}}
	\caption{Simulated outputs of the local agent controllers with line colours as per the legend in \autoref{fig:Simulation:network}.}
	\label{fig:Simulation:agent_control}
\end{figure}
\begin{figure}[!t]
	\centering
	\resizebox{\columnwidth}{!}{\includegraphics[scale=1]{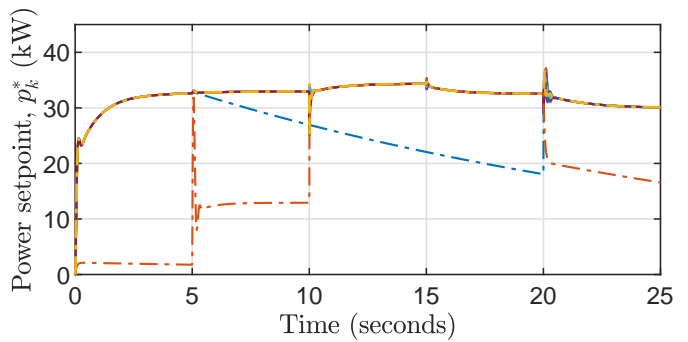}}
	\caption{Simulated power setpoints with line colours as per the legend in \autoref{fig:Simulation:network}.}
	\label{fig:Simulation:power_setpoints}
\end{figure}
In this section, we demonstrate the coordination and robustness of the proposed control structure by means of a \Matlab/\Simulink simulation using \Simscape components.
We consider the network comprising 10 buses depicted in \autoref{fig:Simulation:network}.
In \autoref{sec:Simulation:Setup}, we describe the setup of the simulation along with the various changes that the network is subjected to.
Next, in \autoref{sec:Simulation:Reults}, simulation results are presented for the case where \autoref{thm:Stability:Ctrl_MG_Stability} holds, i.e. with strictly passive loads and leaky agent PI controllers.
Finally, in \autoref{sec:Simulation:Robust}, we show the robust stability of the proposed control structure for passive loads and ideal agent PI controllers.
\subsection{Simulation Setup} \label{sec:Simulation:Setup}
%
The DC microgrid in \autoref{fig:Simulation:network} is simulated with the parameters in \autoref{tab:Simulation:Sim_params}. The ZIP load parameters are chosen randomly in the specified ranges such that the required passivity measures are fulfilled (see \autoref{rem:Passivity:load_measure}). Furthermore, typical values are used for the DGUs and the lines \cite{Nasirian2015, Tucci2016, Cucuzzella2019cons}. The lines exhibit the same per kilometer parameter values and the line length are chosen randomly in the given interval. The line lengths are given in \autorefapp{sec:App_Sim_Data}.
\begin{table}[!t]
	\centering
	\renewcommand{\arraystretch}{1.25}
	\caption{Simulation Parameter Values}
	\label{tab:Simulation:Sim_params}
	\begin{tabular}{lr@{\;}lr@{\;}l}
		\noalign{\hrule height 1.0pt}
		Voltages & $\svRef =$ & $\SI{380}{\volt}$ & $\svCrit =$ & $\SI{266}{\volt}$  \\
		\hline
		DGU Filters \eqref{eq:Problem:node_act_dynamics} & $\sR[k] =$ & $\SI{0.2}{\ohm}$ & $\sL[k] =$ & $\SI{1.8}{\milli\henry}$\\
		& $\sC[k] =$ & $\SI{2.2}{\milli\farad}$ && \\
		\hline
		ZIP Loads \eqref{eq:Problem:ZIP_model} & $|\sZinv| \le$ & $\SI{0.1}{\per\ohm}$ & $|\sI\:\!| \le$ & $\SI{21}{\ampere}$ \\
		& $|\sP\,| \le$ & $\SI{3}{\kilo\watt}$ && \\
		\hline
		Elec. Lines \eqref{eq:Problem:line_dynamics} & $\sR[kl] =$ & $\SI{0.1}{\ohm\per\kilo\meter}$ & $\sL[kl] =$ & $\SI{2}{\micro\henry\per\kilo\meter}$ \\
		& $\sC[kl] =$ & $\SI{22}{\nano\farad\per\kilo\meter}$ & length $\in$ & $[0.2; 10]\,\si{\kilo\meter}$ \\
		\noalign{\hrule height 1.0pt}
	\end{tabular}
\end{table}
\begin{table}[!t]
	\centering
	\renewcommand{\arraystretch}{1.25}
	\caption{Controller Parameter Values}
	\label{tab:Simulation:Control_params}
	\begin{tabular}{lr@{\;}l@{\quad}r@{\;}l@{\quad}r@{\;}l}
		\noalign{\hrule height 1.0pt}
		Power PI Control \eqref{eq:Problem:DGU_Power_Control} & $\skpDGU =$ & $90$ & $\skiDGU =$ & $90$ & $\sRDGUdamp =$ & $-8$ \\
		DDA Control \eqref{eq:Control:DDA} & $\skpDDA =$ & $50$ & $\skiDDA =$ & $100$ & $\sgDDA =$ & $16$ \\ 
		Agent PI Control \eqref{eq:Control:local_PI_vector} & $\skpCtrl =$ & $160$ & $\skiCtrl =$ & $600\!\!$ & $\sdampCtrl =$ & $0.08$ \\
		Weighting Function \eqref{eq:Control:NL_Func} & $\saNLweight =$ & $0.1$ & $\sbNLweight =$ & $1.1$ & $\scNLweight =$ & $\SI{7.5}{\volt}$ \\
		\noalign{\hrule height 1.0pt}
	\end{tabular}
\end{table}

The simulation starts off in State~A (see \autoref{fig:Simulation:network}) with Bus~9 connected and with all states at zero. The following changes are made at the indicated times.
\begin{itemize}
	\item $t = \SI{5}{\second}$: The actuation states $\sAct[i]$ of the buses switches from State~A to State~B and Bus~9 is disconnected.
	\item $t = \SI{10}{\second}$: The communication topology switches from State~A to State~B and Bus~10 is connected.
	\item $t = \SI{15}{\second}$: The electrical topology switches from State~A to State~B.
	\item $t = \SI{20}{\second}$: The bus actuation status along with the communication and electrical topologies revert to State~A. Bus~9 is connected and Bus~10 is disconnected.
\end{itemize}
Furthermore, at each change, half of the buses are randomly selected and assigned new ZIP load parameters. The ZIP load parameters can be found in \autorefapp{sec:App_Sim_Data}.

The parameters for the closed-loop controller, as specified in \autoref{tab:Simulation:Control_params}, are designed constructively, starting from the microgrid subsystems. First, the passivity indices for the lines ($\siOutTx = 0.01$) and loads ($\siOutLoad = \scLoadlo = 0.05$) are calculated from \autoref{prop:Passivity:Line_Passivity} and \autoref{prop:Passivity:Load_passivity}, respectively. Next, the parameters for the power regulator \eqref{eq:Problem:DGU_Power_Control} are chosen and the DGU passivity indices are calculated from \autoref{thm:Passivity:Shifted_DGU}, with the optimisation verified for the practically relevant intervals $\cramped{\sv \in [\SI{200}{\volt}, \SI{550}{\volt}]}$ and $\cramped{\sieq \in [\SI{10}{\ampere}, \SI{350}{\ampere}]}$. Note that adding the restriction $\cramped{\siInDGU[2] \ge -\siOutTx}$ to the optimisation in \autoref{thm:Passivity:Shifted_DGU} ensures that \eqref{eq:Passive:MG_lines_dominate_DGU} will be met. This yields a solution $\siInDGU[1] = -4.686$, $\siInDGU[2] = -0.01$ and $\siOutDGU = 0.01$, from which the microgrid supply rate is constructed as per \autoref{thm:Passive:MG_indep_passive}. Finally, parameters for the agent PI controllers are chosen and the weighting function parameters are designed using \autoref{thm:Stability:Ctrl_MG_Stability}. Note that \autoref{thm:Passive:MG_indep_passive} requires strictly passive loads ($\scLoadlo > 0$) and \autoref{thm:Stability:Ctrl_MG_Stability} necessitates leaky integrators $(\cramped{\sdampCtrl > 0})$.
\subsection{Results} \label{sec:Simulation:Reults}
The bus voltages $\sv[k]$ shown in \autoref{fig:Simulation:voltages} confirm the stability of the closed loop results, although the voltages tend to be lower than desired, due to the use of leaky integrators. The remaining steady-state offset can also be seen in the weighted errors plotted in \autoref{fig:Simulation:errors}, where the average tends towards a non-zero value in each instance (see \autoref{rem:Control:Agent_PI:Non_ideal}). Despite this, the four stage controller reaches a consensus on the average of the nonlinear weighted voltage errors. Moreover, the advantage of the weighting function can be seen at Bus~6 in $t \in [\SI{20}{\second}, \SI{25}{\second})$, where a significant weighted error only appears in \autoref{fig:Simulation:errors} when the voltage in \autoref{fig:Simulation:voltages} is not close to $\svRef$. Note that the voltages of Buses~9 and 10 are at $\SI{0}{\volt}$ during the respective periods where they are disconnected and not actuated.

In \autoref{fig:Simulation:agent_control}, the outputs of the agent controllers show that no synchronisation of the agent controllers are required. The agent controller outputs at Buses~1 to 8, which are continuously connected to the communication network, are near identical. However, the disconnecting buses, e.g.\ Bus~9 after $t = \SI{5}{\second}$, rapidly diverge from other controllers and do not synchronise on reconnect. Despite this, the final stage of the controller ensures cooperation of the buses, as demonstrated in the power setpoints $\spSP[k]$ in \autoref{fig:Simulation:power_setpoints}. When Bus~10 connects at $t = \SI{10}{\second}$, its setpoint $\spSP[k]$ rapidly converges to the coordinated common setpoint used by all connected agents.

Although the leaky integrators yield imperfect results (see \autoref{rem:Control:Agent_PI:Non_ideal} and \autoref{fig:Simulation:errors}), this can be mitigated by choosing a higher $\svRef$. Indeed, by combining the steady state of the agent PI controller \eqref{eq:Control:leaky_PI_steady_state} with the \ac{DDA} steady state \eqref{eq:Control:DDA_equilibrium}, we see that injecting power into the system $\vpSP > 0$ results in positive voltage errors. Since we consider (strictly) passive loads, increasing $\svRef$ is thus a viable method for correcting the imperfect results whilst retaining the advantageous properties of the stability analysis in \autoref{thm:Stability:Ctrl_MG_Stability}.
\subsection{Robustness Test} \label{sec:Simulation:Robust}
\begin{figure}[!t]
	\centering
	\resizebox{\columnwidth}{!}{\includegraphics[scale=1]{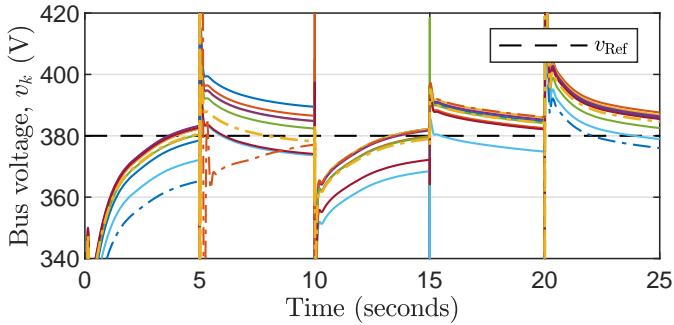}}
	\caption{Simulated bus voltages with ideal PI controllers and with line colours as per the legend in \autoref{fig:Simulation:network}.}
	\label{fig:Simulation:robust_voltages}
\end{figure}
\begin{figure}[!t]
	\centering
	\resizebox{\columnwidth}{!}{\includegraphics[scale=1]{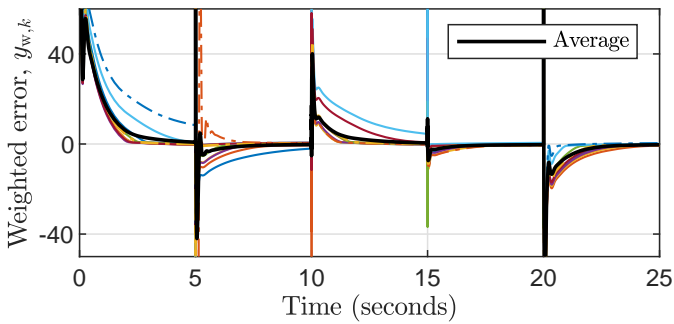}}
	\caption{Simulated weighted voltage errors and the average error of connected agents with ideal PI controllers and with agent line colours as per the legend in \autoref{fig:Simulation:network}.}
	\label{fig:Simulation:robust_errors}
\end{figure}
\begin{figure}[!t]
	\centering
	\resizebox{\columnwidth}{!}{\includegraphics[scale=1]{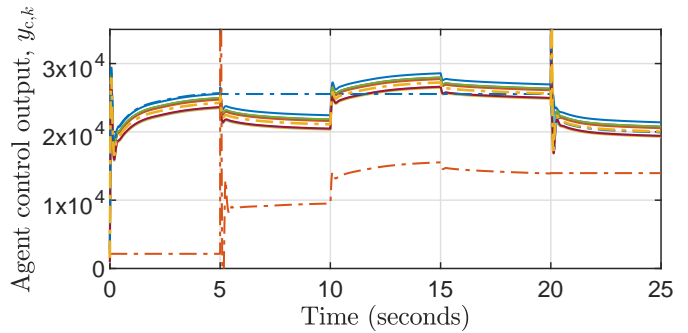}}
	\caption{Simulated outputs of the local agent controllers with ideal PI controllers and with line colours as per the legend in \autoref{fig:Simulation:network}.}
	\label{fig:Simulation:robust_agent_control}
\end{figure}
\begin{figure}[!t]
	\centering
	\resizebox{\columnwidth}{!}{\includegraphics[scale=1]{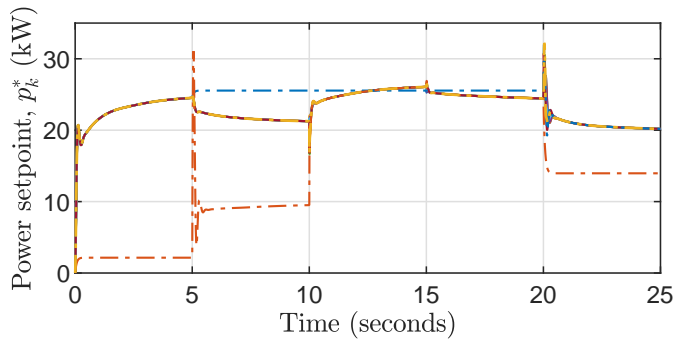}}
	\caption{Simulated power setpoints with ideal PI controllers and with line colours as per the legend in \autoref{fig:Simulation:network}.}
	\label{fig:Simulation:robust_power_setpoints}
\end{figure}
We now repeat the simulation described in \autoref{sec:Simulation:Setup} with the following changes. 1) Passive loads with $\scLoadlo = 0$ are allowed at all buses, and 2) ideal agent PI controllers with $\sdampCtrl = 0$ are used. Under these conditions, \autoref{thm:Stability:Ctrl_MG_Stability} can no longer be used to verify the stability. However, the stability may still be verified using classical approaches such as evaluating the eigenvalues for the closed loop linearised about the equilibrium. Note that the same random seed is used as for the results in \autoref{sec:Simulation:Setup}, allowing for a comparison between the scenarios to be made.

\autoref{fig:Simulation:robust_voltages} demonstrates the improved consensus achieved by the ideal PI agents, in that the bus voltages are closer to $\svRef$ at steady state than in \autoref{fig:Simulation:voltages}. Moreover, \autoref{fig:Simulation:robust_errors} shows that perfect consensus is achieved, where the average error tends to zero in each case. This figure also demonstrates the robustness against communication interruptions, as is the case for Bus~10 which, for the period $t \in [\SI{5}{\second}, \SI{10}{\second})$, is actuated but does not communicate with the other buses. Despite this, it is able to accurately regulate its own bus voltage (compared to the imperfect regulation achieved with leaky integrators as in \autoref{fig:Simulation:voltages}). The lack of leaky integrators is also evident in \autoref{fig:Simulation:robust_agent_control}, where the output of the agent controllers stay constant when a bus is disconnected and not actuated. Lastly, the power setpoints in \autoref{fig:Simulation:robust_power_setpoints} converging to a common value for the communicating agents confirm the coordination of the agents.

Note that while tests with non-passive loads can also yield a stable closed loop, instability can occur when the non-passive loads dominate. To address this, a targeted compensation of non-passive loads is required (see \autoref{rem:Passivity:Non_passive_loads}).
	\section{Conclusion}\label{sec:Conclusion}
%
In this paper, we proposed a four-stage distributed control structure that achieves power sharing in a DC microgrid while ensuring voltage regulation for the voltages of both actuated and unactuated buses. 
%
We demonstrated how the passivity properties of various subsystems can be determined and combined these in a stability analysis that is independent of topological changes, actuation changes, bus connections or disconnections and load changes.

Future work includes the consideration of non-passive loads at arbitrary locations in the microgrid and the construction of an interface to allow for the presented work to be combined with tertiary optimal controllers. 
	
	\appendices
	
	\section{Proofs} \label{sec:App_Proofs}
\begin{proof}[Proof of \autoref{prop:Control:desired_equilibrium}]
	For the control structure in steady state, $\vxCtrldot = 0$ and thus $\vyCtrl$ is constant. The steady-state output \eqref{eq:Control:DDA_equilibrium} of the Stage~4 \ac{DDA} therefore ensures \autoref{obj:Problem:power_sharing} is achieved. Furthermore, consider the steady state of the Stage~2 \ac{DDA}
	\begin{align}
		\label{eq:Control:DDA_2_equilibrium_input}
		\suDDA[s, k] &= \lim_{t \to \infty} \shNLweight(\svRef - \sv[k]), \\
		\label{eq:Control:DDA_2_equilibrium_output}
		\lim_{t \to \infty} \syDDA[2, k] &= \frac{\vuDDAT[s]\vOneCol[N]}{N} = \lim_{t \to \infty} \frac{1}{N} \sum_{k \in \setN} \left(\svRef - \sh(\sv[k])\right),
	\end{align}
	if $\sv[k]$ is in equilibrium and where $\sh$ is obtained by shifting $\shNLweight$ by $\svRef$. Note that \eqref{eq:Control:DDA_2_equilibrium_output} corresponds to the condition of \eqref{eq:Problem:Obj_voltage_consensus} in \autoref{obj:Problem:weighted_errors}. Therefore,  $\vyDDA[2]$ specifies the regulation error of the average weighted voltage error in steady state. From the steady state of the agent PI controller in \eqref{eq:Control:local_PI_vector}, we have $\sdampCtrl\vxCtrl = \vyDDA[2]$. Thus, ideal integrators with $\sdampCtrl = 0$ ensure that \autoref{obj:Problem:weighted_errors} is met exactly. For $\sdampCtrl > 0$, substitute the PI equilibrium into the output of the agent PI controller in \eqref{eq:Control:local_PI_vector} to obtain the steady state equation
	\begin{equation} \label{eq:Control:PI_forced_equilibrium}
		\vxCtrl = \frac{1}{\skiCtrl} \left(\vyCtrl + \skpCtrl \vyDDA[2]\right).
	\end{equation}
	Substitute $\sdampCtrl\vxCtrl = \vyDDA[2]$ into \eqref{eq:Control:DDA_equilibrium} and simplify to find
	\begin{equation} \label{eq:Control:Steady_state_error}
		\vyDDA[2] = \frac{\sdampCtrl}{\skiCtrl (1 + \sdampCtrl \skpCtrl)} \vyCtrl,
	\end{equation}
	for the steady state.
	Since the entries of the vector $\vyDDA[2]$ and thus of $\vxCtrl$ and $\vyCtrl$ are the same at steady state. Therefore the steady state output for the Stage~4 \ac{DDA} in \eqref{eq:Control:DDA_equilibrium} gives $\vyCtrl = \vyDDA[4]$, which we combine with \eqref{eq:Control:Steady_state_error} to obtain the error for \autoref{obj:Problem:weighted_errors} in \eqref{eq:Control:Steady_state_error_output}.
\end{proof}

\begin{proof}[Proof of \autoref{thm:Passive:MG_indep_passive}]
		Consider the supply rates which describe the actuated and unactuated states, respectively, for a given bus $k \in \setN$
	\begin{align} \label{eq:Passive:MG_supply_act}
		\swMGAct[k] &= (1 {+} \siInDGU[1]\siOutDGU)\speSPAct[k]\sveAct[k] - \siInDGU[1] {(\speSPAct[k]{})}^2 - \siOutDGU \sveAct[k]^2 ,\!\!  \\
		\label{eq:Passive:MG_supply_unact}
		\swMGUnact[k] &= - \siOutLoad \sveUnact[k]^2 .
	\end{align}
	These allow the microgrid supply rate in \eqref{eq:Passive:MG_supply_rate} to be decomposed according to the actuation states $\sAct[k]$
	\begin{equation} \label{eq:Passive:MG_indep_supply}
		\begin{aligned}
			\swMGdep =& \; \sum_{k \in \setNAct}\swMGAct[k] + \sum_{k \in \setNUnact}\swMGUnact[k] \\
			=& \sum_{k \in \setN}\left( \sAct[k]\swMGAct[k] + (1-\sAct[k])\swMGUnact[k] \right)
		\end{aligned}
	\end{equation}
	Enlarge the supply rate of the unactuated buses in \eqref{eq:Passive:MG_supply_unact} by adding the positive term $\siInLoad {(\speSPUnact[k]{})}^2$ for an arbitrarily small $\siInLoad > 0$ such that 
	\begin{equation} \label{eq:Passive:MG_supply_unact_hi}
		\begin{aligned}	
			\swMGUnact[k] &\le \swMGUnacthi[k] \,= \, \siInLoad {(\speSPUnact[k]{})}^2 - \siOutLoad \sveUnact[k]^2 \\
			&\le \frac{\swMGUnacthi[k]}{\siOutLoad} = \frac{\siInLoad}{\siOutLoad} {(\speSPUnact[k]{})}^2 - \sveUnact[k]^2
		\end{aligned}
	\end{equation}
	for $\siOutLoad$ as in \eqref{eq:Passive:MG_Unact_passive_loads}. The supply rate $\swMGUnacthi[k]/\siOutLoad$ is equivalent to the \Ltwo{} supply rate in \autoref{def:Prelim:passive_rates} and is thus bounded by the sector $[-\sqrt{\frac{\siInLoad}{\siOutLoad}},\sqrt{\frac{\siInLoad}{\siOutLoad}}]$ \cite[Lemma~4]{Malan2022b}.
	Consider now the supply rate of the actuated agents \eqref{eq:Passive:MG_supply_act} narrowed down to an \ac{IFP} sector for the case that $\siOutDGU < 0$, i.e.\
	\begin{equation} \label{eq:Passive:MG_supply_act_lo}
		\swMGAct[k] \ge \swMGActlo[k] = \left\{
		\begin{aligned}
			&\swMGAct[k], \!\quad &&\text{if} \;\; \siOutDGU \ge 0, \\
			&\speSPAct[k]\sveAct[k] - \siInDGU[1] {(\speSPAct[k]{})}^2, \!\quad &&\text{if} \;\; \siOutDGU < 0,
		\end{aligned}
		\right.
	\end{equation}
	such that $\swMGActlo[k]$ is sector bounded by $[\siInDGU[1], \frac{1}{\siOutDGU}]$ if $\siOutDGU > 0$ and $[\siInDGU[1], \infty)$ if $\siOutDGU < 0$ or if $\siOutDGU = 0$ (see \cite[p.~231]{Khalil2002}). A relation bewteen $\swMGActlo$ and $\swMGUnacthi/\siOutLoad$ can now be established by comparing their respective sector bounds:
	\begin{equation} \label{eq:Passive:MG_compare_sectors}
		\frac{\swMGUnacthi[k]}{\siOutLoad} \le \swMGActlo[k] \; \text{if}
		\left\{
		\begin{aligned}
			&[-\sqrt{\tfrac{\siInLoad}{\siOutLoad}},\sqrt{\tfrac{\siInLoad}{\siOutLoad}}] \subseteq [\siInDGU[1], \tfrac{1}{\siOutDGU}], \!\!\! &\text{if} \; \siOutDGU > 0, \\
			&[-\sqrt{\tfrac{\siInLoad}{\siOutLoad}},\sqrt{\tfrac{\siInLoad}{\siOutLoad}}] \subseteq [\siInDGU[1], \infty), \!\!\! &\text{if} \; \siOutDGU \le 0,
		\end{aligned}
		\right.
	\end{equation}
	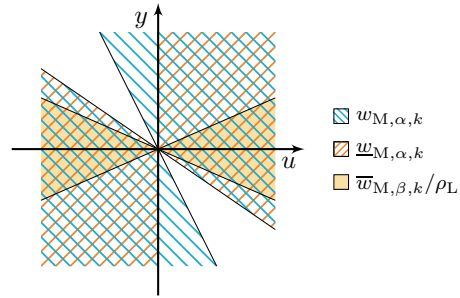
\begin{figure}[!t]
		\centering
		\resizebox{0.7\columnwidth}{!}{%
		\tikzsetnextfilename{03_Img/mg_supply_sectors}%
		\begin{tikzpicture}[>=latex']	
    \def\c{2};
    \def\cArea{1.6};
    \def\UnactLtwoOffset{0.7}
    \def\OutXOffset{-0.8}
    \def\InYOffset{1.1}
    \def\rectSize{0.1}
    
    \coordinate (cTopLeft) at (-\c,\c);
    \coordinate (cTopRight) at (\c,\c);
    \coordinate (cBotLeft) at (-\c,-\c);
    \coordinate (cBotMid) at (0,-\c);
    \coordinate (cMidRight) at (\c,0);
    
    \coordinate (UnactL2_top_left) at (-\cArea,\UnactLtwoOffset);
    \coordinate (UnactL2_top_right) at (\cArea,\UnactLtwoOffset);
    \coordinate (UnactL2_bot_right) at (\cArea,-\UnactLtwoOffset);
    \coordinate (UnactL2_bot_left) at (-\cArea,-\UnactLtwoOffset);
    
    \coordinate (ActMin_top_left) at (-\cArea,\InYOffset);
    \coordinate (ActMin_top_right) at (0,\cArea);
    \coordinate (ActMin_bot_right) at (\cArea,-\InYOffset);
    \coordinate (ActMin_bot_left) at (0,-\cArea);

	\coordinate (Act_top_left) at (-\cArea,\InYOffset);
	\coordinate (Act_top_right) at (\OutXOffset,\cArea);
	\coordinate (Act_bot_right) at (\cArea,-\InYOffset);
	\coordinate (Act_bot_left) at (-\OutXOffset,-\cArea);

	\fill[matlabCol3!40!white] (UnactL2_top_left) -- (UnactL2_bot_right) -- (UnactL2_top_right) -- (UnactL2_bot_left) -- cycle;

	\begin{scope}
		\clip (ActMin_top_left) -- (ActMin_bot_right) |- (ActMin_top_right) -- (ActMin_bot_left) -| cycle;
		\foreach \x in {-4,-3.8,...,6}
		{\draw[bronze,line width=0.7pt]  (-\c,-\x) -- (\x,\c);}
	\end{scope}
	
	\begin{scope}
		\clip (Act_top_left) -- (Act_bot_right) |- (Act_top_right) -- (Act_bot_left) -| cycle;
		\foreach \x in {-6,-5.8,...,6}
		{\draw[ballblue,line width=0.7pt] (-\c,\x) -- (\x,-\c) ;}
	\end{scope}
    
    \draw[->,thick](-\c,0) -- (\c,0) node[below left=0.15, circle, fill=white!50, inner sep=-0.4, align=center] {$\su$};
    \draw[->,thick](0,-\c) -- (0,\c) node[below left] {$\sy$};
    
    
    \draw[-](UnactL2_top_left) -- (UnactL2_bot_right);
    \draw[-](UnactL2_bot_left) -- (UnactL2_top_right);
    \draw[-](ActMin_top_left) -- (ActMin_bot_right);
    \draw[-](ActMin_bot_left) -- (ActMin_top_right);
    \draw[-](Act_bot_left) -- (Act_top_right);
    
 	\path (cMidRight)
 	++ (0.4,-0.45-\rectSize) coordinate (cBoxUnactL2)
 	++ (0,0.45) coordinate (cBoxActMin)
 	++ (0,0.45) coordinate (cBoxAct);
	 	
	\begin{scope}[fill=white!0]
		\fill[clip] (cBoxAct) rectangle +(2*\rectSize,2*\rectSize);
		\foreach \x in {-4,-3.9,...,4}
		{\draw[ballblue,line width=0.6pt]  (-\c+2,\x) -- (\x+2,-\c);}
		\draw (cBoxAct) rectangle +(2*\rectSize,2*\rectSize);
	\end{scope}
	\path (cBoxAct) +(\rectSize*.8,\rectSize*.6) node[right=0.1, font=\footnotesize, anchor=west] {$\swMGAct[k]$};
	 	
 	\begin{scope}[fill=white!0]
		\fill[clip,draw] (cBoxActMin) rectangle +(2*\rectSize,2*\rectSize);
	 	\foreach \x in {-4,-3.9,...,6}
		 	{\draw[bronze,line width=0.6pt]  (-\c+2,-\x) -- (\x+2,\c);}
 		\draw (cBoxActMin) rectangle +(2*\rectSize,2*\rectSize);
 	\end{scope}
 	\path (cBoxActMin) +(\rectSize*.8,\rectSize*.6) node[right=0.1, font=\footnotesize, anchor=west] {$\swMGActlo[k]$};
 
	\begin{scope}[fill=white!0]
		 \fill[matlabCol3!40!white, draw=black] (cBoxUnactL2) rectangle +(2*\rectSize,2*\rectSize);
	\end{scope}
	\path (cBoxUnactL2) +(\rectSize*.8,\rectSize*.6) node[right=0.1, font=\footnotesize, anchor=west] {$\swMGUnacthi[k] / \siOutLoad$};

\end{tikzpicture}%
	}
		\caption{Comparison of the microgrid supply rate sectors in the proof of \autoref{thm:Passive:MG_indep_passive} if $\siOutDGU < 0$.}
		\label{fig:Passive:MG_supply_rates}
	\end{figure}
	Since $\siInLoad$ can be arbitrarily small, we derive \eqref{eq:Passive:MG_Act_iIn_passivity} by comparing the lower bounds in \eqref{eq:Passive:MG_compare_sectors} and note that the upper bound relation can be met for any $\siOutDGU$. A visual comparison of the sector conditions is made in \autoref{fig:Passive:MG_supply_rates}. The combination of \eqref{eq:Passive:MG_supply_unact_hi}--\eqref{eq:Passive:MG_compare_sectors} results in
	\begin{equation} \label{eq:Passive:MG_rate_comparison}
		\swMGUnact[k] \le \swMGUnacthi[k] \le \frac{\swMGUnacthi[k]}{\siOutLoad} \le \swMGActlo[k] \le \swMGAct[k] .
	\end{equation}
	Therefore, for the microgrid with the storage function $\sSMG$ that is dissipative w.r.t.\ \eqref{eq:Passive:MG_supply_rate}, it holds that
	%
	\begin{equation} \label{eq:Passive:MAG_passivity}
		\begin{aligned}
			\sSMGdot \le \swMGdep \le \sum_{k \in \setN}\swMGAct[k] = \swMG,
		\end{aligned}
	\end{equation}
	which is found by combining \eqref{eq:Passive:MG_indep_supply} with \eqref{eq:Passive:MG_rate_comparison}.
\end{proof}
	\section{Simulation Data} \label{sec:App_Sim_Data}
The simulation parameters used for the lines in \autoref{sec:Simulation} are given in \autoref{tab:App_Sim_Data:Line_Params}. Furthermore, the strictly passive load parameters for the simulation results in \autoref{sec:Simulation:Reults} and the passive load parameters for the results in \autoref{sec:Simulation:Robust} are given in \autoref{tab:App_Sim_Data:Load_params_strict} and \autoref{tab:App_Sim_Data:Load_params_pasive}, respectively. Note that the $\sP$ parameter for the loads in \autoref{tab:App_Sim_Data:Load_params_pasive} are the same as listed in \autoref{tab:App_Sim_Data:Load_params_strict}.
\begin{table}[!t]
	\centering
	\renewcommand{\arraystretch}{1.25}
	\caption{Rounded Line Lengths}
	\label{tab:App_Sim_Data:Line_Params}
	\begin{tabular}{cc@{\qquad}cc@{\qquad}cc}
		\noalign{\hrule height 1.0pt}
		Line & Length & Line & Length & Line & Length \\
		\hline
		1 -- 2 & $\SI{1.19}{\kilo\meter}$ & 1 -- 4 & $\SI{7.74}{\kilo\meter}$ & 2 -- 3 & $\SI{2.23}{\kilo\meter}$ \\
		2 -- 4 & $\SI{7.20}{\kilo\meter}$ & 3 -- 5 & $\SI{3.14}{\kilo\meter}$ & 3 -- 8 & $\SI{2.82}{\kilo\meter}$ \\
		4 -- 5 & $\SI{3.72}{\kilo\meter}$ & 4 -- 6 & $\SI{6.75}{\kilo\meter}$ & 4 -- 7 & $\SI{1.16}{\kilo\meter}$ \\
		6 -- 7 & $\SI{4.44}{\kilo\meter}$ & 6 -- 9 & $\SI{3.11}{\kilo\meter}$ & 7 -- 8 & $\SI{3.69}{\kilo\meter}$ \\
		8 -- 10 & $\SI{1.21}{\kilo\meter}$ &&&& \\
		\noalign{\hrule height 1.0pt}
	\end{tabular}
\end{table}
\begin{table}[!t]
	\centering
	\renewcommand{\arraystretch}{1.25}
	\caption{Strictly Passive Load Values}
	\label{tab:App_Sim_Data:Load_params_strict}
	\begin{tabular}{c@{\;}cccccc}
		\noalign{\hrule height 1.0pt}
		Bus & Parameter & $t = \SI{0}{\second}$ & $t = \SI{5}{\second}$ & $t = \SI{10}{\second}$ & $t = \SI{15}{\second}$ & $t = \SI{20}{\second}$ \\
		\noalign{\hrule height 1.0pt}
		& $\sZinv$ $ (\si{\per\ohm})$ &	0.103	&	0.103	&	0.106	&	0.106	&	0.083	\\
		1	& $\sI$ $ (\si{\ampere})$ &	4.66	&	2.15	&	-6.08	&	-6.08	&	14.45	\\
		& $\sP$ $ (\si{\watt})$ &	3599	&	-4055	&	4133	&	4133	&	-4927	\\
		\hline
		& $\sZinv$ $ (\si{\per\ohm})$ &	0.099	&	0.099	&	0.096	&	0.096	&	0.080	\\
		2	& $\sI$ $ (\si{\ampere})$ &	-16.09	&	-16.09	&	19.68	&	19.68	&	2.49	\\
		& $\sP$ $ (\si{\watt})$ &	3204	&	3204	&	2659	&	2659	&	1346	\\
		\hline
		& $\sZinv$ $ (\si{\per\ohm})$ &	0.128	&	0.105	&	0.105	&	0.105	&	0.096	\\
		3	& $\sI$ $ (\si{\ampere})$ &	10.27	&	-0.09	&	-0.09	&	-0.09	&	-11.09	\\
		& $\sP$ $ (\si{\watt})$ &	-1479	&	-3659	&	-3659	&	-3659	&	3031	\\
		\hline
		& $\sZinv$ $ (\si{\per\ohm})$ &	0.079	&	0.079	&	0.079	&	0.079	&	0.079	\\
		4	& $\sI$ $ (\si{\ampere})$ &	10.15	&	10.15	&	10.15	&	10.15	&	10.15	\\
		& $\sP$ $ (\si{\watt})$ &	-2711	&	-2711	&	-2711	&	-2711	&	-2711	\\
		\hline
		& $\sZinv$ $ (\si{\per\ohm})$ &	0.095	&	0.095	&	0.095	&	0.064	&	0.107	\\
		5	& $\sI$ $ (\si{\ampere})$ &	-6.64	&	-6.64	&	-6.64	&	16.68	&	2.10	\\
		& $\sP$ $ (\si{\watt})$ &	2768	&	2768	&	2768	&	-3798	&	4242	\\
		\hline
		& $\sZinv$ $ (\si{\per\ohm})$ &	0.089	&	0.089	&	0.106	&	0.103	&	0.103	\\
		6	& $\sI$ $ (\si{\ampere})$ &	6.87	&	6.87	&	7.85	&	-5.17	&	-5.17	\\
		& $\sP$ $ (\si{\watt})$ &	948		&	948		&	4321	&	370		&	370		\\
		\hline
		& $\sZinv$ $ (\si{\per\ohm})$ &	0.065	&	0.092	&	0.092	&	0.118	&	0.118	\\
		7	& $\sI$ $ (\si{\ampere})$ &	11.96	&	6.51	&	6.51	&	2.77	&	2.77	\\
		& $\sP$ $ (\si{\watt})$ &	-3624	&	-3442	&	-3442	&	-3890	&	-3890	\\
		\hline
		& $\sZinv$ $ (\si{\per\ohm})$ &	0.102	&	0.102	&	0.086	&	0.086	&	0.124	\\
		8	& $\sI$ $ (\si{\ampere})$ &	-16.85	&	-16.85	&	20.71	&	20.71	&	-4.68	\\
		& $\sP$ $ (\si{\watt})$ &	3529	&	3529	&	-4773	&	-4773	&	-3832	\\
		\hline
		& $\sZinv$ $ (\si{\per\ohm})$ &	0.111	&	0.103	&	0.109	&	0.077	&	0.077	\\
		9	& $\sI$ $ (\si{\ampere})$ &	13.79	&	-19.74	&	9.53	&	1.26	&	1.26	\\
		& $\sP$ $ (\si{\watt})$ &	-2645	&	1830	&	4215	&	1549	&	1549	\\
		\hline
		& $\sZinv$ $ (\si{\per\ohm})$ &	0.072	&	0.100	&	0.100	&	0.111	&	0.111	\\
		10	& $\sI$ $ (\si{\ampere})$ &	7.77	&	9.02	&	9.02	&	10.98	&	10.98	\\
		& $\sP $ $ (\si{\watt})$ &	-3538	&	-4143	&	-4143	&	-2795	&	-2795	\\
		\noalign{\hrule height 1.0pt}
	\end{tabular}
\end{table}
\begin{table}[!t]
	\centering
	\renewcommand{\arraystretch}{1.25}
	\caption{Passive Load Values, $\sP$ as in \autoref{tab:App_Sim_Data:Load_params_strict}}
	\label{tab:App_Sim_Data:Load_params_pasive}
	\begin{tabular}{c@{\;}cccccc}
		\noalign{\hrule height 1.0pt}
		Bus & Parameter & $t = \SI{0}{\second}$ & $t = \SI{5}{\second}$ & $t = \SI{10}{\second}$ & $t = \SI{15}{\second}$ & $t = \SI{20}{\second}$ \\
		\noalign{\hrule height 1.0pt}
		1	& $\sZinv$ $ (\si{\per\ohm})$ &	0.091	&	0.093	&	0.087	&	0.087	&	0.063	\\
		& $\sI$ $ (\si{\ampere})$ &	4.66	&	-8.15	&	-6.08	&	-6.08	&	9.71	\\
		\hline
		2	& $\sZinv$ $ (\si{\per\ohm})$ &	0.069	&	0.069	&	0.071	&	0.071	&	0.046	\\
		& $\sI$ $ (\si{\ampere})$ &	-16.09	&	-16.09	&	19.68	&	19.68	&	0.20	\\
		\hline
		3	& $\sZinv$ $ (\si{\per\ohm})$ &	0.095	&	0.082	&	0.082	&	0.082	&	0.059	\\
		& $\sI$ $ (\si{\ampere})$ &	8.91	&	-7.12	&	-7.12	&	-7.12	&	-11.09	\\
		\hline
		4	& $\sZinv$ $ (\si{\per\ohm})$ &	0.038	&	0.038	&	0.038	&	0.038	&	0.038	\\
		& $\sI$ $ (\si{\ampere})$ &	8.82	&	8.82	&	8.82	&	8.82	&	8.82	\\
		\hline
		5	& $\sZinv$ $ (\si{\per\ohm})$ &	0.065	&	0.065	&	0.065	&	0.027	&	0.078	\\
		& $\sI$ $ (\si{\ampere})$ &	-6.64	&	-6.64	&	-6.64	&	15.25	&	2.10	\\
		\hline
		6	& $\sZinv$ $ (\si{\per\ohm})$ &	0.071	&	0.071	&	0.089	&	0.102	&	0.102	\\
		& $\sI$ $ (\si{\ampere})$ &	4.04	&	4.04	&	7.85	&	-9.19	&	-9.19	\\
		\hline
		7	& $\sZinv$ $ (\si{\per\ohm})$ &	0.029	&	0.070	&	0.070	&	0.079	&	0.079	\\
		& $\sI$ $ (\si{\ampere})$ &	9.04	&	0.89	&	0.89	&	0.58	&	0.58	\\
		\hline
		8	& $\sZinv$ $ (\si{\per\ohm})$ &	0.075	&	0.075	&	0.057	&	0.057	&	0.111	\\
		& $\sI$ $ (\si{\ampere})$ &	-16.85	&	-16.85	&	20.55	&	20.55	&	-14.31	\\
		\hline
		9	& $\sZinv$ $ (\si{\per\ohm})$ &	0.105	&	0.102	&	0.061	&	0.036	&	0.036	\\
		& $\sI$ $ (\si{\ampere})$ &	10.71	&	-19.75	&	9.53	&	-0.05	&	-0.05	\\
		\hline
		10	& $\sZinv$ $ (\si{\per\ohm})$ &	0.042	&	0.091	&	0.091	&	0.088	&	0.088	\\
		& $\sI$ $ (\si{\ampere})$ &	2.53	&	2.03	&	2.03	&	8.34	&	8.34	\\
		\noalign{\hrule height 1.0pt}
	\end{tabular}
\end{table}
	

	
	\bibliographystyle{IEEEtran}
	
	\bibliography{ms}
	
	
	
	%
\end{document}